\documentclass[a4paper,10pt,oneside]{article}

\usepackage{tgpagella}
\usepackage[T1]{fontenc}
\usepackage{latexsym}
\usepackage{amssymb}
\usepackage{amsmath}

\usepackage{hyperref}

\usepackage{pgf}
\usepackage{xxcolor} 

\usepackage{natbib}

\usepackage{mathtools,amsthm}
\usepackage{color}
\usepackage{tikz}
\usepackage{graphicx}
\usepackage{ifpdf}
\usetikzlibrary{calc}
\usepackage{fp} 
\usepackage{enumitem}

\bibpunct[, ]{(}{)}{,}{a}{}{;}

\def\qed{\relax
   \ifmmode
    ~\hfill\Box
   \else
    \unskip\nobreak ~\hfill$\square$%
   \fi \par}

\newcommand{\sep}{$\cdot$ }

\theoremstyle{definition} \newtheorem{cor}{Corollary}
\theoremstyle{definition} \newtheorem{conv}{Convention}
\newcommand\defeq{\stackrel{\mathclap{\normalfont\mbox{\tiny def}}}{\equiv}}
\makeatletter
\newcommand{\BIG}{\bBigg@{2}}
\newcommand{\BIGG}{\bBigg@{3}}
\newcommand{\vast}{\bBigg@{4}}
\newcommand{\Vast}{\bBigg@{5}}
\makeatother
\definecolor{axcolor}{rgb}{.3,0,.3}

\newcommand{\bexists}[2]{(\exists #1 \in #2 )} 
\newcommand{\bforall}[2]{(\forall #1 \in #2 )} 
\newcommand{\Bexists}[2]{\big(\exists #1 \in #2 \big)} 
\newcommand{\Bforall}[2]{\big(\forall #1 \in #2 \big)}

\newcommand{\AND}{\land}

\newcommand{\de}{\stackrel{\text{\tiny def}}{=}}
\newcommand{\defiff}{\ \stackrel{\text{\tiny{def}}}{\Longleftrightarrow}\ }

\newcommand{\fblock}[1]{\left[
\begin{array}{c}
#1
\end{array}
\right]}

\usepackage{xspace}

\newcommand{\ax}[1]{{\ensuremath{\mathsf{#1}}}}
\newcommand{\sy}[1]{{\ensuremath{\mathsf{#1}}}}

\newcommand{\AxLine}{\ax{AxLine}\xspace}
\newcommand{\AxEField}{\ax{AxEField}\xspace}

\theoremstyle{definition} \newtheorem{thm}{Theorem}
\theoremstyle{definition} \newtheorem{lem}{Lemma}
\theoremstyle{definition} 
 
\newcommand{\IOb}{\ensuremath{\mathit{IOb}}} 
\newcommand{\B}{\ensuremath{\mathit{B}}} 
\newcommand{\Ph}{\ensuremath{\mathit{Ph}}} 
\newcommand{\Q}{\ensuremath{\mathit{Q}}} 
\newcommand{\W}{\ensuremath{\mathit{W}}} 

\newcommand{\ev}{\ensuremath{\mathit{ev}}} 
\newcommand{\wl}{\ensuremath{\mathit{wl}}} 
\newcommand{\w}{\ensuremath{\mathit{w}}} 

\newcommand{\Triv}{\mathit{Triv}}
\newcommand{\speed}{\mathit{speed}}
\newcommand{\Ether}{\ensuremath{\mathit{Ether}}\xspace}

\newcommand{\vx}{\bar x}
\newcommand{\vy}{\bar y}
\newcommand{\vz}{\bar z}

\begin{document}
	\title{{\scriptsize -- final manuscript --}\\
	COMPARING CLASSICAL AND RELATIVISTIC KINEMATICS IN FIRST-ORDER LOGIC}
	\author{Koen Lefever \and Gergely Sz{\' e}kely}
	\date{}
	\maketitle

\begin{abstract}
	The aim of this paper is to present a new logic-based understanding of the connection between classical kinematics and relativistic kinematics. 

We show that the axioms of special relativity can be interpreted in the language of classical kinematics. This means that there is a logical translation function from the language of special relativity to the language of classical kinematics which translates the axioms of special relativity into consequences of classical kinematics. 

We will also show that if we distinguish a class of observers (representing observers stationary with respect to the ``Ether'') in special relativity and exclude the non-slower-than light observers from classical kinematics by an extra axiom, then the two theories become definitionally equivalent (i.e., they become equivalent theories in the sense as the theory of lattices as algebraic structures is the same as the theory of lattices as partially ordered sets). 

Furthermore, we show that classical kinematics is definitionally equivalent to classical kinematics with only slower-than-light inertial observers, and hence by transitivity of definitional equivalence that special relativity theory extended with ``Ether'' is definitionally equivalent to classical kinematics.

So within an axiomatic framework of mathematical logic, we explicitly show that the transition from classical kinematics to relativistic kinematics is the knowledge acquisition that there is no ``Ether'', accompanied by a redefinition of the concepts of time and space.
	
\end{abstract}

\section{Introduction}

The aim of this paper is to provide a new, deeper and more systematic understanding of the connection of classical kinematics and special relativity beyond the usual ``they agree in the limit for slow speeds''. To compare theories we use techniques, such as logical interpretation and definitional equivalence, from definability theory. Those are usually used to show that theories are equivalent, here we use them to pinpoint the exact differences between both theories by showing how the theories need to be changed to make them equivalent. 

To achieve that, both theories have been axiomatized within \textit{many-sorted first-order logic with equality}, in the spirit of the algebraic logic approach of the Andr{\' e}ka--N{\' e}meti school,\footnote{The epistemological significance of the Andr\'eka--N\'emeti school's research project in general and the the kind of research done in the current paper in particular is being discussed in \citep{Friend15}.} e.g., in \citep{BigBook}, \citep{AMNSamples}, \citep{logst}, \citep{Synthese}, \citep{wnst}, \citep{Judit} and \citep{SzPhd}. 
Our axiom system for special relativity is one of the many slightly different variants of \sy{SpecRel}. The main differences from stock \sy{SpecRel} are firstly that all other versions of \sy{SpecRel} use the lightspeed $c=1$ (which has the advantage of simpler formulas and calculations), while we have chosen to make our results more general by not assuming any units in which to measure the speed of light; and secondly we have chosen to fill the models with all the potential inertial observers by including axioms \ax{AxThExp} and \ax{AxTriv}, which already exists in \citep[p.135]{BigBook} and \citep[p.81]{Judit}, but which is not included in the majority of the axiom systems that can be found in the literature. 

There also already exists \sy{NewtK} as a set of axioms for classical kinematics in \citep[p.426]{BigBook}, but that has an infinite speed of light, which is well-suited to model ``early'' classical kinematics before the discovery of the speed of light in 1676 by O.~R{\o}mer, while we target ``late'' classical kinematics, more specifically in the nineteenth century at the time of J.~C.~Maxwell and the search for the \textit{luminiferous ether}. 

An advantage of using first-order logic is that it enforces us to reveal all the tacit assumptions and formulate explicit formulas with clear and unambiguous meanings. Another one is that it would make it easier to validate our proofs by machine verification, see \citep{Sen}, \citep{ProofVerification} and \citep{StannettNemeti}. For the precise definition of the syntax and semantics of first-order logic, see e.g., \citep[\S 1.3]{CK}, \citep[\S 2.1, \S 2.2]{End}. 

In its spirit relativity theory has always been axiomatic since its birth, as in 1905 A.~Einstein
introduced special relativity by two informal postulates in \citep{Einstein}.  This original informal axiomatization was soon followed by formal ones, starting with \citep{Robb}, many others, for example in \citep{Ax}, \citep{Benda}, \citep{Goldblatt}, \citep{Guts}, \citep{Mundy}, \citep{Reichenbach}, \citep{Schutz}, \citep{Szekeres} and \citep{Winnie}, several of which are still being investigated. For example, the historical axiom system of J.~Ax which uses simple primitive concepts but a lot of axioms to axiomatize special relativity has been proven in \citep{Comparing} to be definitionally equivalent to a variant of the Andr{\' e}ka--N{\' e}meti axioms which use only four axioms but more complex primitive notions.

Our use of techniques such as \textit{logical interpretation} and \textit{definitional equivalence} can be situated among a wider interest and study of these concepts currently going on, for example in \citep{Glymour}, \citep{Morita}, \citep{Laurenz} and \citep{weatherall}. Definitional equivalence has also been called \emph{logical synonymity} or \emph{synonymy}, for example in \citep{deBouvere}, \citep{string} and \citep{Friedman}. The first known use of the method of definitional equivalence is, according to \citep{corcoran}, in \citep{Montague}.

Our approach of using Poincar{\' e}--Einstein synchronisation in classical mechanics was inspired by the ``sound model for relativity'' of \citep{Ax}. We were also inspired by \citep{Szabo} claiming that Einstein's main contribution was redefining the basic concepts of time and space in special relativity.

Let us now formally introduce the concepts \textit{translation}, \textit{interpretation} and  \textit{definitional equivalence} and present our main results:

A \textit{translation} $Tr$ is a function between formulas of many-sorted languages having the same sorts which
\begin{itemize}
\item translates any $n$-ary relation\footnote{In the definition we concentrate only on the translation of relations because functions and constants can be reduced to relations, see e.g., \citep[p.97 \S 10]{bell}.} $R$ into a formula having $n$ free variables of the corresponding sorts: $Tr[R(x_1 \ldots x_n)] \equiv \varphi(x_1 \ldots x_n)$,
\item preserves the equality for every sort, i.e. $Tr(v_i = v_j) \equiv v_i = v_j$,
\item preserves the quantifiers for every sort, i.e. $Tr[(\forall v_i)(\varphi)] \equiv (\forall v_i)[Tr(\varphi)]$ and $Tr[(\exists v_i)(\varphi)] \equiv (\exists v_i)[Tr(\varphi)]$,
\item preserves complex formulas composed by logical connectives, i.e. $Tr(\neg\varphi) \equiv \neg Tr(\varphi)$, $Tr(\varphi \AND \psi) \equiv Tr(\varphi) \AND Tr(\psi)$, etc.
\end{itemize}

By a the translation of a set of formulas \sy{Th}, we mean the set of the translations of all formulas in the set \sy{Th}:
\begin{equation*}
Tr(\sy{Th})\de\{ Tr(\varphi) : \varphi \in \sy{Th} \}.
\end{equation*}

An \textit{interpretation} of theory \sy{Th_1} in theory \sy{Th_2} is a translation $Tr$ which translates all axioms (and hence all theorems) of \sy{Th_1} into theorems of \sy{Th_2}:
\[ (\forall \varphi)[\sy{Th_1} \vdash \varphi \Rightarrow \sy{Th_2} \vdash Tr(\varphi)]. \]

There are several definitions for \textit{definitional equivalence}, see e.g., \citep[p.39-40, \S 4.2]{definability}, \citep[p.469-470]{Glymour} \citep[p.42]{Judit} \citep[pp.60-61]{Ho93},  and \citep[p.42]{TG98}, which are all equivalent if the languages of the theories have disjoint vocabularies. Our definition below is a syntactic version of the semantic definition in \citep[p.56, \S 0.1.6]{HMT85}:

An interpretation $Tr$ of \sy{Th_1} in \sy{Th_2} is a \textit{definitional equivalence} if there is another interpretation $Tr'$ such that the following holds for every formula $\varphi$ and  $\psi$ of the corresponding languages:
\begin{itemize}
\item $\sy{Th_1}\vdash Tr'\big(Tr(\varphi)\big)\leftrightarrow \varphi$
\item $\sy{Th_2}\vdash Tr\big(Tr'(\psi)\big)\leftrightarrow \psi$
\end{itemize}

We denote the definitional equivalence of \sy{Th_1} and \sy{Th_2} by $\sy{Th_1} \equiv_{\Delta} \sy{Th_2}$.

\begin{thm}\label{thm-eqrel} Definitional equivalence is an equivalence relation, i.e. it is reflexive, symmetric and transitive.
\end{thm}
For a proof of this theorem, see e.g., \citep[p.7]{diss}.

In this paper, we introduce axiom systems \sy{ClassicalKin_{Full}} for classical kinematics, \sy{SpecRel_{Full}} for special relativity and their variants based on the framework and axiom system of \citep{BigBook}, \citep{AMNSamples}, \citep{logst}, \citep{Synthese}, \citep{wnst} and \citep{Judit}.  Then we construct logical interpretations between these theories translating the axioms of one system into theorems of the other. In more detail, we show the following connections:

Special relativity can be interpreted in classical kinematics, i.e., there is a translation $Tr$ that translates the axioms of special relativity into theorems of classical kinematics:
\begin{itemize}
\item $\sy{ClassicalKin_{Full}} \vdash Tr(\sy{SpecRel_{Full}}).$\hfill [see Theorem \ref{thm-tr} on p.\pageref{thm-tr}]
\end{itemize}
Special relativity extended with a concept of ether, \sy{SpecRel^e_{Full}}, and classical kinematics restricted to slower-than-light  observers, \sy{ClassicalKin^{STL}_{Full}}, can be interpreted in each other: 
\begin{itemize}
\item $\sy{ClassicalKin^{STL}_{Full}} \vdash Tr_+ (\sy{SpecRel^e_{Full}}),$ \hfill [see Theorem \ref{thm-tr+} on p.\pageref{thm-tr+}]
\item $\sy{SpecRel^e_{Full}} \vdash Tr'_+ (\sy{ClassicalKin^{STL}_{Full}}).$ \hfill [see Theorem \ref{thm-tr'+} on p.\pageref{thm-tr'+}]
\end{itemize}
Moreover, these axiom systems are definitionally equivalent ones:
\begin{itemize}
\item $\sy{ClassicalKin^{STL}_{Full}} \equiv_{\Delta} \sy{SpecRel^e_{Full}}.$ \hfill  [see Theorem \ref{thm-defeq} on p.\pageref{thm-defeq}]
\end{itemize}
Furthermore, we establish the definitional equivalence between $\sy{ClassicalKin^{STL}_{Full}}$ and $\sy{ClassicalKin_{Full}}$:
\begin{itemize}
\item $\sy{ClassicalKin_{Full}} \equiv_{\Delta} \sy{ClassicalKin^{STL}_{Full}},$ \hfill  [see Theorem \ref{thm-defeqFTL} on p.\pageref{thm-defeqFTL}]
\end{itemize}
from which follows, by transitivity of definitional equivalence, that classical kinematics is definitionally equivalent to special relativity extended with ether:
\begin{itemize}
\item $\sy{ClassicalKin_{Full}} \equiv_{\Delta} \sy{SpecRel^e_{Full}}$, \hfill  [see Corollary \ref{cor-strong} on p.\pageref{cor-strong}]
\end{itemize}
which is the main result of this paper.

\begin{figure}[ht]
  \begin{center}
    \scalebox{.8}{\usetikzlibrary{calc}
\begin{tikzpicture}[scale=0.6, every node/.style={scale=0.8}]

\pgfmathsetmacro\r{0.07}
\pgfmathsetmacro\rSRe{2}
\pgfmathsetmacro\rSR{1}
\pgfmathsetmacro\rNK{3}
\pgfmathsetmacro\x{10}
\pgfmathsetmacro\y{-6}

\tikzstyle{SR}=[blue,ultra thick]
\tikzstyle{NK}=[red,ultra thick]
\tikzstyle{SRe}=[green!80!black,ultra thick]

\coordinate (O) at (0,0);
\coordinate (S) at (0,\y);
\coordinate (N) at (\x-1.5,0);
\coordinate (E) at (0.5,0.7*\rSRe);
\coordinate (Eth) at (\x-2.4,0.7*\rSRe);
\coordinate (FTL) at (\x-3,0.79*\rNK);
\coordinate (N2) at (\x+7,0);
\coordinate (E2) at (0.5+7,0.7*\rSRe);
\coordinate (Eth2) at (\x-0.6+7,0.7*\rSRe);
\coordinate (FTL2) at (\x-1.4+7,0.79*\rNK);

\draw[very thick, SR] (S) circle (\rSR);
\node[below] at ([shift={(1,-\rSR)}]S){$\mathsf{SpecRel_{Full}}$};

\draw[SR] (O) circle (\rSR);
\draw[SRe] (O) circle (\rSRe);
\node[above] at (1,\rSRe){$\mathsf{SpecRel^e_{Full}}$};
\fill (E) circle (\r) node[left]  {$E$};

\fill (Eth) circle (\r) node[right]{Ether};
\draw[NK] (N) circle (\rNK);
\draw[SR] (N) circle (\rSR);

\node[below] at ([shift={(0,-\rNK)}]N){$\mathsf{ClassicalKin_{Full}}$};
\fill (FTL) circle (\r) node[right, xshift=-1.3]{FTL-IOb};
\draw[SRe] (N) circle (\rSRe);

\draw[thick, ->,>=latex, shorten <=5,shorten >=5] (E) to (Eth);
\draw[very thick,->,>=latex, shorten <=0.9*\rNK cm,shorten >=0.9*\rNK cm] (S) to node[above left, shift={(0.1,-0.1)}] {$Tr$} (N);
\draw[very thick,->,>=latex, shorten <=1.2*\rSRe cm,shorten >=1.2*\rSRe cm] (-1,0) to node[above, yshift=-1] {$Tr_+$} node[below] {$\equiv_{\Delta}$}(\x-1,0);
\draw[very thick,->,>=latex, shorten <=1.5*\rSR cm,shorten >=1.5*\rSR cm] (S) to node[left] {$Id$} (O);

\draw[SR,thin] ([xshift=\rSR cm]S) to ([xshift=\rSR cm]O);
\draw[SR,thin] ([xshift=-\rSR cm]S) to ([xshift=-\rSR cm]O);
\draw[SR,thin] ([yshift=\rSR cm]N) to ([yshift=\rSR cm]O);
\draw[SR,thin] ([yshift=-\rSR cm]N) to ([yshift=-\rSR cm]O);
\draw[SR,thin] ([shift={(-0.45*\rSR,0.9*\rSR)}]N) to ([shift={(-0.45*\rSR,0.9*\rSR)}]S);
\draw[SR,thin] ([shift={(0.45*\rSR,-0.9*\rSR)}]N) to ([shift={(0.45*\rSR,-0.9*\rSR)}]S);

\draw[SRe,thin] ([yshift=\rSRe cm]N) to ([yshift=\rSRe cm]O);
\draw[SRe,thin] ([yshift=-\rSRe cm]N) to ([yshift=-\rSRe cm]O);

\node at (9,-7) {$Tr_{+}|_{\mathsf{SpecRel_{Full}}} = Tr $};
\node at (9,-6) {$Tr_{+} \big(E(x) \big) = Ether(x) $};

\draw[very thick,<-,>=latex, shorten <=1.5*\rSRe cm,shorten >=1.5*\rSRe cm] (\x-12,-0.8) to node[below, yshift=-0] {$Tr_+'$}(\x,-0.8);

\fill (Eth2) circle (\r) node[right]{Ether};
\draw[SRe] (N2) circle (\rNK);
\draw[SRe,red,dashed] (N2) circle (\rNK);
\node[below] at ([shift={(0,-\rNK)}]N2){$\mathsf{ClassicalKin_{Full}}$};
\fill (FTL2) circle (\r) node[right, xshift=-1.3]{FTL-IOb};
\draw[SRe,thin] ([yshift=\rSRe+3 cm]N2) to ([yshift=\rSRe cm]N);
\draw[SRe,thin] ([yshift=-\rSRe-3 cm]N2) to ([yshift=-\rSRe cm]N);
\draw[very thick,->,>=latex, shorten <=1.05*\rSRe cm,shorten >=1.05*\rSRe cm] (N) to node[above, yshift=-1] {$Tr_*$} node[below] {$\equiv_{\Delta}$}(N2);
\draw[very thick,<-,>=latex, shorten <=1.05*\rSRe cm,shorten >=1.05*\rSRe cm] (\x-1.7,-0.8) to node[below, yshift=-0] {$Tr_*'$}(\x+7,-0.8);
\draw[thick, ->,>=latex, shorten <=30,shorten >=5] (Eth) to (Eth2);
\node[below] at ([shift={(2.9,\rNK-0.7)}]N){$\mathsf{ClassicalKin_{Full}^{STL}}$};

\end{tikzpicture}}
    \caption{\label{fig-trans} Translations: $Tr$ translates from special relativity to classical kinematics. $Tr_+$ and $Tr'_+$ translate between special relativity extended with a primitive ether $E$ and classical kinematics without faster-than-light observers, which are definitionally equivalent theories, see Theorem~\ref{thm-defeq}. $Tr_*$ and $Tr'_*$ translate between classical kinematics without faster-than-light observers and classical kinematics, which are definitionally equivalent theories, see Theorem~\ref{thm-defeqFTL}.}
  \end{center}
\end{figure}
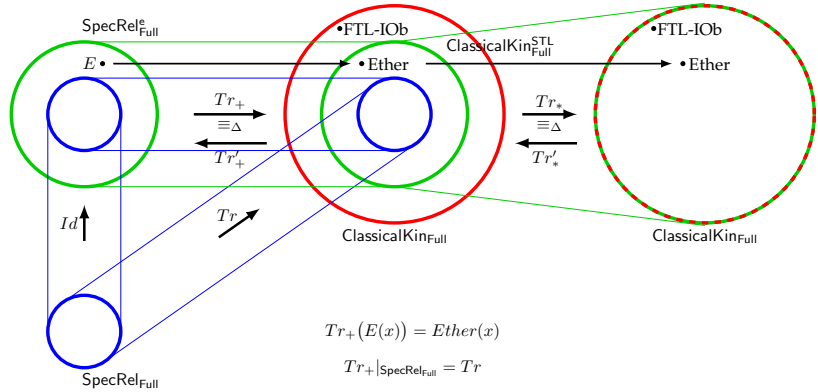

\section{The language of our theories}
\label{sec-lang}

We will work in the axiomatic framework of \citep{Synthese}.
Therefore, there will be two sorts of basic objects: \emph{bodies} $\B$ (thing that can move) and \emph{quantities} $\Q$ (numbers used by observers to describe motion via coordinate systems). We will distinguish two kinds of bodies: \emph{inertial observers} and \emph{light signals} by one-place relation symbols $ IOb$ and $\Ph$ of sort $\B$. We will use the usual algebraic operations and ordering ($+$,  $\cdot$ and $\le$) on sort $\Q$. Finally, we will formulate coordinatization by using a $6$-place \emph{worldview relation} $\W$ of sort $\B^2\times\Q^4$. 

\noindent
That is, we will use the following two-sorted first-order logic with equality:
\begin{equation*}
\{\, \B,\Q\,;  IOb, \Ph,+,\cdot,\le,\W\,\}.
\end{equation*}

Relations $ IOb(k)$ and $\Ph(p)$ are translated as ``\textit{$k$ is an inertial observer},'' and ``\textit{$p$ is a light signal},'' respectively. $\W(k,b,x_0,x_1,x_2,x_3)$ is translated as ``\textit{body $k$ coordinatizes body $b$ at space-time location $\langle x_0,x_1,x_2,x_3\rangle$},'' (i.e., at spatial location $\langle x_1,x_2,x_3\rangle$ and instant $x_0$).  

Since we have two sorts (quantities and bodies), we have also two kinds of variables, two kinds of terms, two equation signs and two kinds of quantifiers (one for each corresponding sort). Quantity variables are usually denoted by $x$, $y$, $z$, $t$, $v$, $c$ (and their indexed versions), body variables are usually denoted by $b$, $k$, $h$, $e$, $p$ (and their indexed versions).  Since we have no function symbols of sort $\B$, \emph{body terms} are just the body variables. \emph{Quantity terms} are what can be built from quantity variables using the two functions symbols $+$ and $\cdot$ of sort $\Q$. We denote quantity terms by $\alpha$, $\beta$, $\gamma$ (and their indexed versions). For convenience, we use the same sign ($=$) for both sorts because from the context it is always clear whether we mean equation between quantities or bodies.

The so called \emph{atomic formulas} of our language are $\W(k,b,\alpha_0,\alpha_1,\alpha_2,\alpha_3)$, $ IOb(k)$, $\Ph(p)$, $\alpha=\beta$,  $\alpha\le \beta$  and $k=b$ where $k$, $p$, $b$,  $\alpha$, $\beta$, $\alpha_0$, $\alpha_1$, $\alpha_2$, $\alpha_3$ are arbitrary terms of the corresponding sorts. 

The \emph{formulas} are built up from these atomic formulas by using the logical connectives \textit{not} ($\lnot$), \textit{and} ($\AND$), \textit{or} ($\lor$), \textit{implies} ($\to$), \textit{if-and-only-if} ($\leftrightarrow$) and the quantifiers \textit{exists} ($\exists$) and \textit{for all} ($\forall$). 
In long expressions, we will denote the logical \textit{and} by writing formulas below each other between rectangular brackets:
\[\fblock{\varphi \\ \psi} \text{ is a notation for } \varphi \AND \psi.\]
To distinguish formulas \textit{in} our language from formulas \textit{about} our language, we use in the meta-language the symbols $\Leftrightarrow$ (as illustrated in the definition of bounded quantifiers below) and $\equiv$ (while translating formulas between languages) for the logical equivalence. We use the symbol $\vdash$ for syntactic consequence.

We use the notation $\Q^n$ for the set of $n$-tuples of $\Q$. If $\vx\in \Q^n$, we assume that $\vx=\langle x_1,\ldots,x_n\rangle$, i.e., $x_i$ denotes the $i$-th component of the $n$-tuple $\vx$. We also write $\W(k,b,\vx)$ in place of $\W(k,b,x_0,x_1,x_2,x_3)$, etc. 

We will treat unary relations as sets. If $R$ is a unary relation, then we use bounded quantifiers in the following way:
\[
\bforall {u}{R}[\varphi] \defiff \forall u [R(u) \to \varphi]\  \text{ and }\ 
\bexists {u}{R}[\varphi] \defiff \exists u [R(u) \AND \varphi]. 
\]
We will also use bounded quantifiers to make it explicit which the sort of the variable is, such as $\exists x \in Q$ and $\forall b \in B$, to make our formulas easier to comprehend.

Worldlines and events can be easily expressed by the worldview relation $W$ as follows. The \emph{worldline} of body $b$ according to observer $k$ is the set of coordinate points where $k$ have coordinatized $b$:
\begin{equation*}
\vx \in \wl_k(b)\defiff \W(k,b,\vx).
\end{equation*}

The \emph{event} occurring for observer $k$ at coordinate point $\vx$ is the set of bodies $k$ observes at $\vx$:
\begin{equation*}
b \in \ev_k(\vx)\defiff \W(k,b,\vx).
\end{equation*}

We will use a couple of shorthand notations to discuss spatial distance,  speed, etc.\footnote{Since in our language we only have addition and multiplication, we need some basic assumptions on the properties of these operators on numbers ensuring the definability of subtraction, division, and square roots. These properties will follow from the Euclidian field axiom (\AxEField, below on page \pageref{PageAxEField}). Also, the definition of speed is based on the axiom saying that inertial observers move along straight lines relative to each other (\AxLine, below on page \pageref{PageAxLine}).} 

\noindent
\emph{Spatial distance} of $\vx,\vy\in\Q^4$:
\begin{equation*}
   space(\vx,\vy)\de \sqrt{(x_1-y_1)^2+(x_2-y_2)^2+(x_3-y_3)^2}.
\end{equation*}
\emph{Time difference} of coordinate points $\vx,\vy\in\Q^4$:
\begin{equation*}
 time(\vx,\vy)\de \vert x_0-y_0 \vert. 
\end{equation*}
The \emph{speed} of body $b$ according to observer $k$ is defined as:
\begin{multline*}
\speed_{k}(b) = v \defiff 
\Bexists{\vx,\vy}{\wl_k(b)}( \vx \neq \vy ) \AND \\ 
\Bforall{\vx,\vy}{\wl_k(b)} 
\left[ space(\vx,\vy) = v \cdot time(\vx,\vy) \right].
\end{multline*}
The \emph{velocity} of body $b$ according to observer $k$ is defined as:
\begin{multline*}
\bar{v}_{k}(b) = \bar{v} \defiff  
\Bexists{\vx,\vy}{\wl_k(b)}( \vx \neq \vy ) \AND \\
\Bforall{\vx,\vy}{\wl_k(b)} 
\left[ (y_1-x_1,y_2-x_2,y_3-x_3) = \bar{v} \cdot (y_0 - x_0)\right].
\end{multline*}

Relations $speed$ and $\bar v$ are partial functions from $B\times B$ respectively to $Q$ and $Q^3$  which are defined if $\wl_k(b)$ is a subset of the non-horizontal line which contains at least two points.

Let us define the \emph{worldview transformation}\footnote{While the worldview transformation $w$ is here only defined as a binary relation, our axioms will turn it into a transformation for inertial observers, see Theorems \ref{thm-gal} and \ref{thm-poi} below.} between observers $k$ and $k'$ as the following binary relation on $\Q^4$:
\begin{equation*}
\w_{kk'}(\vx,\vy)\defiff \ev_k(\bar x)=\ev_{k'}(\vy).
\end{equation*}

\begin{conv}\label{partial-conv}
We use partial functions in our formulas as special relations, which means that when we write $f(x)$ in a formula we also assume that $f(x)$ is defined, i.e. $x$ is in the domain of $f$. This is only a notational convention that makes our formulas more readable and can be systematically eliminated, see \citep[p.61, Convention 2.3.10]{BigBook} for further discussion.
\end{conv}

\noindent
The \textit{models} of this language are of the form
\begin{equation*}
\mathfrak{M} = \langle \B_\mathfrak{M}, \Q_\mathfrak{M}; \IOb_\mathfrak{M},\Ph_\mathfrak{M},+_\mathfrak{M},\cdot_\mathfrak{M},\le_\mathfrak{M},\W_\mathfrak{M}\rangle,
\end{equation*}
where $\B_\mathfrak{M}$ and $\Q_\mathfrak{M}$ are nonempty sets, 
$\IOb_\mathfrak{M}$ and $\Ph_\mathfrak{M}$ are unary relations on
$\B_\mathfrak{M}$, $+_\mathfrak{M}$ and $\cdot_\mathfrak{M}$ are binary functions
and $\le_\mathfrak{M}$ is a binary relation on $\Q_\mathfrak{M}$, and
$\W_\mathfrak{M}$ is a relation on $\B_\mathfrak{M}\times \B_\mathfrak{M}\times
\Q^{d}_\mathfrak{M}$.
The subscript $_\mathfrak{M}$ for the sets and relations indicates that those are set-theoretical objects rather than the symbols of a formal language.

\section{Axioms}
\subsection{Axioms for the common part}

For the structure $\langle \Q,+,\cdot,\le \rangle$ of quantities, we assume some basic algebraic properties of addition, multiplication and ordering true for real numbers. 
\begin{description}
\item[\underline{\AxEField}]
 $\langle \Q,+,\cdot,\le \rangle$ is a Euclidean field. That is, $\langle\Q,+,\cdot\rangle$ is a field in the sense of algebra; 
$\le$ is a linear ordering on $\Q$ such that  
$x \le y\to x + z \le y + z$, and
$(0 \le x \AND 0 \le y) \to 0 \le xy$; and
every positive number has a square root.
\end{description}

Some notable examples of Euclidean fields are the real numbers, the real algebraic numbers, the hyperreal numbers and the real constructable numbers\footnote{It is an open question if the rational numbers would be sufficient, which would allow to replace \ax{AxEField} by the weaker axiom that the quantities only have to be an ordered field, and hence have a stronger result since we would be assuming less. See \citep{MSzRac} for a possible approach.}.

The rest of our axioms will speak about how inertial observers coordinatize the events. Naturally, we assume that they coordinatize the same set of events.   

\begin{description}
\item[\underline{\ax{AxEv}}] \label{PageAxEField}
All inertial observers coordinatize the same events:
\begin{equation*}
\bforall {k,h}{ IOb} 
\Bforall {\vx}{\Q^4}
\Bexists {\vy}{\Q^4} 
[\ev_k(\vx)=\ev_{h}(\vy)].
\end{equation*}
\end{description}

We assume that inertial observers move along straight lines with respect to each other.
\begin{description}
\item[\underline{\AxLine}] The worldline of an inertial observer is a straight line according to inertial observers: \label{PageAxLine}
\begin{multline*}
\bforall {k,h}{ IOb}
\Bforall {\vx,\vy,\vz}{\wl_k(h)}\\
\bexists {a} {\Q}  \big[\vz-\vx=a(\vy-\vx) \lor
\vy-\vz=a(\vz-\vx)\big].
\end{multline*}
\end{description}

As usual we speak about the motion of reference frames by using their time-axes. Therefore, we assume the following. 

\begin{description}
\item[\underline{\ax{AxSelf}}] Any inertial observer is stationary in his own coordinate system:
\begin{equation*}
\bforall {k}  {IOb}
\bforall {t,x,y,z}{\Q}
\big[\W(k,k,t,x,y,z) \leftrightarrow x=y=z=0\big].
\end{equation*}
\end{description}

The following axiom is a symmetry axiom saying that observers (can) use the
same units to measure spatial distances.

\begin{description}
\item[\underline{\ax{AxSymD}}]
Any two inertial observers agree as to the spatial distance between two events if these two events are simultaneous for both of them:
\begin{multline*}
\bforall {k,k'}{ IOb}
\Bforall {\vx,\vy,\vx',\vy'}{\Q^4}\\
\left(\fblock{  time(\vx,\vy)= time(\vx',\vy')=0\\ \ev_k(\vx)=\ev_{k'}(\vx') \\
\ev_k(\vy)=\ev_{k'}(\vy')}\to  space(\vx,\vy)= space(\vx',\vy')\right).
\end{multline*}
\end{description}

When we choose an inertial observer to represent an inertial frame\footnote{We use the word ``frame'' here in its intuitive meaning, as in \citep[p.40]{Rind}. For a formal definition of a \textit{frame} we need the concept of \textit{trivial transformation}, see page \pageref{frame}.} of reference, then the origin of that observer can be chosen anywhere, as well as the orthonormal basis they use to coordinatize space. To introduce an axiom capturing this idea, let $\Triv$ be the set of \emph{trivial transformations}, by which we simply mean transformations that are isometries on space and translations along the time axis. For more details, see \citep[p.81]{Judit}.
\begin{description}
\item[\underline{\ax{AxTriv}}]
Any trivial transformation of an inertial coordinate system is also an inertial coordinate system:
\begin{equation*}
\bforall{T}{\Triv} 
\bforall {k}  {IOb} 
\bexists {k'}  {IOb}
[\w_{kk'}=T].
\footnote{$\bforall{T}{\Triv}$ may appear to the reader to be in second-order logic. However, since a trivial transformation is nothing but an isometry on space ($4 \times 4$ parameters) and a translation along the time axis ($4$ parameters), this is just an abbreviation for $\bforall{q_1, q_2, \ldots q_{20}}{Q}$ and together with $\w_{kk'}=T$ a system of equations with 20 parameters in first-order logic.}
\end{equation*}
\end{description}

Axioms \ax{AxTriv}, \ax{AxThExp_+} (on page \pageref{axthexp+}), and \ax{AxThExp} (on page \pageref{axthexp}) make spacetime full of inertial observers.\footnote{These  inertial observers are only potential and not actual observers, in the same way as the light signals required in every coordinate points by axioms \ax{AxPh_c} and \ax{AxEther} below are only potential light signals. For a discussion on how actual and potential bodies can be distinguished using modal logic, see \citep{Attila}.} We will use the subscript $\mathsf{Full}$ to denote that these axioms are part of our axiom systems.

A set of all observers which are at rest relative to each other, which is a set of all observers which are related to each other by a trivial transformation, we call a \textit{frame}. \label{frame}
Because of \ax{AxTriv}, a frame contains an infinite number of elements. From now on, we may informally abbreviate ``the speed/velocity/movement relative to all observers which are elements of a frame'' to ``the speed/velocity/movement relative to a frame''.

\noindent 
Let us define an \textit{observer} as any body which can coordinatize other bodies:
\[Ob(k) \defiff \bexists{b}{B}\Bexists{\vx}{Q^4}\big(W(k,b,\vx)\big).\]

Since we will be translating back and forth, we need a guarantee that all observers have a translation\footnote{We do not need this axiom for the interpretation, but we do need it for the definitional equivalence. In \citep{diss}, we postpone the introduction of this axiom until the chapter on definitional equivalence, resulting in a slightly stronger theorem \ref{thm-tr}.}. 
\begin{description}
\item[\underline{\ax{AxNoAcc}}]
All observers are inertial observers:
\[\bforall{k}{B}[W(Ob(k)\rightarrow IOb(k)].\]
\end{description}

The axioms above will be part of all the axiom systems that we are going to use in this paper. Let us call their collection \sy{Kin_{Full}}:
\[\sy{\mathsf{Kin_{Full}}}\de \{\AxEField, \ax{AxEv}, \ax{AxSelf}, \ax{AxSymD}, \AxLine ,\ax{AxTriv}, \ax{AxNoAcc} \}.\]

\subsection{Axioms for classical kinematics}
\label{ax-nk}

A key assumption of classical kinematics is that the time difference between two events is observer independent.
\begin{description}
\item[\underline{\ax{AxAbsTime}}] The time difference between any two events is the same for all inertial observers:
\begin{multline*}
\bforall {k,k'}{ IOb}
\Bforall {\vx,\vy,\vx',\vy'}{\Q^4} \\
\left(\fblock{\ev_k(\vx)=\ev_{k'}(\vx')\\\ev_k(\vy)=\ev_{k'}(\vy')}\to
 time(\vx,\vy)= time(\vx',\vy') \right).
\end{multline*}
\end{description}

We also assume that inertial observers can move with arbitrary (finite) speed in any direction everywhere.

\begin{description} \label{axthexp+}
\item[\underline{\ax{AxThExp_+}}]
Inertial observers can move along any non-horizontal straight line\footnote{The first part of this axiom (before the conjumction) is not necessary since we will asume the axiom \ax{AxEther} which guarantees that we have at least one inertial observer, see page \pageref{PageAxEther}.}:
\begin{multline*}
\bexists {h}{B} \big[ IOb(h)\big] \AND\\ 
\bforall {k}{ IOb} 
\Bforall {\vx,\vy}{\Q^4} 
\left(x_0\neq y_0 \to \Bexists {k'}{ IOb}\big[\vx,\vy\in \wl_k(k')\big]\right).
\end{multline*}
\end{description}

The motion of light signals in classical kinematics is captured by assuming that there is at least one inertial observer according to which the speed of light is the same in every direction everywhere.  Inertial observers with this property will be called \emph{ether observers} and the unary relation \Ether appointing them is defined as follows:  
\begin{multline*}
  \Ether(e) \defiff 
 IOb(e)\AND \bexists {c} {\Q} \Big[c>0\AND \Bforall {\vx,\vy}{\Q^4} \\ \Big(\bexists {p} {\Ph} \big[\vx,\vy\in\wl_{e}(p) \big] \leftrightarrow  space(\vx,\vy)=c\cdot time(\vx,\vy)\Big)\Big].
\end{multline*}

\begin{description}
\item[\underline{\ax{AxEther}}]\label{PageAxEther}
  There exists at least one ether observer:
  \begin{equation*}
    \bexists {e}{B} \big[\Ether(e)\big].
  \end{equation*}
\end{description}
Let us introduce the following axiom system for classical kinematics:
\[\sy{\mathsf{ClassicalKin}_{Full}}\de \sy{\mathsf{Kin_{Full}}} \cup \{ \ax{AbsTime}, \ax{AxThExp_+}, \ax{AxEther} \}.\]

The map $G:\Q^4\to\Q^4$ is called a \emph{Galilean transformation} iff it is an affine bijection having the following properties:
\begin{equation*}
|time(\vx,\vy)| = |time(\vx',\vy')|,
\text{ and }
\end{equation*}
\begin{equation*}
x_0=y_0\to x'_0=y'_0\AND space(\vx,\vy)=space(\vx',\vy')
\end{equation*}
for all $\vx,\vy,\vx',\vy'\in\Q^4$ for which  $G(\vx)=\vx'$ and $G(\vy)=\vy'$.

In \citep[p.20]{diss} we prove the \textit{justification theorem}, to establish that the above is indeed an axiomatization of classical kinematics:

\begin{thm}\label{thm-gal} Assume \sy{ClassicalKin_{Full}}. Then
$\w_{mk}$ is a Galilean Transformation for all inertial observers $m$ and $k$. 
\end{thm}

Theorem~\ref{thm-gal} shows that \sy{ClassicalKin_{Full}} captures classical kinematics since it implies that the worldview transformations between inertial observers are the same as in the standard non-axiomatic approaches. There is a similar theorem as Theorem~\ref{thm-gal} for \sy{NewtK}, a version of classical kinematics with $c=\infty$, in \citep[p.439, Proposition 4.1.12 Item 3]{BigBook}.

\begin{cor}\label{thm-lightspeed}
Assuming $\sy{\mathsf{ClassicalKin}_{Full}}$, all ether observers are stationary with respect to each other, and hence they agree on the speed of light.
\end{cor}

Since $Ether$ is an unary relation, we can also treat it as a set. So, by Corollary~\ref{thm-lightspeed}, $Ether$ as a set is the \textit{ether frame} and its elements (usually denoted by $e_1$, $e_2$, $e_3$, \ldots or $e$, $e'$, $e''$, \ldots) are the \textit{ether observers}. The way in which we distinguish \textit{frames} from \textit{observers} is inspired by W. Rindler in \citep[p.40]{Rind}. 

By Corollary~\ref{thm-lightspeed}, we can speak about the \emph{ether-observer-independent speed of light}, denoted by $\mathfrak{c_e}$ which is the unique quantity satisfying the following formula:
\[\bforall{e}{Ether}\bforall{p}{Ph}[speed_e(p)=\mathfrak{c_e}].\]

\begin{cor}\label{cor-s} Assuming \sy{\mathsf{ClassicalKin}_{Full}}, the speed of any inertial observer is the same according to all ether observers:
\begin{equation*}
\bforall{e, e'}{\Ether}\bforall{k}{ IOb}[\speed_{e}(k) = \speed_{e'}(k)].
\end{equation*}
\end{cor}
\begin{cor}\label{cor-v} Assuming \sy{\mathsf{ClassicalKin}_{Full}}, all ether observers have the same velocity according to any inertial observer:
\begin{equation*}
\bforall{e, e'}{\Ether}\bforall{k}{ IOb}[\bar{v}_k(e) = \bar{v}_k(e')].
\end{equation*}
\end{cor}

\subsection{Axioms for special relativity}
\label{sec-sr}

A possible key assumption of special relativity is that the speed of light signals is observer independent. 

\begin{description}
\item[\underline{\ax{AxPh_c}}] For any inertial observer, the speed of light is the same everywhere and in every direction. Furthermore, it is possible to send out a light signal in any direction everywhere:
\begin{multline*}
\bexists {c} {\Q} 
\Big[ c>0 \AND 
  \bforall {k}  {IOb}
  \Bforall {\vx,\vy} {\Q^4}\\
  \big(\bexists {p}{\Ph} 
  \big[\vx,\vy\in \wl_k(p)\big] \leftrightarrow  space(\vx,\vy)= c\cdot time(\vx,\vy)\big)\Big] .
\end{multline*}
\end{description}

By \ax{AxPh_c}, we have an observer-independent speed of light. From now on, we will denote this speed of light as $\mathfrak{c}$. From $\ax{AxPh_c}$, it follows that observers (as considered by the theory) use units of measurement which have the same numerical value for the speed of light. The value of the constant speed of light depends on the choice of units (for example $\mathfrak{c} = 299 792 458$ when using meters and seconds or $\mathfrak{c} = 1$ when using light-years and years as units). We prove below, in Lemma \ref{lemma-trc} and Corollary \ref{cor-Tr_+'(c_e)}, that the relativistic speed of light $\mathfrak{c}$ and the ether-observer-independent speed of light $\mathfrak{c_e}$ translate into each other. Note that $c$ in \ax{AxPh} and \ax{AxEther} is a variable, while $\mathfrak{c}$ and $\mathfrak{c_e}$ are model dependent constants. 

\sy{Kin_{Full}} and \ax{AxPh_c} imply that no inertial observer can move faster than
light if $d \ge 3$, see e.g., \citep{Synthese}. Therefore, we will use the following version of \ax{AxThExp_+}.

\begin{description} \label{axthexp}
\item[\underline{\ax{AxThExp}}] Inertial observers can move along any  straight line of any speed less than the speed of light:
\begin{multline*}
\bexists {h}{B}  \big[ IOb(h)\big]\AND 
\bforall {k}  {IOb} 
\Bforall {\vx,\vy}{\Q^4} \\
\big( space(\vx,\vy)<\mathfrak{c}\cdot time(\vx,\vy) \to \bexists {k'}  {IOb}  \big[\vx,\vy\in \wl_k(k')\big]\big).
\end{multline*}
\end{description}
Let us introduce the following axiom system for special relativity:
\[\sy{\mathsf{SpecRel}_{Full}}\de\ \sy{\mathsf{Kin_{Full}}} \cup \{\ax{AxPh_c},\ax{AxThExp} \}.\]

The map $P:\Q^4\to\Q^4$ is called a \emph{Poincar\'e transformation} corresponding to light speed $c$ iff it is an affine bijection having the following property
\begin{equation*}
 c^2\cdot time(\vx,\vy)^2- space(\vx,\vy)^2= c^2\cdot time(\vx',\vy')^2- space(\vx',\vy')^2
\end{equation*}
for all $\vx,\vy,\vx',\vy'\in\Q^4$ for which $P(\vx)=\vx'$ and
$P(\vy)=\vy'$.

In \citep[p.24]{diss} based on theorem \citep[Thm 2.1, p.639]{Synthese} for stock \sy{SpecRel}, we prove\footnote{An alternative proof for Theorem~\ref{thm-poi} would be by using the Alexandrov--Zeeman Theorem, which states that any causal automorphism of spacetime is a Lorentz transformation up to a dilation, a translation and a field-automorphism-induced collineation, for the case when the field is the field of real numbers. It was independently discovered by A.~D.~Aleksandrov in 1949, L.-K.~Hua in the 1950s and E.~C.~Zeeman in 1964, see \citep[p.179]{Goldblatt}. See \citep{Zeeman} for E.~C.~Zeeman's proof of the theorem, \citep{vro} and \citep{VKK} for algebraic generalizations, and \citep{Pambuccian} for a proof using definability theory.} the following justification theorem  which establishes that \sy{SpecRel_{Full}} is indeed an axiomatization of special relativity:

\begin{thm}\label{thm-poi} Assume \sy{SpecRel_{Full}}. Then 
$\w_{mk}$ is a Poincar{\'e} transformation corresponding to $\mathfrak{c}$ for all inertial observers $m$ and $k$. 
\end{thm}

Theorem~\ref{thm-poi} shows that \sy{SpecRel_{Full}} captures the kinematics of special relativity since it implies that the worldview transformations between inertial observers are the same as in the standard non-axiomatic approaches.
Note that the Poincar{\'e} transformations in Theorem~\ref{thm-poi} are model-dependent. When we talk about Poincar{\'e} transformations below, we mean Poincar{\'e} transformations corresponding to the speed of light of the investigated model.

\begin{cor}\label{cor-max} Assuming \sy{\mathsf{SpecRel_{Full}}}, the speed of any inertial observer relative to any other inertial observer is slower than the speed of light:
\begin{equation*}
\bforall{h,k}{IOb}[\speed_{h}(k) < \mathfrak{c}].
\end{equation*}
\end{cor}

\section[Using Poincar\' e--Einstein synchronisation to construct relativistic coordinate systems]{Using Poincar\' e--Einstein synchronisation to construct relativistic coordinate systems for classical observers}
\label{sec:rad}

In this section, we are going to give a systematic translation of the formulas of \sy{SpecRel_{Full}} to the language of \sy{ClassicalKin_{Full}} such that the translation of every consequence  of \sy{SpecRel_{Full}} will follow from \sy{ClassicalKin_{Full}}, see Theorem~\ref{thm-tr}.

The basic idea is that if classical observers use light signals and Poincar\' e--Einstein synchronisation, then the coordinate systems of the slower-than-light observers after a natural time adjustment will satisfy the axioms of special relativity. 
Hence we will use light signals to determine simultaneity and measure distance in classical physics in the same way as in relativity theory. 
For every classical observer $k$, we will redefine the coordinates of events using Poincar\' e--Einstein synchronization. For convenience, we will work in the ether frame because there the speed of light is the same in every direction. 

Let us consider the ether coordinate system $e$  which agrees with $k$ in every aspect: it intersects the worldline of observer $k$ at the origin according to both $k$ and $e$; it agrees with $k$ on the direction of time and the directions of space axes; it agrees with $k$ in the units of time and space.

First we consider the case where observer $k$ is moving in the $x$ direction according to $e$. We will discuss the general case when discussing Figure~\ref{fig-radar4D} below.
So let us first understand what happens in the $tx$-plane when classical observer $k$ uses Poincar\' e--Einstein synchronization to determine the coordinates of events. 

Let $(t_1,x_1)$ be an arbitrary point of the $tx$-plane in the coordinate system of $e$. Let $v$ be the speed of $k$ with respect to $e$. Then the worldline of $k$ according to $e$ is defined by equation $x=vt$, as illustrated in Figure ~\ref{fig-radar2D}.  Therefore, the worldline of the light signal sent by $k$ in the positive $x$ direction at instant $t_0$ satisfies  equation $t-t_0=\frac{x-vt_0}{\mathfrak{c_e}}$. Similarly, the light signal received by $k$ in the negative $x$ direction at instant $t_2$ satisfies equation $t-t_2=\frac{x-vt_2}{-\mathfrak{c_e}}$. Using that $(t_1, x_1)$ are on these lines, we get 
    \begin{equation}\label{eq02}
      \begin{array}{rl}
      \mathfrak{c_e}(t_1-t_0)&=x_1-vt_0, \\ \mathfrak{-c_e}(t_1-t_2)&=x_1-vt_2.
      \end{array}.
    \end{equation}
Solving \eqref{eq02} for $t_0$ and $t_2$, we get 
    \begin{equation}\label{t0,t2} 
           t_0=\frac{\mathfrak{c_e}t_1-x_1}{\mathfrak{c_e}-v} ,\quad t_2=\frac{\mathfrak{c_e}t_1+x_1}{\mathfrak{c_e}+v}.
    \end{equation}

 Let us call the coordinates of the event at $(t_1,x_1)$ in the ether coordinates and $(t'_1,x'_1)$ in the coordinates of the moving observer. Since the light signal has the same speed in both directions, $t'_1$ is in the middle between $t_0$ and $t_2$, which leads to $t'_1 = \frac{t_0+t_2}{2}$ and $x'_1=\frac{\mathfrak{c_e}(t_2-t_0)}{2}$.\\
Substituting \eqref{t0,t2} in the above equations , we find that  
  \begin{equation*}
           t_1'=\frac{\mathfrak{c^2_e}t_1-vx_1}{\mathfrak{c^2_e}-v^2} ,\quad x_1'=\frac{\mathfrak{c^2_e}x_1-\mathfrak{c^2_e}vt_1}{\mathfrak{c^2_e}-v^2}.
   \end{equation*}

Therefore, if $k$ uses radar to coordinatize spacetime,  $(t_1,x_1)$ is mapped to $\frac{1}{\mathfrak{c}_\mathfrak{e}^2-v^2}\left(\mathfrak{c}_\mathfrak{e}^2 t_1-v x_1, \mathfrak{c}_\mathfrak{e}^2 x_1-\mathfrak{c}_\mathfrak{e}^2 v t_1\right)$ when switching from the coordinate system of $e$ to that of $k$.

Now, let us consider the effect of using radar to coordinatize spacetime in the directions orthogonal to the movement of $k$, as illustrated in Figure ~\ref{fig-correction}. Consider a light signal moving with speed $\mathfrak{c_e}$ in an orthogonal direction, say the $y$ direction, which is reflected by a mirror after traveling one time unit. 
Due to the movement of the observer through the ether frame, there will be an apparent deformation of distances orthogonal to the movement: if the light signal takes one time unit relative to the ether, it will appear to only travel a distance which would, because of Pythagoras' theorem, be covered in $\sqrt{1-v^2/\mathfrak{c}_\mathfrak{e}^2}$ time units to the observer moving along in the $x$ direction. The same holds for the $z$ direction. So the point with coordinates $(0,0,1,1)$ in the ether coordinates will have coordinates 
\begin{equation*}
\left(0,0,\frac{1}{\sqrt{1-v^2/\mathfrak{c}_\mathfrak{e}^2}},\frac{1}{\sqrt{1-v^2/\mathfrak{c}_\mathfrak{e}^2}}\right)
\end{equation*}
 relative to the moving observer. This explain the values on the diagonal line in the lower right corner of the Poincar\' e--Einstein synchronisation matrix $E_v$. 

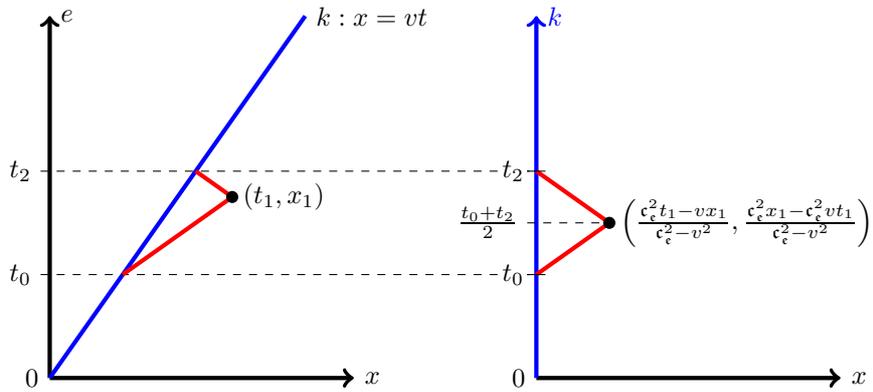
\begin{figure}
  \begin{center}
    \begin{tikzpicture}[scale=0.4]

\pgfmathsetmacro{\c}{1.4}
\pgfmathsetmacro{\v}{.5*\c}
\pgfmathsetmacro{\x}{6}
\pgfmathsetmacro{\t}{6}
\pgfmathsetmacro{\cpv}{\c+\v}
\pgfmathsetmacro{\cmv}{\c-\v}
\pgfmathsetmacro{\to}{\c*\t/\cmv-\x/\cmv}
\pgfmathsetmacro{\tt}{\c*\t/\cpv+\x/\cpv}

\coordinate (X) at  (\x,\t);
\coordinate (R) at  (\v*\tt,\tt);
\coordinate (S) at  (\v*\to,\to);
\draw [->, ultra thick] (0,0) node[left] {$0$} -- (10,0)  node[right] {$x$} ;
\draw [->, ultra thick] (0,0) -- (0,12) node[right] {$e$};
\draw [blue, ultra thick] (0,0) -- (\v*12,12)  node[right,black] {$k: x=vt$};
\draw [dashed]  (-0.3,\tt)  node[left] {$t_2$} to  (16,\tt)  ;
\draw [dashed] (-0.3,\to) node[left] {$t_0$} -- (16,\to) ;
\draw [red, ultra thick] (\to*\v,\to) -- (X) node[right,black] {$(t_1,x_1)$};
\draw [red, ultra thick] (X) -- (\tt*\v,\tt);
\fill (X) circle (0.2);

\begin{scope}[shift={(16,0)}]
\coordinate (Z) at (\c*\tt/2-\c*\to/2,\tt/2+\to/2);
\draw [->, ultra thick] (0,0) node[left] {$0$} -- (10,0)  node[right] {$x$};
\draw [->, blue, ultra thick] (0,0) -- (0,12)  node[right] {$k$};
\draw [red, ultra thick] (0,\to) -- (Z) node[right,black] {$\left(\frac{\mathfrak{c}_\mathfrak{e}^2t_1-vx_1}{\mathfrak{c}_\mathfrak{e}^2-v^2},\frac{\mathfrak{c}_\mathfrak{e}^2x_1-\mathfrak{c}_\mathfrak{e}^2vt_1}{\mathfrak{c}_\mathfrak{e}^2-v^2}\right)$};
\draw [red, ultra thick] (Z) -- (0,\tt);
\fill (Z) circle (0.2);
\draw [dashed] (-0.3,\tt/2+\to/2) node[left] {$\frac{t_0+t_2}{2}$}-- (Z);
\draw (-0.3,\to) node[left,fill=white,inner sep=1] {$t_0$} -- (0.3,\to) ;
\draw (-0.3,\tt) node[left,fill=white,inner sep=1] {$t_2$} -- (0.3,\tt) ;
\end{scope}
\end{tikzpicture}
    \caption{\label{fig-radar2D}Einstein--Poincar{\' e} synchronisation in two-dimensional classical kinematics: on the left in the coordinate system of an ether observer and on the right in the coordinate system of the moving observer. We only assume that the moving observer goes through the origin of the ether observer.}
  \end{center}
\end{figure}

Consequently, the following matrix describes the Poincar\' e--Einstein synchronisation in classical physics:
\begin{equation*}
E_v \de \left[ \begin{array}{cccc}
\frac{1}{1-v^2/\mathfrak{c}_\mathfrak{e}^2} & \frac{-v/\mathfrak{c}_\mathfrak{e}^2}{1-v^2/\mathfrak{c}_\mathfrak{e}^2} & 0 & 0\\
\frac{-v}{1-v^2/\mathfrak{c}_\mathfrak{e}^2} & \frac{1}{1-v^2/\mathfrak{c}_\mathfrak{e}^2} & 0 & 0\\
0 & 0 & \frac{1}{\sqrt{1-v^2/\mathfrak{c}_\mathfrak{e}^2}} & 0\\
0 & 0 & 0 & \frac{1}{\sqrt{1-v^2/\mathfrak{c}_\mathfrak{e}^2}}
\end{array} \right].
\end{equation*}

This transformation generates some asymmetry: an observer in a moving spaceship would dtermine their spaceship to be smaller in the directions orthogonal to its movement, see Figure ~\ref{fig-correction}. We can eliminate this asymmetry by multiplying with a scale factor $S_v \de \sqrt{1-v^2/\mathfrak{c}_\mathfrak{e}^2}$ which slows the clock of $k$ down. The combined transformation $S_v \circ E_v$ is the following Lorentz transformation:\[
L_v \de  S_v \circ E_v = \begin{bmatrix}
\frac{1}{\sqrt{1-v^2/\mathfrak{c}_\mathfrak{e}^2}} & \frac{-v/\mathfrak{c}_\mathfrak{e}^2}{\sqrt{1-v^2/\mathfrak{c}_\mathfrak{e}^2}} & 0 & 0\\
\frac{-v}{\sqrt{1-v^2/\mathfrak{c}_\mathfrak{e}^2}} & \frac{1}{\sqrt{1-v^2/\mathfrak{c}_\mathfrak{e}^2}} & 0 & 0\\
0 & 0 & 1 & 0\\
0 & 0 & 0 & 1
\end{bmatrix}.
\]

\begin{figure}
  \begin{center}
    \begin{tikzpicture}[scale=1.4]

\draw[white, fill=lightgray!42,ultra thick](0,0) rectangle (3,1);
\draw[white, fill=lightgray!42,ultra thick](3,1)--(3.5,0.5) -- (3,0)-- cycle;

\draw[dashed, thick] (1,0) to (4,0);
\draw[dashed, thick] (1,1) to (4,1);
\draw[dashed, thick] (4,1)--(4.5,0.5) -- (4,0)-- cycle;

\draw[red, ultra thick](1,0)-- node[left] {$1$} (0,1) ;
\draw[ultra thick](1,0)-- node[right, xshift = -3] {$\sqrt{1-v^2/\mathfrak{c}_\mathfrak{e}^2}$} (1,1);
\draw[green!50!black, ultra thick](1,1)-- node[above] {$v/\mathfrak{c_e}$} (0,1);
\draw[->,green!50!black, ultra thick] (3.5,0.5) -- node[above left] {$\bar v/\mathfrak{c_e}$} (4.5,0.5);

\end{tikzpicture}
    \caption{\label{fig-correction}Correction in the directions orthogonal to movement: The diagonal line is the path of a light signal which travels 1 space unit to the opposite side of a spaceship moving with speed $v/\mathfrak{c_e}$ with respect to the Ether frame. By Pythagoras' theorem the width of the spaceship is not 1 but just $\sqrt{1-v^2/\mathfrak{c}_\mathfrak{e}^2}$ space units in the Ether frame.}
  \end{center}
\end{figure}
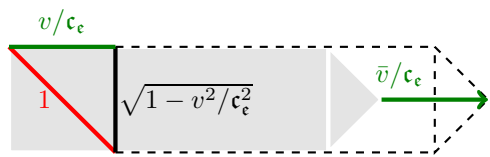

For all $\bar v=(v_x,v_y,v_z)\in\Q^3$ satisfying $v=|\bar v|<\mathfrak{c_e}$, we construct a bijection $Rad_{\bar{v}}$ (for ``\textit{radarization}'') between Minkowski spacetime and Newtonian absolute spacetime. 

We start by using the unique  spatial rotation $R_{\bar{v}}$ that rotates $(1,v_x,v_y,v_z)$ to $(1,-v,0 ,0)$ if the ether frame is not parallel with the $tx$-plane: 

\begin{itemize}
\item if $v_y \neq 0$ or $v_z \neq 0$ then
\[R_{\bar{v}} \de   \frac{1}{|\bar{v}|}\begin{bmatrix}
       1 & 0 & 0 & 0\\
       0 & -v_x & -v_y & -v_z\\
       0 & v_y & -v_x-\frac{v_z^2(|\bar{v}|-v_x)}{v_y^2+v_z^2} &\frac{v_yv_z(|\bar{v}|-v_x)}{v_y^2+v_z^2} \\
       0 & v_z & \frac{v_yv_z(|\bar{v}|-v_x)}{v_y^2+v_z^2} & -v_x-\frac{v_y^2(|\bar{v}|-v_x)}{v_y^2+v_z^2} 
     \end{bmatrix},
     \]
\item if $v_x \leq 0$ and $v_y = v_z = 0$ then $R_{\bar{v}}$ is the identity map,
\item if $v_x > 0$ and $v_y = v_z = 0$ then 
\[R_{\bar{v}} \de \begin{bmatrix}
       1 & 0 & 0 & 0\\
       0 & -1 & 0 & 0\\
       0 & 0 & -1 & 0 \\
       0 & 0 & 0 & -1 
     \end{bmatrix}.
     \] 
\end{itemize}

Rotation $R_{\bar{v}}$ is only dependent on the velocity $\bar{v}$. 

Then we take the Galilean boost\footnote{A Galilean boost is a time dependent translation: take a spacelike vector $\bar v$ and translate by $t_0\bar v$ in the horizontal hyperplane $t=t_0$, i.e. $(t,x,y,z) \mapsto (t,x+tv_x,y+tv_y,z+tv_z)$.} $G_v$ that maps the line $x=-vt$ to the time-axis, i.e., 
\[G_v \de
    \begin{bmatrix}
      1 & 0 & 0 & 0\\
      v & 1& 0 & 0\\
      0 & 0& 1 &0 \\
      0 & 0 & 0& 1
    \end{bmatrix}. \] 

Next we have the Lorentz transformation $S_v \circ E_v$. Finally, we use the reverse rotation $R^{-1}_{\bar{v}}$ to put the directions back to their original positions. So $Rad_{\bar{v}}$ is, as illustrated in Figure ~\ref{fig-radar4D}, the following composition:

\begin{equation*}
 Rad_{\bar{v}} \de R^{-1}_{\bar{v}} \circ S_v \circ E_v \circ G_{v} \circ R_{\bar{v}}.
\end{equation*}

$Rad_{\bar{v}}$ is a unique well-defined linear bijection for all $|\bar v|<\mathfrak{c_e}$. 

\label{core-matrix}
The \emph{core map} is the transformations in between the rotations:
\begin{equation*}
  C_v \de  S_v \circ E_v\circ G_v = \begin{bmatrix}
    \sqrt{1-v^2/\mathfrak{c}_\mathfrak{e}^2} & \frac{-v/\mathfrak{c}_\mathfrak{e}^2}{\sqrt{1-v^2/\mathfrak{c}_\mathfrak{e}^2}} & 0 & 0\\
    0 & \frac{1}{\sqrt{1-v^2/\mathfrak{c}_\mathfrak{e}^2}} & 0 & 0\\
    0 & 0 & 1 & 0\\
    0 & 0 & 0 & 1
  \end{bmatrix}.
\end{equation*}

It is worth noting that the core map is only dependent on the speed $v$.

\begin{figure}
  \begin{center}
    \begin{tikzpicture}[scale=0.47]
\pgfmathsetmacro\size{4}
\pgfmathsetmacro\v{0.5}
\pgfmathsetmacro\r{0.13}
\pgfmathsetmacro\scaling{1.1/(1-\v*\v)}

\tikzstyle{T0}=[cm={1,0,-\v,1,(-8,0)}]
\tikzstyle{T5}=[cm={1,0,-\v,1,(2,0)}]
\tikzstyle{T1}=[cm={1,0,0,1,(9,0)}]
\tikzstyle{T2}=[cm={1,0,0,1,(9,-7)}]
\tikzstyle{T3}=[cm={0.8,0,0,0.8,(0,-7)}]
\tikzstyle{T4}=[cm={1,0,0,1,(-8,-7)}]

\newcommand{\cone}[3]{\draw[thick, red,shift={(#2,#3)}] (0,#1) ellipse (#1 and 0.1*#1) (-#1,#1)--(0,0)--(#1,#1);}

\newcommand{\conemirror}[3]{\draw[thick, red,shift={(#2,#3)}] (3,#1) ellipse (#1 and 0.1*#1) (-#1+3,#1)--(0,0)--(#1+3,#1);} 

\pgfmathsetmacro\Xx{0.5*\size}
\pgfmathsetmacro\Xy{-.3*\size}
\pgfmathsetmacro\Yx{-.7*\size}
\pgfmathsetmacro\Yy{-.3*\size}

\begin{scope}[T0]
\draw[ultra thick, blue] (0,0) to (\v*\size,\size) node[right]{$k$};
\draw (4,\size) node[right] {$e$} -- (0,0)  ;
\conemirror{0.75*\size}{0}{0}
\draw[ultra thick, blue] (0.5*\v*\size,0.5*\size) to (\v*\size,\size);
\end{scope}

\draw[shift={(-7,0)}] (\Xx-1,\Xy) -- (-1,0) -- (\Yx-1,\Yy )node[below]{$x$};

\draw[->,ultra thick,>=latex] (-4,1)  to [out=30,in=150] node[above]{$R_{\bar{v}}$} (-2,1);

\begin{scope}[T5]
\draw[ultra thick, blue] (0,0) to (\v*\size,\size)  node[right]{$k$};
\draw (0,\size) node[right] {$e$} -- (0,0) -- (\size,0)node[below]{$x$} (0,0) -- (-.5*\size,-.3*\size);
\cone{0.75*\size}{0}{0} 
\draw[ultra thick, blue] (0.5*\v*\size,0.5*\size) to (\v*\size,\size);

\end{scope}

\draw[->,ultra thick,>=latex] (4,1)  to [out=30,in=150] node[above]{$G_v$} (6,1);

\begin{scope}[T1]
\draw[ultra thick, blue] (0,0) to (\v*\size,\size)  node[left]{$k$};
\draw (0,\size) node[right] {$e$} -- (0,0) -- (\size,0)node[below]{$x$} (0,0) -- (-.5*\size,-.3*\size);
\draw[thick, red] (\size*0.75,\size*0.75) to (0.5*\size,\size) ;
\cone{0.75*\size}{0}{0} 
\draw[ultra thick, blue] (0.5*\v*\size,0.5*\size) to (\v*\size,\size);

\end{scope}

\draw[->,ultra thick,>=latex] (10,-1)  to  node[right]{$E_v$} (10,-3);

\begin{scope}[T2]
\draw[ultra thick, blue] (0,0) to (0,\size) node[left]{$k$};
\cone{0.5*\size}{0}{0} 
\draw (-\v*\size,\size) node[right] {$e$} -- (0,0) -- (\size,0)node[below]{$x$} (0,0) -- (-.5*\size,-.3*\size);
\draw[thick, red] (\size/2,\size/2) to (0,\size) ;
\draw[ultra thick, blue] (0,0.5*\size) to (0,\size);
\end{scope}

\draw[<-,ultra thick,>=latex] (4,-5)  to [out=30,in=150] node[above]{$S_v$} (6,-5);

\begin{scope}[T3]
\draw[ultra thick, blue] (0,0) to (0,\size*1.25) node[right]{$k$};
\cone{0.5*\size}{0}{0} 
\draw (-\v*\size*1.25,\size*1.25) node[right] {$e$} -- (0,0) -- (\size,0)node[below]{$x$} (0,0) -- (-.5*\size,-.3*\size);
\draw[ultra thick, blue] (0,0.5*\size) to (0,\size);
\end{scope}

\draw[->,ultra thick,>=latex] (-3,-5)  to [out=120,in=30] node[above]{$R^{-1}_{\bar{v}}$} (-5,-5);

\begin{scope}[T4]
\draw[ultra thick, blue] (0,0) to (0,\size) node[right]{$k$};
\cone{0.4*\size}{0}{0} 
\draw (-\v*\size+4,\size) node[right] {$e$} -- (0,0) 
 (\Xx,\Xy) -- (0,0) -- (\Yx,\Yy)node[below]{$x$};
\draw[ultra thick, blue] (0,0.4*\size) to (0,\size);
\end{scope}

\draw[<-,ultra thick,>=latex]  (-9,-3) -- node[right]  {$Rad_{\bar{v}}$} (-9,-1);

\draw[->,ultra thick,>=latex] (0,-0.8)  to  node[right]{$C_v$} (0,-2.8);
\draw[->,ultra thick,>=latex] (5,-2)  to  node[above]{$L_v$} (3,-3);

\end{tikzpicture}
    \caption{\label{fig-radar4D}The components of transformation $Rad_{\bar{v}}$: reading from the top left corner to the bottom left corner, first  we put ether observer $e$ in the $kx$-plane by rotation $R_{\bar{v}}$, we put $e$ to the time axis by Galilean transformation $G_{v}$, then by using Einstein-Poincar{\'e } transformation $E_v$ we put $k$ to the time axis, then by scaling $S_v$ we correct the asymmetry in the directions orthogonal to the movement, and finally we use the inverse rotation $R^{-1}_{\bar{v}}$ to put the direction of $e$ back  into place. The Lorentz transformation $L_v$ and the core map $C_v$ are  also being displayed. The triangles formed by $k$, the outgoing lightbeam on the right of the lightcone and the incoming light beam, on the right hand size of Figure ~\ref{fig-radar4D}, are the same triangles as in Figure ~\ref{fig-radar2D}.}
  \end{center}
\end{figure}
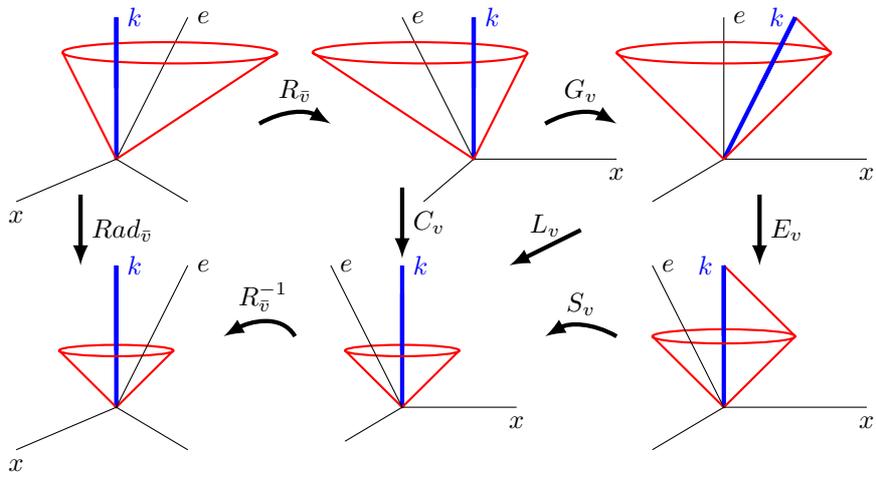

We say that cone $\Lambda$ is a \emph{light cone moving with velocity $\bar v=(v_x,v_y,v_z)$} if $\Lambda$ is the translation, by a $Q^4$ vector, of the following cone  \[\left\{ (t,x,y,z)  \in Q^4: (x-v_x t)^2+(y-v_y t)^2+(z-v_z t)^2=(\mathfrak{c_e}t)^2 \right\}.\] We call light cones moving with velocity $(0,0,0)$ \emph{right light cones}.

\begin{lem}\label{thm-lemma} 
Assume \sy{\mathsf{ClassicalKin}_{Full}}. Let $\bar {v}=(v_x, v_y, v_z) \in \Q^3$ such that $|\bar {v}|<\mathfrak{c_e}$. Then $Rad_{\bar{v}}$ is a linear bijection that has the following properties:

\begin{enumerate}
\item\label{rad-0} If $\bar{v}=(0,0,0)$, then $Rad_{\bar{v}}$ is the identity map.
\item\label{rad-time} \label{tx} $Rad_{\bar{v}}$ maps the time axis to the time axis, i.e., 
\begin{equation*}
\Bforall {\bar y,\bar x}{Q^4} \Big( Rad_{\bar v}(\bar x) = \bar y \\ \to \big[x_1=x_2=x_3=0  \leftrightarrow y_1=y_2=y_3=0 \big]  \Big). 
\end{equation*}
\item\label{rad-time-scale} $Rad_{\bar{v}}$ scales the time axis down by factor $\sqrt{1-|\bar{v}|^2}$.
\item\label{rad-cone} $Rad_{\bar{v}}$ transforms light cones moving with velocity $\bar {v}$ into right light cones.
\item\label{rad-perp} $Rad_{\bar {v}}$ is the identity on vectors orthogonal to the plane containing the time axis and the direction of motion of the ether frame $(0, v_x, v_y, v_z)$, i.e., 
\[\Bforall{\bar x}{Q^4}\big([ x_1 v_x + x_2 v_y + x_3 v_z = 0 \AND x_0 = 0 ] \to Rad_{\bar v}(\bar x)=\bar x\big).\]
\item\label{rad-vel} The line through the origin moving with velocity $\bar{v}$ is mapped to itself and lines parallel to this line are mapped to parallel ones by $Rad_{\bar {v}}$.
\end{enumerate}
\end{lem}
\begin{proof}
To define $Rad_{\bar{v}}$, we need \AxEField to allow us to use subtractions, divisions and square roots; \AxLine because worldlines of observers must be straight lines, which also enables us to calculate the speed $v$; and \ax{AxEther} because we need the ether frame of reference.
$Rad_{\bar{v}}$ is well-defined because it is composed of well-defined components. Speed $v$ is well-defined because of \AxLine and \AxEField.

$Rad_{\bar{v}}$ is a bijection since it is composed of bijections: Galilean transformation $G$ is a bijection, rotations $R$ and $R^{-1}$ are bijections, matrix $E$ defines a bijection, multiplication by $S$ is a bijection (because we only use the positive square root). $Rad_{\bar{v}}$ maps lines to lines since all components of $Rad_{\bar{v}}$ are linear. By definition, $Rad_{\bar 0}$ is the identity map.

The time axis is mapped to the time axis by $Rad_{\bar{v}}$:  $R_{\bar {v}}$ leave the time axis in place. Galilean transformation $G_{v}$ maps the time axis to the line defined by $x=vt$. Matrix $E_{v}$ maps this line back to the time axis. $S_{v}$ and $R^{-1}_{\bar{v}}$ leaves the time axis in place.
The rotations around the time axis do not change the time axis. The core map $C_v$ scales the time axis by factor $\sqrt{1-|\bar{v}|^2}$.

$G_{v}\circ R_{\bar{v}}$ transforms light cones moving with velocity $\bar {v}$ into right light cones. The rest of the transformations giving $Rad_{\bar v}$ map right light cones to right ones.

If $\bar x$ is orthogonal to $(1,0,0,0)$ and $(0,\bar v)$, then $R_{\bar v}$ rotates it to become orthogonal to the $tx$-plane. Therefore, the Galilean boost $G_{v}$ and the Lorentz boost $S_v\circ E_v$ do not change the $R_{\bar v}$ image of $\bar x$. Finally $R^{-1}_{\bar v}$ rotates this image back to $\bar x$. Hence $Rad_{\bar v}(\bar x)=\bar x$.

The line through the origin with velocity $\bar{v} = (v_x,v_y,v_z)$, which is line $e$ in Figure \ref{fig-radar4D}, is defined by equation system $x=v_xt$, $y=v_yt$ and $z=v_zt$. After rotation $R_{\bar v}$, this line has speed $-v$ in the $tx$ plane, which core map $C_{v}$ maps to itself:
\begin{equation*}
\begin{bmatrix}
    \sqrt{1-v^2/\mathfrak{c}_\mathfrak{e}^2} & \frac{-v/\mathfrak{c}_\mathfrak{e}^2}{\sqrt{1-v^2/\mathfrak{c}_\mathfrak{e}^2}} & 0 & 0\\
    0 & \frac{1}{\sqrt{1-v^2/\mathfrak{c}_\mathfrak{e}^2}} & 0 & 0\\
    0 & 0 & 1 & 0\\
    0 & 0 & 0 & 1
\end{bmatrix} \cdot
\begin{bmatrix}
    1\\
    -v\\
    0\\
    0
\end{bmatrix} =
\begin{bmatrix}
    \frac{1}{\sqrt{1-v^2/\mathfrak{c}_\mathfrak{e}^2}}\\
    \frac{-v}{\sqrt{1-v^2/\mathfrak{c}_\mathfrak{e}^2}}\\
    0\\
    0
\end{bmatrix} = \frac{1}{\sqrt{1-v^2/\mathfrak{c}_\mathfrak{e}^2}}
\begin{bmatrix}
    1\\
    -v\\
    0\\
    0
\end{bmatrix}.
\end{equation*}
Rotation $R^{-1}_{\bar v}$ puts the line with speed $-v$ back on the line with velocity $\bar{v}$.
Since $Rad_{\bar{v}}$ is a linear bijection, lines parallel to this line are mapped to parallel lines.
\end{proof}

\section{A formal translation of SpecRel into ClassicalKin}

In this section, using the radarization transformations of section~\ref{sec:rad}, we give a formal translation from the language of \sy{SpecRel} to that of \sy{ClassicalKin} such that all the translated axioms of \ax{SpecRel_{Full}} become theorems of \sy{ClassicalKin_{Full}}. To do so, we will have to translate the basic concepts of \sy{SpecRel} to formulas of the language of \sy{ClassicalKin}.

Since the basic concepts of the two languages use the same symbols, we indicate in a superscript whether we are speaking about the classical or the relativistic version when they are not translated identically. So we use $IOb^{SR}$ and $W^{SR}$ for relativistic inertial observers and worldview relations, and  $IOb^{CK}$ and $W^{CK}$ for classical inertial observers and worldview relations. Even though from the context it is always clear which language we use because formulas before the translation are in the language of \sy{SpecRel} and formulas after the translation are in the language of \sy{ClassicalKin}, sometimes we use this notation even in defined concepts (such as events, worldlines, and worldview transformations) to help the readers.

\noindent
Let us define $Rad_{\bar v_k(e)}(\bar x)$ and its inverse $Rad^{-1}_{\bar v_k(e)}(\bar y)$ as
\begin{equation*}
Rad_{\bar v_k(e)}(\bar x) = \bar y 
\defiff \Bexists{\bar v}{Q^3}[\bar v = \bar v_k(e) \AND Rad_{\bar v}(\bar x) = \bar y ]
\end{equation*}
\begin{equation*}
Rad^{-1}_{\bar v_k(e)}(\bar y) = \bar x 
\defiff \Bexists{\bar v}{Q^3}[\bar v = \bar v_k(e) \AND Rad_{\bar v}(\bar x) = \bar y].
\end{equation*}

Let us now give the translation of all the basic concepts  of \sy{SpecRel} in the language of \sy{ClassicalKin}. Mathematical expressions are translated into themselves:
\[Tr ( a+b=c ) \defeq (a+b=c),\ 
Tr ( a \cdot b=c ) \defeq (a \cdot b=c),\ 
Tr ( a < b) \defeq a < b.\]

\noindent
Light signals are translated to light signals:
$Tr\big(Ph(p)\big)\defeq Ph(p).$

\noindent
The translation of relativistic inertial observers are classical inertial observers which are slower-than-light with respect to the ether frame:
\begin{equation*}
Tr \big( IOb^{SR}(k) \big) \defeq
IOb^{CK}(k)\AND \bforall {e} {\Ether} \big[ \speed_{e}(k) < \mathfrak{c_e}\big].
\end{equation*}

\noindent
Relativistic coordinates are translated into classical coordinates by radarization\footnote{By Convention~\ref{partial-conv} on page \pageref{partial-conv}, a relation defined by formula $Tr \big( W^{SR}(k,b,\bar{x}) \big)$ is empty if $wl_k(e)$ is not a subset of a straight line for every ether observer $e$ (because in this case the partial function $\bar{v}_k(e)$ is undefined). The same applies to the translations of the defined concepts \textit{event}, \textit{worldline} and \textit{worldview transformation}.}:
\begin{equation*}
Tr \big( W^{SR}(k,b,\bar{x}) \big)
\defeq \bforall{e}{\Ether} \big[W^{CK}\big(k,b,Rad^{-1}_{\bar{v}_k(e)}(\bar{x})\big)\big].
\end{equation*}

\noindent
Complex formulas are translated by preserving the logical connectives: 
\[
Tr(\neg \varphi) \defeq \neg  Tr(\varphi),\  
Tr(\psi \AND \varphi) \defeq Tr (\psi) \land Tr(\varphi),\  
Tr(\exists x [\varphi]) \defeq \exists x [Tr(\varphi)],\   
\text{ etc.}
\]
This defines translation $Tr$ on all formulas in the language of \sy{SpecRel_{Full}}.

Let us now see into what $Tr$ translates the important defined concepts, such as events, worldlines and worldview transformations. 

\noindent
Worldlines are translated as:
\begin{equation*}
Tr \big(\vx \in wl_k(b) \big) \equiv
\bforall {e} {\Ether}\left[Rad^{-1}_{\bar{v}_k(e)}(\vx)\in \wl_k(b)\right]
\end{equation*}
and events as:
\begin{equation*}
Tr \big(b \in \ev_k(\vx) \big) \equiv
\bforall {e} {\Ether}\left[b \in \ev_k(Rad^{-1}_{\bar{v}_k(e)}(\vx))\right].
\end{equation*}
Since these translations often lead to very complicated formulas, we provide some techniques to simplify translated formulas in the Appendix on p.\pageref{appendix}. In the proofs below, we will always use the simplified formulas. The simplified translation of the worldview transformation is the following:
\[Tr \big( w^{SR}_{hk}(\vx, \vy) \big) \equiv
\bforall {e} {\Ether} \left[ w^{CK}_{hk}\big( Rad^{-1}_{\bar{v}_h(e)}(\bar x), Rad^{-1}_{\bar{v}_k(e)}(\bar y) \big) \right].
\]
Since \sy{ClassicalKin_{Full}} implies that $w^{CK}_{hk}$ is a transformation (and not just a relation) if $k$ and $h$ are inertial observers, 
\[Tr\big(w^{SR}_{hk} (\vx, \vy)\big)\equiv \bforall {e} {\Ether}\left[\left(Rad_{\bar{v}_k(e)}\circ w^{CK}_{hk}\circ Rad^{-1}_{\bar{v}_h(e)}\right)(\bar x)=\bar y\right]\]
in this case.

\begin{lem}\label{lemma-trc}
Assume \sy{ClassicalKin_{Full}}. Then $Tr(\mathfrak{c}) \equiv \mathfrak{c_e}$.\footnote{That is, the translation of the defining formula of constant $\mathfrak{c}$ is equivalent to the defining formula of constant $\mathfrak{c_e}$ in \sy{ClassicalKin_{Full}}. The same remark, with $\mathfrak{c_e}$ and $\mathfrak{c}$ switched and on \sy{SpecRel_{Full}}, can be made for Corollary \ref{cor-Tr_+'(c_e)} below.}
\end{lem}
\begin{proof} 
By $Tr(\ax{AxPh_c^{SR}})$ we know there is an observer-independent speed of light for the translated inertial observers. So we can chose any translated inertial observer to establish the speed of light. Ether observers are also translations of some inertial observers because they are inertial observers moving slower than $\mathfrak{c_e}$ with respect to ether observers. Let us take an ether observer, which in the translation has $Rad_{\bar 0}$ being the identity. Hence, $Tr(\mathfrak{c})$ is the speed of light according to our fixed ether observer, which is $\mathfrak{c_e}$ in \sy{ClassicalKin_{Full}} by definition.
\end{proof}

Lemma~\ref{lemma-cannon} is helpful for proving properties of translations involving more than one observer:

\begin{lem}\label{lemma-cannon}  
Assuming \sy{{ClassicalKin}_{Full}}, if $e$ and $e'$ are ether observers and $k$ and $h$ are slower-than-light inertial observers, then
\begin{equation*}
Rad_{\bar{v}_k(e)}\circ w^{CK}_{hk}\circ Rad^{-1}_{\bar{v}_h(e)}= Rad_{\bar{v}_k(e')}\circ w^{CK}_{hk}\circ Rad^{-1}_{\bar{v}_h(e')}
\end{equation*}
and it is a Poincar\' e transformation.
\end{lem}

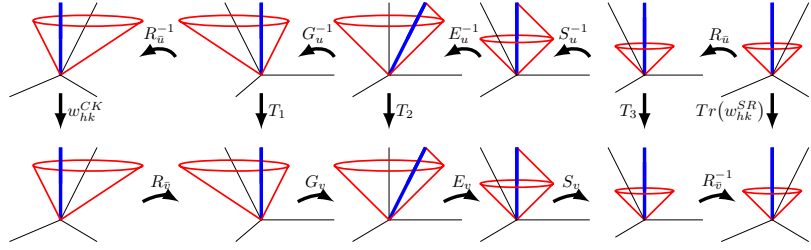
\begin{figure}
  \begin{center}
    \scalebox{.8}{\begin{tikzpicture}[scale=0.3, every node/.style={scale=0.8}]
\pgfmathsetmacro\size{4}
\pgfmathsetmacro\v{0.5}
\pgfmathsetmacro\r{0.13}
\pgfmathsetmacro\scaling{1.1/(1-\v*\v)}

\tikzstyle{T0}=[cm={1,0,-\v,1,(-9,0)}]
\tikzstyle{T5}=[cm={1,0,-\v,1,(2,0)}]
\tikzstyle{T1}=[cm={1,0,0,1,(9,0)}]
\tikzstyle{T2}=[cm={1,0,0,1,(16,0)}]
\tikzstyle{T3}=[cm={0.8,0,0,0.8,(23,0)}]
\tikzstyle{T4}=[cm={1,0,0,1,(30,0)}]
\tikzstyle{T10}=[cm={1,0,-\v,1,(-9,-8)}]
\tikzstyle{T15}=[cm={1,0,-\v,1,(2,-8)}]
\tikzstyle{T11}=[cm={1,0,0,1,(9,-8)}]
\tikzstyle{T12}=[cm={1,0,0,1,(16,-8)}]
\tikzstyle{T13}=[cm={0.8,0,0,0.8,(23,-8)}]
\tikzstyle{T14}=[cm={1,0,0,1,(30,-8)}]

\newcommand{\cone}[3]{\draw[thick, red,shift={(#2,#3)}] (0,#1) ellipse (#1 and 0.1*#1) (-#1,#1)--(0,0)--(#1,#1);}

\newcommand{\conemirror}[3]{\draw[thick, red,shift={(#2,#3)}] (3,#1) ellipse (#1 and 0.1*#1) (-#1+3,#1)--(0,0)--(#1+3,#1);} 

\pgfmathsetmacro\Xx{0.5*\size}
\pgfmathsetmacro\Xy{-.3*\size}
\pgfmathsetmacro\Yx{-.7*\size}
\pgfmathsetmacro\Yy{-.3*\size}

\begin{scope}[T0]
\draw[ultra thick, blue] (0,0) to (\v*\size,\size) ;
\draw (4,\size)  -- (0,0)  ;
\conemirror{0.75*\size}{0}{0}
\draw[ultra thick, blue] (0.5*\v*\size,0.5*\size) to (\v*\size,\size);
\end{scope}

\draw[shift={(-8,-8)}] (\Xx-1,\Xy) -- (-1,0) -- (\Yx-1,\Yy );

\draw[->,ultra thick,>=latex] (-2.7,1)  to [out=120,in=30] node[above]{$R^{-1}_{\bar{u}}$} (-4.7,1);

\begin{scope}[T5]
\draw[ultra thick, blue] (0,0) to (\v*\size,\size)  ;
\draw (0,\size)  -- (0,0) -- (\size,0) (0,0) -- (-.5*\size,-.3*\size);
\cone{0.75*\size}{0}{0} 
\draw[ultra thick, blue] (0.5*\v*\size,0.5*\size) to (\v*\size,\size);

\end{scope}

\draw[->,ultra thick,>=latex] (6,1)  to [out=120,in=30] node[above]{$G^{-1}_u$} (4,1);

\begin{scope}[T1]
\draw[ultra thick, blue] (0,0) to (\v*\size,\size)  ;
\draw (0,\size)  -- (0,0) -- (\size,0) (0,0) -- (-.5*\size,-.3*\size);
\draw[thick, red] (\size*0.75,\size*0.75) to (0.5*\size,\size) ;
\cone{0.75*\size}{0}{0} 
\draw[ultra thick, blue] (0.5*\v*\size,0.5*\size) to (\v*\size,\size);

\end{scope}

\draw[->,ultra thick,>=latex] (14,1)  to [out=120,in=30] node[above]{$E^{-1}_{u}$} (12,1);

\begin{scope}[T2]
\draw[ultra thick, blue] (0,0) to (0,\size) ;
\cone{0.5*\size}{0}{0} 
\draw (-\v*\size,\size)  -- (0,0) -- (\size,0) (0,0) -- (-.5*\size,-.3*\size);
\draw[thick, red] (\size/2,\size/2) to (0,\size) ;
\draw[ultra thick, blue] (0,0.5*\size) to (0,\size);
\end{scope}

\draw[->,ultra thick,>=latex] (20,1)  to [out=120,in=30] node[above]{$S^{-1}_{u}$} (18,1);

\begin{scope}[T3]
\draw[ultra thick, blue] (0,0) to (0,\size*1.25) ;
\cone{0.5*\size}{0}{0} 
\draw (-\v*\size*1.25,\size*1.25)  -- (0,0) -- (\size,0) (0,0) -- (-.5*\size,-.3*\size);
\draw[ultra thick, blue] (0,0.5*\size) to (0,\size);
\end{scope}

\draw[->,ultra thick,>=latex] (28,1)  to [out=120,in=30] node[above]{$R_{\bar u}$} (26,1);

\begin{scope}[T4]
\draw[ultra thick, blue] (0,0) to (0,\size) ;
\cone{0.4*\size}{0}{0} 
\draw (-\v*\size+4,\size)  -- (0,0) 
 (\Xx,\Xy) -- (0,0) -- (\Yx,\Yy);
\draw[ultra thick, blue] (0,0.4*\size) to (0,\size);
\end{scope}

\draw[<-,ultra thick,>=latex]  (-9,-3) -- node[right]  {$w^{CK}_{hk}$} (-9,-1);
\draw[<-,ultra thick,>=latex]  (30,-3) -- node[left]  {$Tr \big(w^{SR}_{hk}\big)$} (30,-1);
\draw[<-,ultra thick,>=latex]  (2,-3) -- node[right]  {$T_1$} (2,-1);
\draw[<-,ultra thick,>=latex]  (9,-3) -- node[right]  {$T_2$} (9,-1);
\draw[<-,ultra thick,>=latex]  (23,-3) -- node[left]  {$T_3$} (23,-1);

\begin{scope}[T10]
\draw[ultra thick, blue] (0,0) to (\v*\size,\size) ;
\draw (4,\size)  -- (0,0)  ;
\conemirror{0.75*\size}{0}{0}
\draw[ultra thick, blue] (0.5*\v*\size,0.5*\size) to (\v*\size,\size);
\end{scope}

\draw[shift={(-8,0)}] (\Xx-1,\Xy) -- (-1,0) -- (\Yx-1,\Yy );

\draw[->,ultra thick,>=latex] (-4.5,-7)  to [out=30,in=150] node[above]{$R_{\bar{v}}$} (-2.5,-7);

\begin{scope}[T15]
\draw[ultra thick, blue] (0,0) to (\v*\size,\size)  ;
\draw (0,\size)  -- (0,0) -- (\size,0) (0,0) -- (-.5*\size,-.3*\size);
\cone{0.75*\size}{0}{0} 
\draw[ultra thick, blue] (0.5*\v*\size,0.5*\size) to (\v*\size,\size);

\end{scope}

\draw[->,ultra thick,>=latex] (4,-7)  to [out=30,in=150] node[above]{$G_v$} (6,-7);

\begin{scope}[T11]
\draw[ultra thick, blue] (0,0) to (\v*\size,\size)  ;
\draw (0,\size)  -- (0,0) -- (\size,0) (0,0) -- (-.5*\size,-.3*\size);
\draw[thick, red] (\size*0.75,\size*0.75) to (0.5*\size,\size) ;
\cone{0.75*\size}{0}{0} 
\draw[ultra thick, blue] (0.5*\v*\size,0.5*\size) to (\v*\size,\size);

\end{scope}

\draw[->,ultra thick,>=latex] (12,-7)  to   [out=30,in=150] node[above]{$E_v$} (14,-7);

\begin{scope}[T12]
\draw[ultra thick, blue] (0,0) to (0,\size) ;
\cone{0.5*\size}{0}{0} 
\draw (-\v*\size,\size)  -- (0,0) -- (\size,0) (0,0) -- (-.5*\size,-.3*\size);
\draw[thick, red] (\size/2,\size/2) to (0,\size) ;
\draw[ultra thick, blue] (0,0.5*\size) to (0,\size);
\end{scope}

\draw[->,ultra thick,>=latex] (18,-7)  to [out=30,in=150] node[above]{$S_{v}$} (20,-7);

\begin{scope}[T13]
\draw[ultra thick, blue] (0,0) to (0,\size*1.25) ;
\cone{0.5*\size}{0}{0} 
\draw (-\v*\size*1.25,\size*1.25)  -- (0,0) -- (\size,0) (0,0) -- (-.5*\size,-.3*\size);
\draw[ultra thick, blue] (0,0.5*\size) to (0,\size);
\end{scope}

\draw[->,ultra thick,>=latex] (26,-7)  to [out=30,in=150] node[above]{$R^{-1}_{\bar{v}}$} (28,-7);

\begin{scope}[T14]
\draw[ultra thick, blue] (0,0) to (0,\size) ;
\cone{0.4*\size}{0}{0} 
\draw (-\v*\size+4,\size)  -- (0,0) 
 (\Xx,\Xy) -- (0,0) -- (\Yx,\Yy);
\draw[ultra thick, blue] (0,0.4*\size) to (0,\size);
\end{scope}


\end{tikzpicture}}
    \caption{\label{fig-cannon}Lemma \ref{lemma-cannon}: Read the figure starting in the top-right corner and follow the arrows to the left along the components of $Rad^{-1}_{\bar u}$, down along Galilean transformation $w^{CK}_{hk}$ and right along the components of $Rad_{\bar v}$, which results in Poincar\' e transformation $Tr\big(w^{SR}_{hk}\big)$.}
  \end{center}
\end{figure}

\begin{proof}
Let $e$ and $e'$ be ether observes and let  $k$ and $h$ be inertial observers with velocities $\bar v = \bar v_{k}(e)$ and $\bar u = \bar v_{h}(e')$. By Corollary~\ref{cor-v}, $\bar v=\bar v_{k}(e')$ and $\bar u= \bar v_{h}(e)$. Therefore,
\[Rad_{\bar{v}_k(e)}\circ w^{CK}_{hk}\circ Rad^{-1}_{\bar{v}_h(e)}= Rad_{\bar{v}_k(e')}\circ w^{CK}_{hk}\circ Rad^{-1}_{\bar{v}_h(e')}.\]
By Theorem~\ref{thm-gal},
$w^{CK}_{hk}$ is a Galilean transformation. Trivial Galilean transformations are also (trivial) Poincar\' e transformations. Therefore
\begin{equation*}
T_1 = R_{\bar v} \circ w^{CK}_{hk} \circ R^{-1}_{\bar u}
\end{equation*}
is a Galilean transformation because it is a composition of a Galilean transformation and two rotations, which are also (trivial) Galilean transformations. $T_1$ is also a trivial transformation because it is a transformation between two ether observers, which are at rest relative to each other by Corollary~\ref{thm-lightspeed}.
\begin{equation*}
T_2 = G_{v} \circ T_1 \circ G^{-1}_{u}
\end{equation*}
is a (trivial) Galilean transformation. Since $S_{v} \circ \ E_{v}$ and $E^{-1}_{u} \circ S^{-1}_{u}$ are Lorentz transformations (which are special cases of Poincar\' e transformations),
\begin{equation*}
T_3 = S_{v} \circ \ E_{v} \circ T_2 \circ E^{-1}_{u} \circ S^{-1}_{u}
\end{equation*}
is a Poincar\' e transformation. Since rotations are (trivial) Poincar\'e transformations,
\begin{equation*}
Rad_{\bar{v}_k(e)}\circ w^{CK}_{hk}\circ Rad^{-1}_{\bar{v}_h(e)} = 
R^{-1}_{\bar v} \circ T_3 \circ R_{\bar u}
\end{equation*}
is a Poincar\' e transformation.
\end{proof}

Since, by Lemma \ref{lemma-cannon}, 
$Rad_{\bar{v}_k(e)}\circ w^{CK}_{hk}\circ Rad^{-1}_{\bar{v}_h(e)}$ leads to the same Poincar\' e transformation independently of the choice of ether observer $e$, we can use the notation $Tr\big(w^{SR}_{hk}\big)$ for this transformation, as on the right side of Figure ~\ref{fig-cannon}.

\begin{lem}\label{triv-poi}
Assume  \sy{{ClassicalKin}_{Full}}.  Let $e$ be an ether observer, and let $k$ and $h$ be slower-than-light inertial observers. Assume that $w^{CK}_{hk}$ is a trivial transformation consisting of the translation by the vector $\bar{z}$ after the linear trivial transformation $T$. Then $Rad_{\bar{v}_k(e)}\circ w^{CK}_{hk}\circ Rad^{-1}_{\bar{v}_h(e)}$ is the trivial transformation which is the translation by vector $Rad_{\bar{v}_k(e)}(\bar{z})$ after $T$.
\end{lem}

\begin{proof}
By Lemma~\ref{lemma-cannon}, $Rad_{\bar{v}_k(e)}\circ w^{CK}_{hk}\circ Rad^{-1}_{\bar{v}_h(e)}$ is a Poincar\'e transformation. Since  $w^{CK}_{hk}$ is a trivial transformation, it maps vertical lines to vertical ones. By Lemma~\ref{thm-lemma}, $Rad_{\bar{v}_k(e)}\circ w^{CK}_{hk}\circ Rad^{-1}_{\bar{v}_h(e)}$ also  maps vertical lines to vertical ones. Consequently, $Rad_{\bar{v}_k(e)}\circ w^{CK}_{hk}\circ Rad^{-1}_{\bar{v}_h(e)}$ is a Poincar\'e transformation that maps vertical lines to vertical ones. Hence it is a trivial transformation.

Let $M_{\bar{z}}$ denote the translation by vector $\bar{z}$.
By the assumptions, $w^{CK}_{hk}=M_{\bar{z}}\circ T$. The linear part $T$ of $w^{CK}_{hk}$ transforms the velocity of the ether frame as $(0,\bar{v}_k(e))=T(0,\bar{v}_h(e))$ and the translation part $M_{\bar{z}}$ does not change the velocity of the ether frame. Hence $\bar{v}_k(e)$ is $\bar{v}_h(e)$ transformed by the spatial isometry part of $T$.

We also have that $Rad_{\bar{v}_k(e)}\circ w^{CK}_{hk}\circ Rad^{-1}_{\bar{v}_h(e)}$ is $Rad_{\bar{v}_k(e)}\circ M_{\bar{z}} \circ T\circ Rad^{-1}_{\bar{v}_h(e)}$. Since $Rad_{\bar{v}_k(e)}$ is linear, we have  $Rad_{\bar{v}_k(e)}\circ M_{\bar{z}} = M_{Rad_{\bar{v}_k(e)}(\bar{z})}\circ Rad_{\bar{v}_k(e)}$. Therefore, it is enough to prove that $Rad_{\bar{v}_k(e)}\circ w^{CK}_{hk}\circ Rad^{-1}_{\bar{v}_h(e)}=w^{CK}_{hk}$ if $w^{CK}_{hk}$ is linear. From now on, assume that $w^{CK}_{hk}$ is linear.

Since it is a linear trivial transformation, $w^{CK}_{hk}$ maps $(1,0,0,0)$ to itself. By Item~\ref{rad-time-scale} of Lemma~\ref{thm-lemma}, $Rad_{\bar{v}_k(e)}\circ w^{CK}_{hk}\circ Rad^{-1}_{\bar{v}_h(e)}=w^{CK}_{hk}$ also  maps $(1,0,0,0)$ to itself because $Rad^{-1}_{\bar{v}_h(e)}$ scales up the time axis the same factor as $Rad_{\bar{v}_k(e)}$ scales down because $|\bar{v}_h(e)|=|\bar{v}_k(e)|$. So  $Rad_{\bar{v}_k(e)}\circ w^{CK}_{hk}\circ Rad^{-1}_{\bar{v}_h(e)}$ and $w^{CK}_{hk}$ agree restricted to time.

Now we have to prove that  $Rad_{\bar{v}_k(e)}\circ w^{CK}_{hk}\circ Rad^{-1}_{\bar{v}_h(e)}$ and $w^{CK}_{hk}$ also agree restricted to space. By Item~\ref{rad-perp} of Lemma~\ref{thm-lemma}, $Rad^{-1}_{\bar{v}_h(e)}$ is identical on the vectors orthogonal to the plane containing the time axis and the direction of motion of the ether frame, determined by vector $\bar{v}_h(e)$. The worldview transformation  $w^{CK}_{hk}$ leaves the time axis fixed and maps the velocity of ether frame $\bar{v}_h(e)$ to $\bar{v}_k(e)$. After this $Rad_{\bar{v}_k(e)}$ does not change the vectors orthogonal to the plane containing the time axis and the direction of motion of the ether frame. Therefore, $Rad_{\bar{v}_k(e)}\circ w^{CK}_{hk}\circ Rad^{-1}_{\bar{v}_h(e)}$ and $w^{CK}_{hk}$ do the same thing with the vectors orthogonal to the plane containing the time axis and the direction of motion of the ether frame, determined by vector $\bar{v}_h(e)$. So the space part of $Rad_{\bar{v}_k(e)}\circ w^{CK}_{hk}\circ Rad^{-1}_{\bar{v}_h(e)}$ and $w^{CK}_{hk}$ are isometries of $Q^3$ that agree on two independent vectors. This means that they are either equal or differ in a mirroring. However, they cannot differ in a mirroring as $Rad_{\bar{v}_k(e)}$ and $Rad^{-1}_{\bar{v}_h(e)}$ are orientation preserving maps because all of their components are such. Consequently, $Rad_{\bar{v}_k(e)}\circ w^{CK}_{hk}\circ Rad^{-1}_{\bar{v}_h(e)}=w^{CK}_{hk}$.
\end{proof}

\section{Interpretation}
Now that we have established the translation and developed the tools to simplify\footnote{See Appendix on p.\pageref{appendix}.} translated formulas, we prove that it is an interpretation of \sy{SpecRel_{Full}} in \sy{ClassicalKin_{Full}}.

\begin{thm}\label{thm-tr}
$Tr$ is an interpretation of \sy{SpecRel_{Full}} in \sy{ClassicalKin_{Full}}, i.e.,
\begin{equation*}
\sy{ClassicalKin_{Full}}\vdash Tr(\varphi)\enskip\text{ if }\enskip \sy{SpecRel_{Full}}\vdash \varphi.
\end{equation*}
\end{thm}

\begin{proof} 
It is enough to prove that the $Tr$-translation of every axiom of \ax{SpecRel_{Full}} follows from \ax{ClassicalKin_{Full}}. Now we will go trough all the axioms of \ax{SpecRel_{Full}} and prove their translations one by one from \ax{ClassicalKin_{Full}}.

\begin{itemize}[leftmargin=*]
\item $\ax{AxEField^{CK}} \vdash  Tr(\ax{AxEField^{SR}})$ follows  since all purely mathematical expressions are translated into themselves, hence $Tr(\AxEField)$ is the axiom $\AxEField$ itself.

\item $\sy{\mathsf{ClassicalKin}_{Full}} \vdash Tr(\ax{AxEv^{SR}})$. The translation of $\ax{AxEv^{SR}}$ is equivalent to:
\begin{multline*}
\bforall {k,h}{ IOb} 
\Bforall {\vx}{\Q^4} 
\bforall {e}{\Ether} 
\Bigg(\fblock{\speed_{e}(k) < \mathfrak{c_e} \\ \speed_{e}(h) < \mathfrak{c_e}} \\ \to 
\Bexists {\vy}{\Q^4}
\Big [\ev_k\left(Rad^{-1}_{\bar{v}_k(e)}(\vx)\right)=\ev_{h}\left(Rad^{-1}_{\bar{v}_h(e)}(\vy)\right)\Big]\Bigg).
\end{multline*}

To prove this, let $k$ and $h$ be inertial observers such that $\speed_{e}(k) < \mathfrak{c_e}$ and $\speed_{e}(h) < \mathfrak{c_e}$ according to any Ether observer $e$ and let $\bar x \in \Q^4$.  We have to prove that there is a $\vy \in Q^4$ such that $\ev_k[Rad^{-1}_{\bar{v}_k(e)}(\vx)]=\ev_{h}[Rad^{-1}_{\bar{v}_h(e)}(\vy)]$. Let us denote $Rad^{-1}_{\bar{v}_k(e)}(\bar{x})$ by $\bar{x}'$.  $\bar x'$ exists since $Rad_{\bar{v}_k(e)}$ is a well-defined bijection. There is a $\bar y'$ such that $\ev_k\left(\bar{x}'\right)=\ev_{h}\left(\bar{y}'\right)$ because of \ax{AxEv^{CK}}. Then  $\bar y=Rad_{\bar{v}_h(e)}(\bar{y}')$ has the requited properties.

\item  $\sy{\mathsf{ClassicalKin}_{Full}} \vdash Tr(\ax{AxSelf^{SR}})$. The translation of $\ax{AxSelf^{SR}}$ is equivalent to
\begin{multline*}
\bforall {k} { IOb}  \bforall {e} {\Ether}  
\Big( \speed_{e}(k) < \mathfrak{c_e} \\\to
\Bforall {\bar{y}}{\Q^4} \big[ W\big((k,k,Rad^{-1}_{\bar{v}_k(e)}(\bar{y})\big) \leftrightarrow y_1=y_2=y_3=0 \big] \Big).
\end{multline*}

To prove the formula above, let $k$ be an inertial observer such that $\speed_{e}(k) < \mathfrak{c_e}$ according to any Ether observer $e$ and let $\bar y \in \Q^4$. We have to prove that $W\big((k,k,Rad^{-1}_{\bar{v}_k(e)}(\bar{y})\big)$ if and only if $y_1=y_2=y_3=0$. Let $\bar x \in \Q^4$ be such that $Rad_{\bar{v}_k(e)}(\bar{x})=\bar y$. By \ax{AxSelf^{CK}}, $W\big((k,k,\bar x)\big)$ if and only if $x_1=x_2=x_3=0$. This holds if and only if $y_1=y_2=y_3=0$ since by item \ref{tx} of Lemma \ref{thm-lemma} $Rad_{\bar{v}}$ transformation maps the time axis on the time axis.

\item $\sy{\mathsf{ClassicalKin}_{Full}} \vdash Tr(\ax{AxSymD^{SR}})$. The translation of $\ax{AxSymD}$ is equivalent to:
\begin{multline*}
\bforall {k,k'}{ IOb}\Bforall{\vx,\vy,\vx',\vy'}{Q^4}\bforall{e}{Ether} \\
\vast(\fblock{
\speed_{e}(k) < \mathfrak{c_e} \AND \speed_{e}(k') < \mathfrak{c_e}\big) \\
time(\bar x,\bar y) = time(\bar x', \bar y') = 0 \\ 
ev_k(Rad^{-1}_{\bar{v}_k(e)}(\vx)) = \ev_{k'}(Rad^{-1}_{\bar{v}_{k'}(e)}(\vx')) \\
ev_k(Rad^{-1}_{\bar{v}_k(e)}(\vy)) = \ev_{k'}(Rad^{-1}_{\bar{v}_{k'}(e)}(\vy')) }
\to  space(\vx,\vy)= space(\vx',\vy')\vast).
\end{multline*}

Let $k$ and $k'$ be inertial observers, let $\bar x$, $\bar y$, $\bar x'$, and $\bar y'$ be coordinate points, and let $e$ be an ether observer such that $\speed_{e}(k)<\mathfrak{c_e}$, $\speed_{e}(k') < \mathfrak{c_e}$, $time(\bar x,\bar y) = time(\bar x', \bar y') = 0$, $ev_{k}(Rad^{-1}_{\bar{v}_k(e)}(\vx)) = \ev_{k'}(Rad^{-1}_{\bar{v}_{k'}(e)}(\vx'))$, and $ev_k(Rad^{-1}_{\bar{v}_k(e)}(\vy)) = \ev_{k'}(Rad^{-1}_{\bar{v}_{k'}(e)}(\vy'))$. Let $P=Rad_{v_{k'}(e)}\circ w_{kk'}\circ Rad^{-1}_{v_k(e)}$. By Lemma~\ref{lemma-cannon}, $P$ is a Poincar\'e transformation. By the assumptions, $P(\bar x)=\bar x'$ and  $P(\bar y)=\bar y'$. Therefore,  $time(\vx,\vy)^2- space(\vx,\vy)^2= time(\vx',\vy')^2- space(\vx',\vy')^2$. Since both $time(\bar x,\bar y)$ and $time(\bar x', \bar y')$ are zero, $space(\vx,\vy)^2=space(\vx',\vy')^2$. Consequently, $space(\vx,\vy)= space(\vx',\vy')$ because they are both positive quantities.

\item $\sy{\mathsf{ClassicalKin}_{Full}} \vdash Tr(\ax{AxLine^{SR}})$. The translation of $\ax{AxLine^{SR}}$ is equivalent to:
\begin{multline*}
\bforall {k,h}{ IOb} 
\Bforall {\vx,\vy,\vz}{\Q^4} 
\Bforall {e}{\Ether} 
\Big[ \\
\fblock{
 \speed_{e}(k) < \mathfrak{c_e} \\
 \speed_{e}(h) < \mathfrak{c_e} \\
Rad^{-1}_{\bar{v}_k(e)}(\vx),Rad^{-1}_{\bar{v}_k(e)}(\vy),Rad^{-1}_{\bar{v}_k(e)}(\vz)\in \wl_k(h)}\\ 
\to \bexists {a}{Q} \big[ \vz-\vx=a(\vy-\vx) \lor
\vy-\vz=a(\vz-\vx)\big]\Big]. 
\end{multline*}

Because of $\ax{AxLine^{CK}}$, $Rad^{-1}_{\bar{v}_k(e)}(\bar x)$, $Rad^{-1}_{\bar{v}_k(e)}(\bar y)$ and $Rad^{-1}_{\bar{v}_k(e)}(\bar z)$ are on a straight line. Since $Rad_{\bar v}$ is a linear map, $\bar x$, $\bar y$ and $\bar z$ are on a straight line, hence the translation of $\ax{AxLine^{SR}}$ follows.

\item $\sy{\mathsf{ClassicalKin}_{Full}} \vdash  Tr(\ax{AxTriv^{SR}})$. The translation of $\ax{AxTriv^{SR}}$ is equivalent to:
\begin{multline*}
\bforall{T}{\Triv} 
\bforall {h}  {IOb} 
\bforall {e}{Ether}
\BIG(\speed_{e}(h) < \mathfrak{c_e} \\ \to  
\bexists {k} {IOb}
\fblock{
speed_{e}(k) < \mathfrak{c_e} \\
Rad_{\bar v_{k}(e)}\circ w_{hk} \circ Rad^{-1}_{\bar v_{h}(e)}=T}
\BIG).
\end{multline*}

To prove $Tr(\ax{AxTriv^{SR}})$, we have to find a slower-than-light  inertial observer $k$  for every trivial transformation $T$ and a slower-than-light inertial observer $h$ such that $Rad_{\bar v_{k}(e)}\circ w_{hk} \circ Rad^{-1}_{\bar v_{h}(e)}=T$. 

By $\ax{AxTriv^{CK}}$ and  Lemma~\ref{triv-poi}, there is an inertial observer $k$ such that $Rad_{\bar v_{k}(e)}\circ w^{CK}_{hk} \circ Rad^{-1}_{\bar v_{h}(e)}=T$. This $k$ is also slower-than-light since $w^{CK}_{hk}$ is a trivial transformation.

\item  $\sy{\mathsf{ClassicalKin}_{Full}} \vdash Tr(\ax{AxPh_c^{SR}})$. The translation of $\ax{AxPh_c^{SR}}$ is equivalent to:
\begin{multline*}
\bexists {c}{Q}\BIGG( c>0
\AND 
\bforall {k} {IOb} \Bforall {\vx,\vy}{Q^4} \bforall {e}{\Ether} 
\BIGG[\speed_{e}(k) < \mathfrak{c_e}  
\to \\ \bexists {p}{\Ph} \BIG( 
\fblock{ W^{CK}\big(k,p,Rad^{-1}_{\bar{v}_k(e)}(\vx)\big) \\
W^{CK}\big(k,p,Rad^{-1}_{\bar{v}_k(e)}(\vy)\big)} 
\leftrightarrow 
 space(\vx,\vy)= c\cdot time(\vx,\vy)\BIG)\BIGG]\BIGG) .
\end{multline*}
It is enough to show that any slower-than-light inertial observers $k$ can send a light signal trough coordinate points $\bar x'$ and $\bar y'$ exactly if 
\begin{equation*}
space\big(Rad_{\bar{v}_k(e)}(\vx'),Rad_{\bar{v}_k(e)}(\vy')\big)= \mathfrak{c_e}\cdot time\big(Rad_{\bar{v}_k(e)}(\vx'),Rad_{\bar{v}_k(e)}(\vy')\big)
\end{equation*}
holds for any ether observer $e$, i.e., if $Rad_{\bar{v}_k(e)}(\bar x')$ and $Rad_{\bar{v}_k(e)}(\bar y')$ are on a right light cone.  $k$ can send a light signal through coordinate points $\bar x'$ and $\bar y'$ if they are on a light cone moving with velocity $\bar{v}_{k}(e)$ by Theorem~\ref{thm-gal} because Galilean transformation $w_{ek}$ maps worldlines of light signals to worldlines of light signals and light signals move along right light cones according to $e$ by \ax{AxEther}. $Rad_{\bar{v}_k(e)}$ transform these cones into right light cones by Item~\ref{rad-cone} of Lemma~\ref{thm-lemma}.  Therefore, $Rad_{\bar{v}_k(e)}(\bar x')$ and $Rad_{\bar{v}_k(e)}(\bar y')$ are on a right light cone if $k$ can send a light signal trough $\bar x'$ and $\bar y'$, and this is what we wanted to show.

\item $\sy{\mathsf{ClassicalKin}_{Full}} \vdash Tr(\ax{AxThExp^{SR}})$. $Tr(\ax{AxThExp^{SR}})$ is equivalent to:
\begin{multline*}
\bexists {h}{ IOb} \bforall {e}{\Ether}  
\big[ \speed_{e}(h) < \mathfrak{c_e} \big]
\AND \\ \bforall {k}{ IOb} \Bforall{\vx,\vy}{Q^4} \bforall {e}{\Ether} 
\BIG(
\fblock{\speed_{e}(k) < \mathfrak{c_e} \\
 space(\vx,\vy)< \mathfrak{c_e} \cdot time(\vx,\vy)}
\\ \to 
\bexists {k'}{ IOb} 
\fblock{ \speed_{e}(k') < \mathfrak{c_e}
\\
x,y \in wl_k(k')}
\BIG).
\end{multline*}
The first conjunct of the translation follows immediately from $\ax{AxEther}$. Let us now prove the second conjunct. From \ax{AxThExp_+} we get inertial observers both inside and outside of the light cones. Those observers which are inside of the light cone (which are the ones we are interested in) stay inside the light cone by the translation by Items \ref{rad-time} and \ref{rad-cone} of Lemma~\ref{thm-lemma}. Since we have only used that there are observers on every straight line inside of the light cones, this proof of $Tr(\ax{AxThExp^{SR})}$ goes trough also for the $NoFTL$ case, needed in Theorem \ref{thm-tr+} below.

\item $\sy{ClassicalKin^{STL}_{Full}} \vdash Tr(\ax{AxNoAcc})$. The translation of $\ax{AxNoAcc}$ is equivalent to:
\begin{multline*}
\bforall{k}{B} \bexists{\bar x}{Q^4} \bexists{b}{B} \bforall {e} {\Ether} \\
\big[W^{CK}\big(k,b,Rad^{-1}_{\bar{v}_k(e)}(\bar{x})\big)
\rightarrow \big(IOb^{CK}(k) \AND [ \speed_{e}(k) < \mathfrak{c_e}]\big)\big].
\end{multline*}
which follows directly from $\ax{AxNoAcc}$ since $Rad_{\bar{v}_k(e)}$ is the same bijection for all ether observers $e$, and from $\ax{AxNoFTL}$. 
\qedhere
\end{itemize}
\end{proof}

While our translation function translates axioms of special relativity theory into theorems of classical kinematics, models are transformed the other way round from classical mechanics to special relativity theory. If we take any model $\mathfrak{M}_{CK}$ of \sy{ClassicalKin_{Full}}, then our translation $Tr$ tells us how to understand the basic concepts of \sy{SpecRel_{Full}} in $\mathfrak{M}_{CK}$ in such a way that they satisfy the axioms of \sy{SpecRel_{Full}}, turning the model $\mathfrak{M}_{CK}$ into a model $\mathfrak{M}_{SR}$ of \sy{SpecRel_{Full}}. For an illustration on how a model is transformed by our translation, see the discussion of the Michelson--Morley experiment in the next section.

Let us now show that there is no inverse interpretation $Tr'$ of classical kinematics in special relativity theory.

\begin{thm}\label{thm-noinv}
There is no interpretation of \sy{ClassicalKin_{Full}} in \sy{SpecRel_{Full}}.
\end{thm}
\begin{proof}
Assume towards contradiction that there is an interpretation, say  $Tr'$, of \sy{ClassicalKin_{Full}} in \sy{SpecRel_{Full}}. In the same way as the translation $Tr$ turns models of \sy{ClassicalKin_{Full}} into models of \sy{SpecRel_{Full}}, this inverse translation $Tr'$ would turn every model of \sy{SpecRel_{Full}} into a model  of \sy{ClassicalKin_{Full}}. There is a model $\mathfrak{M}$ of \sy{SpecRel_{Full}} such that $B_\mathfrak{M}=IOb_\mathfrak{M}\cup Ph_\mathfrak{M}$ and the automorphism group of $\mathfrak{M}$ acts transitively on both $IOb_\mathfrak{M}$ and $Ph_\mathfrak{M}$, see e.g., the second model constructed in \citep[Thm.2]{superluminal}. That is, for any two inertial observers $k$ and $h$ in $\mathfrak{M}$, there is an automorphism of $\mathfrak{M}$ taking $k$ to $h$, and the same is true for any two light signals in $\mathfrak{M}$.

Let $\mathfrak{M}'$ be the model of \sy{ClassicalKin_{Full}} that $\mathfrak{M}$ is turned into by translating the basic concepts of \sy{ClassicalKin_{Full}} to defined concepts of \sy{SpecRel_{Full}} via $Tr'$. Then every automorphism of $\mathfrak{M}$ is also an automorphism of $\mathfrak{M'}$. Since $Tr'$ has to translate bodies into bodies, there have to be two sets of bodies in $\mathfrak{M'}$ such that any two bodies from the same set can be mapped to each other by an automorphism of $\mathfrak{M'}$. However, in models of \sy{ClassicalKin_{Full}}, observers moving with different speeds relative to the ether frame cannot be mapped to each other by an automorphism. By $\ax{AxThExp_+}$ there are inertial observers in $\mathfrak{M}$ moving with every finite speed. Therefore, there are infinitely many sets of inertial observers which elements cannot be mapped into each other by an automorphism. This is a contradiction showing that no such model $\mathfrak{M'}$ and hence no such translation $Tr'$ can exist.
\end{proof}

\begin{cor}\label{cor-nodefeq}
\sy{SpecRel_{Full}} and \sy{ClassicalKin_{Full}} are not definitionally equivalent. 
\end{cor}

\section{Intermezzo: The Michelson--Morley experiment}

As an illustration, we will now show how the null result of the \citep{MM} experiment behaves under our interpretation. This is a complicated experiment involving interferometry to measure the relative speed between two light signals, but we can make abstraction of that here.

Let, as illustrated on the right side of Figure \ref{fig-MM}, inertial observer $k$ send out two light signals at time $-L$, perpendicular to each other in the $x$ and $y$ direction. Let us assume that observer $k$ sees the ether frame moving with speed $-v$ in the $tx$-plane --- such that we do not have to rotate into the $tx$-plane (i.e. $\bar{v} = (-v, 0, 0)$ so that $R_{\bar v}$ is the identity as defined on page \pageref{core-matrix}). The light signals are reflected by mirrors at $(0,L,0,0)$ and $(0,0,L,0)$, assuming that the speed of light is $1$. In accordance with the null result of the Michelson--Morley experiment, the reflected light signals are both received by observer $k$ at time $L$.

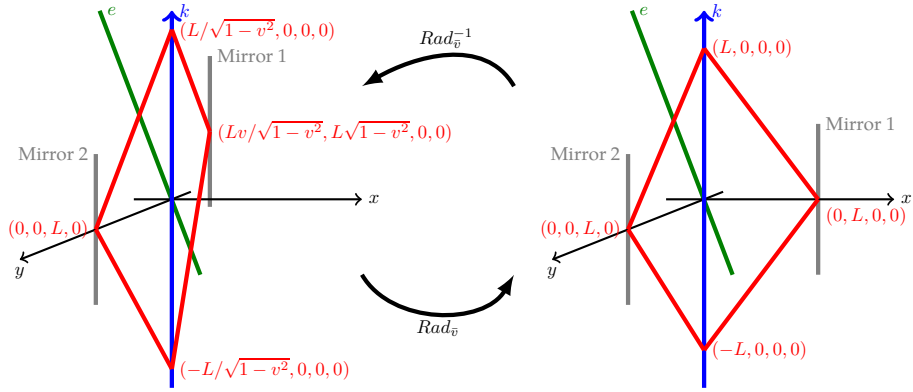
\begin{figure}[ht]
  \begin{center}
    \begin{tikzpicture}[scale=0.5, every node/.style={scale=0.7}]

\begin{scope}[shift={(-22,0)}]
\draw [ultra thick, green!50!black] (20.75,2) -- (18.1,9)  node[right] {$e$};
\draw [gray, ultra thick] (21,3.8) -- (21,7.8)  node[right] {Mirror 1};
\draw [gray, ultra thick] (18,1.2) -- (18,5.2)  node[left] {Mirror 2};
\draw [->,  thick] (19,4)  -- (25,4)  node[right] {$x$};
\draw [<-, thick] (16,2.4)  node[below] {$y$} -- (20.5,4.2);
\draw [->, blue, ultra thick] (20,-1) -- (20,9)  node[right] {$k$};
\draw [red, ultra thick] (20,-0.5) node[right] {$(-L/\sqrt{1-v^2},0,0,0)$} -- (21,5.8) node[right] {$(Lv/\sqrt{1-v^2},L\sqrt{1-v^2},0,0)$};
\draw [red, ultra thick] (20,-0.5) -- (18,3.2) node[left] {$(0,0,L,0)$};
\draw [red, ultra thick] (21,5.8) -- (20,8.5) node[right] {$(L/\sqrt{1-v^2},0,0,0)$} ;
\draw [red, ultra thick] (18,3.2) -- (20,8.5) ;
\end{scope}

\begin{scope}[shift={(-8,0)}]
\draw [ultra thick, green!50!black] (20.75,2) -- (18.1,9)  node[right] {$e$};
\draw [gray, ultra thick] (23,2) -- (23,6)  node[right] {Mirror 1};
\draw [gray, ultra thick] (18,1.2) -- (18,5.2)  node[left] {Mirror 2};
\draw [->, thick] (19,4)  -- (25,4)  node[right] {$x$};
\draw [<-, thick] (16,2.4)  node[below] {$y$} -- (20.5,4.2);
\draw [->, blue, ultra thick] (20,-1) -- (20,9)  node[right] {$k$};
\draw [red, ultra thick] (20,0) node[right] {$(-L,0,0,0)$} -- (23,4) node[below right] {$(0,L,0,0)$};
\draw [red, ultra thick] (20,0) -- (18,3.2) node[left] {$(0,0,L,0)$};
\draw [red, ultra thick] (23,4) -- (20,8) node[right] {$(L,0,0,0)$} ;
\draw [red, ultra thick] (18,3.2) -- (20,8) ;
\end{scope}

\draw[->,ultra thick,>=latex] (7,7)  to [out=120,in=30] node[above]{$Rad^{-1}_{\bar{v}}$} (3,7);
\draw[->,ultra thick,>=latex] (3,2)  to [out=-60,in=-120] node[below]{$Rad_{\bar{v}}$} (7,2);

\end{tikzpicture}
    \caption{\label{fig-MM}  On the left we have the classical setup which is transformed by $Rad_{\bar v}$ into the setup for the Michelson--Morley experiment on the right.}
  \end{center}
\end{figure}

To understand how this typical relativistic setup can be modeled in classical kinematics, we have to find which classical setup is being transformed into it. So we need the inverse $Rad^{-1}_{\bar v}$ of the radarization, which since the ether is in the $tx$-plane is just the inverse $C_v^{-1}$ of the core map $C_v$ as defined on page \pageref{core-matrix}. 

Using the coordinates found by multiplying\footnote{See \citep[p.44]{diss} for the calculations.} the matrix $C_v^{-1}$ by the coordinates from the right side of the figure, we can draw the left side of the figure, which is the classical setup translated by $Rad_{\bar v}$ into our setup of the Michelson--Morley experiment. On the left side of the figure, we see that the speed of light is not the same in every direction according to observer $k$. This illustrates that our translation does not preserve simultaneity (the time at which the light beams hit the mirrors are not the same anymore), speed, distance and time difference.

\section{Definitional Equivalence}
We will now slightly modify our axiom systems \ax{SpecRel_{Full}} and \ax{ClassicalKin_{Full}} to establish a definitional equivalence between them. This will provide us with an insight about the exact differences between special relativity and classical kinematics.

To make classical kinematics equivalent to special relativity, we ban the inertial observers that are not moving slower than light relative to the ether frame. This is done by the next axiom.
\begin{description}
\item[\underline{\ax{AxNoFTL}}]
All inertial observers move slower than light with respect to the ether frames:
\begin{equation*}
\bforall {k} {IOb} 
\bforall {e} {Ether} \big[\speed_{e}(k) < \mathfrak{c_e}\big].
\end{equation*}
\end{description}

Axiom \ax{AxNoFTL} contradicts \ax{AxThExp_+}. Therefore, we replace \ax{AxThExp_+} with the following weaker assumption.
\begin{description}
\item[\underline{\ax{AxThExp^{STL}}}]
Inertial observers can move with any speed which is in the ether frame slower than that of light:
\begin{multline*}
\bexists {h}{B}  [ IOb(h)] \AND 
\bforall {e} {\Ether}
\Bforall {\vx,\vy} {\Q^4}\\
\big( space(\vx,\vy)<\mathfrak{c_e}\cdot time(\vx,\vy)
 \to \bexists {k}  {IOb} \big[\vx,\vy\in \wl_e(k)\big]\big).
\end{multline*}
\end{description}
Let us now introduce our axiom system \sy{ClassicalKin_{Full}^{STL}} as follows:
\[\sy{ClassicalKin_{Full}^{STL}}\de \sy{ClassicalKin_{Full}} \cup \{\ax{AxNoFTL},\ax{AxThExp^{STL}}\}\setminus \{\ax{AxThExp_+}\}.\]

To make special relativity equivalent to classical kinematics, we have to introduce a class of observers which will play the role of Ether, for which we have to extend our language with a unary relation $E$ to
\begin{equation*}
\{\, \B,\Q\,;  IOb, \Ph, E, +,\cdot,\le,\W\,\}.
\end{equation*}
We call the set defined by $E$ the \textit{primitive ether frame}, and its elements \textit{primitive ether observers}. They are \textit{primitive} in the sense that they are concepts which are solely introduced by an axiom  without definition:
\begin{description}
\item[\underline{\ax{AxPrimitiveEther}}]
There is a non-empty class of distinguished observers, stationary with respect to each other, which is closed under trivial transformations:
\begin{equation*}
  \bexists {e}{IOb} \BIGG[ \bforall {k}{B} \BIGG( 
    \fblock{ IOb(k) \\
      \bexists {T}  {\Triv} 
      \w_{e k}^{SR}=T } \leftrightarrow E(k) \BIGG) \BIGG].
\end{equation*}
\end{description}
Let us now introduce our axiom system \sy{SpecRel_{Full}^{e}} as follows:
\[\sy{SpecRel_{Full}^{e}}\de\sy{SpecRel_{Full}}\cup\{\ax{AxPrimitiveEther}\}.\]

Let the translation $Tr_+$ from \sy{SpecRel_{Full}^{e}} to \sy{ClassicalKin_{Full}^{STL}} be the translation that restricted to \sy{SpecRel_{Full}} is $Tr$:
\begin{equation*}
  Tr_{+}|_{\mathsf{SpecRel_{Full}}} \de Tr ,
\end{equation*}
while the translation of the primitive ether is the classical ether:
\begin{equation*}
  Tr_{+} \big(E(x) \big) \defeq Ether(x).
\end{equation*}
Let the inverse translation $Tr_+'$ from \sy{ClassicalKin_{Full}^{STL}} to \sy{SpecRel_{Full}^{e}} be the translation that translates the basic concepts as:
\[Tr_+' ( a+b=c ) \defeq (a+b=c),\ 
Tr_+' ( a \cdot b=c ) \defeq (a \cdot b=c),\ 
Tr_+' ( a < b) \defeq a < b,\]
\[Tr_+' \big(Ph^{CK}(p)\big)\defeq Ph^{SR}(p),\quad 
Tr_+' \big(IOb^{CK}(b)\big)\defeq IOb^{SR}(b),\]
\[Tr_+' \big( W^{CK}(k,b,\bar{x}) \big)
\defeq \bforall{e}{E} \big[W^{SR}\big(k,b,Rad_{\bar{v}_k(e)}(\bar{x})\big)\big].\]

By Theorem~\ref{thm-defeq}, the two slight modifications above are enough to make classical kinematics and special relativity definitionally equivalent.

\begin{lem}\label{lem-Tr_+'(Ether)} Assume \sy{SpecRel_{Full}^{e}}. Then $Tr_+'[Ether(b)] \equiv E(b)$.
\end{lem}
\begin{proof}
By the definition of $Ether$, $Tr_+'[\Ether(b)]$ is equivalent to 
\[ Tr_+' \begin{pmatrix} IOb(b)\AND \bexists {c} {\Q} \Big[c>0\AND \Bforall {\vx,\vy}{\Q^4} \\ \Big(\bexists {p} {\Ph} \big(\big[\vx,\vy\in\wl_{b}(p) \big] \leftrightarrow  space(\vx,\vy)=c\cdot time(\vx,\vy)\big)\Big)\Big]
 \end{pmatrix}, \]
which is by the defintion of $Tr_+'$ equivalent to
\begin{multline*}
IOb(b) \AND \bexists {c} {\Q} \Big[c>0\AND \Bforall {\vx,\vy}{\Q^4} \\ 
 \Big(\bexists {p} {\Ph} \big(Tr_+'\big[[\vx,\vy\in \wl_{b}(p)] \big] \leftrightarrow  space(\vx,\vy)=c\cdot time(\vx,\vy)\big)\Big)\Big],
\end{multline*}
which by the definition of worldlines is equivalent to 
\begin{multline*}
IOb(b) \AND \bexists {c} {\Q} \Big[c>0\AND \Bforall {\vx,\vy}{\Q^4} \\ 
  \Big(\bexists {p} {\Ph} \big(Tr_+'\big[\W(b,p,\vx) \AND \W(b,p,\vy) \big] \leftrightarrow  space(\vx,\vy)=c\cdot time(\vx,\vy)\big)\Big)\Big],
\end{multline*}
which by the definition of $Tr_+'(W)$ is equivalent to
\begin{multline*}
IOb(b) \AND \bexists {c} {\Q} \Big[c>0\AND \Bforall {\vx,\vy}{\Q^4} \\
  \Big(\bexists {p} {\Ph} \big(\big[\bforall{e}{E} \big(W^{SR}[b,p,Rad_{\bar{v}_b(e)}(\bar{x})]\big)\\ \AND \bforall{e}{E} \big(W^{SR}[b,p,Rad_{\bar{v}_b(e)}(\bar{y})]\big) \big]  \leftrightarrow  space(\vx,\vy)=c\cdot time(\vx,\vy)\big)\Big)\Big],
\end{multline*}
which by the definition of worldlines is equivalent to
\begin{multline*}
IOb(b) \AND \bexists {c} {\Q} \Big[c>0\AND \Bforall {\vx,\vy}{\Q^4}  
 \Big(\bexists {p} {\Ph} \big(\bforall{e}{E}\\ [Rad_{\bar{v}_b(e)}(\bar{x}),Rad_{\bar{v}_b(e)}(\bar{y})\in \wl_{b}(p)]  \leftrightarrow  space(\vx,\vy)=c\cdot time(\vx,\vy)\big)\Big)\Big].
\end{multline*}
This means that $b$ is an inertial observer which, after transforming its worldview by $Rad_{\bar{v}(e)}$, sees the light signals moving along right light cones. Since the light cones are already right ones in \sy{SpecRel_{Full}^{e}} and $Rad_{\bar{v}_b(e)}$ would tilt light cones if $\bar{v}_b(e) \neq (0,0,0)$ because of Item \ref{rad-0} of Lemma \ref{thm-lemma} and because $Rad_{\bar{v}_b(e)}$ is a bijection, $b$ must be stationary relatively to primitive ether $e$, and hence the above is equivalent to $E(b)$. 
\end{proof}

\begin{cor}\label{cor-Tr_+'(c_e)} Assume \sy{SpecRel_{Full}^{e}}. Then $Tr_+'(\mathfrak{c_e}) \equiv \mathfrak{c}$.
\end{cor}
\begin{proof}
Since the ether is being translated into the primitive ether, the speed of light $\mathfrak{c_e}$ in the ether frame is translated to the speed of light of the primitive ether. In \ax{SpecRel^{e}_{Full}} the speed of light $\mathfrak{c}$ is the same for all observers.
\end{proof}

\begin{thm}\label{thm-tr+}
$Tr_+$ is an interpretation of \sy{SpecRel^e_{Full}} in \sy{ClassicalKin^{STL}_{Full}}, i.e.,
\begin{equation*}
\sy{ClassicalKin^{STL}_{Full}}\vdash Tr_+(\varphi)\enskip\text{ if }\enskip \sy{SpecRel^e_{Full}}\vdash \varphi.
\end{equation*}
\end{thm}
\begin{proof}

The proof of this theorem is basically the same as that of Theorem \ref{thm-tr}. The only differences are the proof of $Tr_+(\ax{AxThExp})$ as $\ax{AxThExp_+}$ is replaced by $\ax{AxThExp^{STL}}$, and that we have to prove $Tr_+(\ax{AxPrimitiveEther})$.

\begin{itemize}[leftmargin=*]
\item The proof for $Tr(\ax{AxThExp})$ goes through since we only use that there are observers on every straight line inside of the light cones, which is covered by $\ax{AxThExp^{STL}}$.
\item $\sy{ClassicalKin^{STL}_{Full}} \vdash Tr_+(\ax{AxPrimitiveEther})$.  $Tr_+(\ax{AxPrimitiveEther})$ is
\begin{multline*}
Tr_+\Big[\bexists {e}{IOb} \big[ \bforall {k}{B} \big( [ IOb(k) \AND 
\bexists {T}  {\Triv} 
\w_{e k}^{SR}=T ] \leftrightarrow E(k) \big) \big] \Big],
\end{multline*}
which by the previously established translations of $IOb$, $w$ and $Ether$, and by 
using the result of the Appendix is equivalent to 
\begin{multline*}
\bexists {e}{IOb} \bforall {e'}{Ether} \Big[(speed_{e'}(e) < \mathfrak{c_e}) \AND \\
\Big( \bforall {k}{B} \big[ \big( IOb(k) \AND [speed_{e'}(k) < \mathfrak{c_e}] \AND 
\bexists {T}  {\Triv} \Bforall {\vx,\vy}{Q^4} \\
\w_{ek}^{CK}(Rad^{-1}_{\bar{v}_e(e')}(\bar x), Rad^{-1}_{\bar{v}_k(e')}(\bar y))=T(\vx,\vy) \big) \leftrightarrow Ether(k) \big] \Big) \Big].
\end{multline*}
The $\w_{ek}^{CK}(Rad^{-1}_{\bar{v}_e(e')}(\bar x), Rad^{-1}_{\bar{v}_k(e')}(\bar y))=T(\vx,\vy)$ part in the above translation can be written as $Rad_{\bar{v}_k(e')} \circ \w_{ek}^{CK} \circ Rad^{-1}_{\bar{v}_e(e')} = T$, from which, by Item 5 of Lemma \ref{thm-lemma}, follows that there is a trivial transformation $T'$ such that $\w_{ek}^{CK} = Rad^{-1}_{\bar{v}_k(e')} \circ T \circ Rad_{\bar{v}_e(e')} = T'$. So $Tr_+(\ax{AxPrimitiveEther})$ is equivalent to
\begin{multline*}
\bexists {e}{IOb} \bforall {e'}{Ether} \Big[[speed_{e'}(e) < \mathfrak{c_e}] \AND \\
\big( \bforall {k}{B} \big[ \big( IOb(k) \AND [speed_{e'}(k) < \mathfrak{c_e}] \AND \\
\bexists {T'}  {\Triv} \Bforall {\vx,\vy}{Q^4} 
\w_{ek}^{CK}(\bar x,\bar y)=T'(\vx,\vy) \big) \leftrightarrow Ether(k) \big] \big) \Big].
\end{multline*}
To prove this from \sy{ClassicalKin^{STL}_{Full}}, let $e$ be an ether observer. Then $IOb(e)$ holds and $speed_{e'}(e) = 0 < \mathfrak{c_e}$ for every $e'\in Ether$. Therefore, we only have to prove that $k$ is an ether observer if and only if it is a slower-than-light inertial observer such that $w^{CK}_{ek}$ is some trivial transformation.

If $k$ is an ether observer, then $k$ is a slower-than-light inertial observer and $w^{CK}_{ek}$ is indeed a trivial transformation because it is a Galilean transformation between two ether observers which are stationary relative to each other.
The other direction of the proof is that if $k$ is an inertial observer which transforms to an ether observer by a trivial transformation, then $k$ is itself an ether observer  because then $k$ also sees the light cones right.
\qedhere
\end{itemize}
\end{proof}

\begin{lem}\label{lemma-mosquito}   
Assuming \sy{{SpecRel}^{e}_{Full}}, if $e$ and $e'$ are primitive ether observers and $k$ and $h$ are inertial observers, then
\begin{equation*}
Rad^{-1}_{\bar{v}_{h}(e)}\circ w^{SR}_{kh}\circ Rad_{\bar{v}_{k}(e)}= Rad^{-1}_{\bar{v}_{h}(e')}\circ w^{SR}_{kh}\circ Rad_{\bar{v}_{k}(e')}
\end{equation*}
and it is a Galilean transformation.

\end{lem}
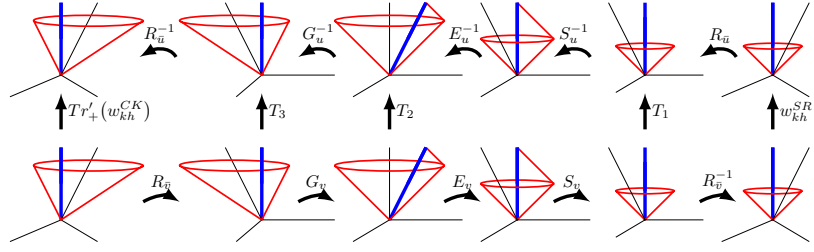
\begin{figure}
  \begin{center}
    \scalebox{.8}{\begin{tikzpicture}[scale=0.3, every node/.style={scale=0.8}]
\pgfmathsetmacro\size{4}
\pgfmathsetmacro\v{0.5}
\pgfmathsetmacro\r{0.13}
\pgfmathsetmacro\scaling{1.1/(1-\v*\v)}

\tikzstyle{T0}=[cm={1,0,-\v,1,(-9,0)}]
\tikzstyle{T5}=[cm={1,0,-\v,1,(2,0)}]
\tikzstyle{T1}=[cm={1,0,0,1,(9,0)}]
\tikzstyle{T2}=[cm={1,0,0,1,(16,0)}]
\tikzstyle{T3}=[cm={0.8,0,0,0.8,(23,0)}]
\tikzstyle{T4}=[cm={1,0,0,1,(30,0)}]
\tikzstyle{T10}=[cm={1,0,-\v,1,(-9,-8)}]
\tikzstyle{T15}=[cm={1,0,-\v,1,(2,-8)}]
\tikzstyle{T11}=[cm={1,0,0,1,(9,-8)}]
\tikzstyle{T12}=[cm={1,0,0,1,(16,-8)}]
\tikzstyle{T13}=[cm={0.8,0,0,0.8,(23,-8)}]
\tikzstyle{T14}=[cm={1,0,0,1,(30,-8)}]

\newcommand{\cone}[3]{\draw[thick, red,shift={(#2,#3)}] (0,#1) ellipse (#1 and 0.1*#1) (-#1,#1)--(0,0)--(#1,#1);}

\newcommand{\conemirror}[3]{\draw[thick, red,shift={(#2,#3)}] (3,#1) ellipse (#1 and 0.1*#1) (-#1+3,#1)--(0,0)--(#1+3,#1);} 

\pgfmathsetmacro\Xx{0.5*\size}
\pgfmathsetmacro\Xy{-.3*\size}
\pgfmathsetmacro\Yx{-.7*\size}
\pgfmathsetmacro\Yy{-.3*\size}

\begin{scope}[T0]
\draw[ultra thick, blue] (0,0) to (\v*\size,\size) ;
\draw (4,\size)  -- (0,0)  ;
\conemirror{0.75*\size}{0}{0}
\draw[ultra thick, blue] (0.5*\v*\size,0.5*\size) to (\v*\size,\size);
\end{scope}

\draw[shift={(-8,-8)}] (\Xx-1,\Xy) -- (-1,0) -- (\Yx-1,\Yy );

\draw[->,ultra thick,>=latex] (-2.7,1)  to [out=120,in=30] node[above]{$R^{-1}_{\bar{u}}$} (-4.7,1);

\begin{scope}[T5]
\draw[ultra thick, blue] (0,0) to (\v*\size,\size)  ;
\draw (0,\size)  -- (0,0) -- (\size,0) (0,0) -- (-.5*\size,-.3*\size);
\cone{0.75*\size}{0}{0} 
\draw[ultra thick, blue] (0.5*\v*\size,0.5*\size) to (\v*\size,\size);

\end{scope}

\draw[->,ultra thick,>=latex] (6,1)  to [out=120,in=30] node[above]{$G^{-1}_u$} (4,1);

\begin{scope}[T1]
\draw[ultra thick, blue] (0,0) to (\v*\size,\size)  ;
\draw (0,\size)  -- (0,0) -- (\size,0) (0,0) -- (-.5*\size,-.3*\size);
\draw[thick, red] (\size*0.75,\size*0.75) to (0.5*\size,\size) ;
\cone{0.75*\size}{0}{0} 
\draw[ultra thick, blue] (0.5*\v*\size,0.5*\size) to (\v*\size,\size);

\end{scope}

\draw[->,ultra thick,>=latex] (14,1)  to [out=120,in=30] node[above]{$E^{-1}_{u}$} (12,1);

\begin{scope}[T2]
\draw[ultra thick, blue] (0,0) to (0,\size) ;
\cone{0.5*\size}{0}{0} 
\draw (-\v*\size,\size)  -- (0,0) -- (\size,0) (0,0) -- (-.5*\size,-.3*\size);
\draw[thick, red] (\size/2,\size/2) to (0,\size) ;
\draw[ultra thick, blue] (0,0.5*\size) to (0,\size);
\end{scope}

\draw[->,ultra thick,>=latex] (20,1)  to [out=120,in=30] node[above]{$S^{-1}_{u}$} (18,1);

\begin{scope}[T3]
\draw[ultra thick, blue] (0,0) to (0,\size*1.25) ;
\cone{0.5*\size}{0}{0} 
\draw (-\v*\size*1.25,\size*1.25)  -- (0,0) -- (\size,0) (0,0) -- (-.5*\size,-.3*\size);
\draw[ultra thick, blue] (0,0.5*\size) to (0,\size);
\end{scope}

\draw[->,ultra thick,>=latex] (28,1)  to [out=120,in=30] node[above]{$R_{\bar u}$} (26,1);

\begin{scope}[T4]
\draw[ultra thick, blue] (0,0) to (0,\size) ;
\cone{0.4*\size}{0}{0} 
\draw (-\v*\size+4,\size)  -- (0,0) 
 (\Xx,\Xy) -- (0,0) -- (\Yx,\Yy);
\draw[ultra thick, blue] (0,0.4*\size) to (0,\size);
\end{scope}

\draw[<-,ultra thick,>=latex]  (-9,-1) -- node[right]  {$Tr_+'\big(w^{CK}_{kh}\big)$} (-9,-3);
\draw[<-,ultra thick,>=latex]  (30,-1) -- node[right]  {$w^{SR}_{kh}$} (30,-3);
\draw[<-,ultra thick,>=latex]  (2,-1) -- node[right]  {$T_3$} (2,-3);
\draw[<-,ultra thick,>=latex]  (9,-1) -- node[right]  {$T_2$} (9,-3);
\draw[<-,ultra thick,>=latex]  (23,-1) -- node[right]  {$T_1$} (23,-3);

\begin{scope}[T10]
\draw[ultra thick, blue] (0,0) to (\v*\size,\size) ;
\draw (4,\size)  -- (0,0)  ;
\conemirror{0.75*\size}{0}{0}
\draw[ultra thick, blue] (0.5*\v*\size,0.5*\size) to (\v*\size,\size);
\end{scope}

\draw[shift={(-8,0)}] (\Xx-1,\Xy) -- (-1,0) -- (\Yx-1,\Yy );

\draw[->,ultra thick,>=latex] (-4.5,-7)  to [out=30,in=150] node[above]{$R_{\bar{v}}$} (-2.5,-7);

\begin{scope}[T15]
\draw[ultra thick, blue] (0,0) to (\v*\size,\size)  ;
\draw (0,\size)  -- (0,0) -- (\size,0) (0,0) -- (-.5*\size,-.3*\size);
\cone{0.75*\size}{0}{0} 
\draw[ultra thick, blue] (0.5*\v*\size,0.5*\size) to (\v*\size,\size);

\end{scope}

\draw[->,ultra thick,>=latex] (4,-7)  to [out=30,in=150] node[above]{$G_v$} (6,-7);

\begin{scope}[T11]
\draw[ultra thick, blue] (0,0) to (\v*\size,\size)  ;
\draw (0,\size)  -- (0,0) -- (\size,0) (0,0) -- (-.5*\size,-.3*\size);
\draw[thick, red] (\size*0.75,\size*0.75) to (0.5*\size,\size) ;
\cone{0.75*\size}{0}{0} 
\draw[ultra thick, blue] (0.5*\v*\size,0.5*\size) to (\v*\size,\size);

\end{scope}

\draw[->,ultra thick,>=latex] (12,-7)  to   [out=30,in=150] node[above]{$E_v$} (14,-7);

\begin{scope}[T12]
\draw[ultra thick, blue] (0,0) to (0,\size) ;
\cone{0.5*\size}{0}{0} 
\draw (-\v*\size,\size)  -- (0,0) -- (\size,0) (0,0) -- (-.5*\size,-.3*\size);
\draw[thick, red] (\size/2,\size/2) to (0,\size) ;
\draw[ultra thick, blue] (0,0.5*\size) to (0,\size);
\end{scope}

\draw[->,ultra thick,>=latex] (18,-7)  to [out=30,in=150] node[above]{$S_{v}$} (20,-7);

\begin{scope}[T13]
\draw[ultra thick, blue] (0,0) to (0,\size*1.25) ;
\cone{0.5*\size}{0}{0} 
\draw (-\v*\size*1.25,\size*1.25)  -- (0,0) -- (\size,0) (0,0) -- (-.5*\size,-.3*\size);
\draw[ultra thick, blue] (0,0.5*\size) to (0,\size);
\end{scope}

\draw[->,ultra thick,>=latex] (26,-7)  to [out=30,in=150] node[above]{$R^{-1}_{\bar{v}}$} (28,-7);

\begin{scope}[T14]
\draw[ultra thick, blue] (0,0) to (0,\size) ;
\cone{0.4*\size}{0}{0} 
\draw (-\v*\size+4,\size)  -- (0,0) 
 (\Xx,\Xy) -- (0,0) -- (\Yx,\Yy);
\draw[ultra thick, blue] (0,0.4*\size) to (0,\size);
\end{scope}


\end{tikzpicture}}
    \caption{\label{fig-mosquito}Lemma \ref{lemma-mosquito}: Read the figure starting in the bottom-left corner and follow the arrows to the right along the components of $Rad_{\bar v}$, up along Poincar{\' e} transformation $w^{SR}_{kh}$ and left along the components of $Rad^{-1}_{\bar u}$, which results in Galilean transformation $Tr'_+(w^{CK}_{kh})$.}
  \end{center}
\end{figure}
\begin{proof}The proof is analogous to that for Lemma \ref{lemma-cannon}, see \citep[p.50]{diss} for the full proof.
\end{proof}

\begin{lem}\label{triv-gal}
Assume  \sy{SpecRel^e_{Full}}.  Let $e$ be a primitive ether observer, and let $k$ and $h$ be inertial observers. Assume that $w^{SR}_{hk}$ is a trivial transformation which is translation by vector $\bar{z}$ after linear trivial transformation $T$. Then $Rad^{-1}_{\bar{v}_k(e)}\circ w^{SR}_{hk}\circ Rad_{\bar{v}_h(e)}$ is the trivial transformation which is translation by vector $Rad^{-1}_{\bar{v}_k(e)}(\bar{z})$ after $T$.
\end{lem}
\begin{proof}
The proof is analogous to that for Lemma \ref{triv-poi}, see \citep[p.51]{diss} for the full proof.
\end{proof}

\begin{thm}\label{thm-tr'+}
$Tr'_+$ is an interpretation of \sy{ClassicalKin^{STL}_{Full}} in \sy{SpecRel^e_{Full}}, i.e.,
\begin{equation*}
\sy{SpecRel^e_{Full}}\vdash Tr'_+(\varphi) \enskip\text{ if }\enskip \sy{ClassicalKin^{STL}_{Full}}\vdash \varphi.
\end{equation*}
\end{thm}
\begin{proof}

\noindent
\begin{itemize}[leftmargin=*] 
\item $\ax{EField^{SR}}\vdash Tr'_+(\ax{AxEField^{CK}})$ since mathematical formulas are translated into themselves.
\item Let us now prove that $\sy{SpecRel^e_{Full}}\vdash Tr'_+(\ax{AxEv^{CK}})$.
The translation of $\ax{AxEv^{CK}}$ is equivalent to:
\begin{multline*}
\bforall {k,h}{ IOb} 
\Bforall {\vx}{\Q^4} 
\bforall {e}{E} 
\Bexists {\vy}{\Q^4}\\
\Big [\ev_k\left(Rad_{\bar{v}_k(e)}(\vx)\right)=\ev_{h}\left(Rad_{\bar{v}_h(e)}(\vy)\right)\Big].
\end{multline*}
To prove the formula above, let $k$ and $h$ be inertial observers, let $e$ be a primitive ether observer, and let $\bar x \in \Q^4$.  We have to prove that there is a $\vy \in Q^4$ such that $\ev_k[Rad_{\bar{v}_k(e)}(\vx)]=\ev_{h}[Rad_{\bar{v}_h(e)}(\vy)]$. Let us denote $Rad_{\bar{v}_k(e)}(\bar{x})$ by $\bar{x}'$.  $\bar x'$ exists since $Rad_{\bar{v}_k(e)}$ is a well-defined bijection. There is a $\bar y'$ such that $\ev_k\left(\bar{x}'\right)=\ev_{h}\left(\bar{y}'\right)$ because of \ax{AxEv^{SR}}. Then  $\bar y=Rad^{-1}_{\bar{v}_h(e)}(\bar{y}')$ has the requited properties.
\item Let us now prove that $\sy{SpecRel^e_{Full}}\vdash Tr'_+(\ax{ax{AxLine^{CK}}})$. The translation of $\ax{AxLine^{CK}}$ is equivalent to:
\begin{multline*}
\bforall {k,h}{ IOb} 
\Bforall {\vx,\vy,\vz}{\Q^4} 
\Bforall {e}{\Ether} \\
\Big[
Rad_{\bar{v}_k(e)}(\vx),Rad_{\bar{v}_k(e)}(\vy),Rad_{\bar{v}_k(e)}(\vz)\in \wl_k(h)\\ 
\to \bexists {a}{Q} \big[ \vz-\vx=a(\vy-\vx) \lor
\vy-\vz=a(\vz-\vx)\big]\Big]. 
\end{multline*}
Because of $\ax{AxLine^{SR}}$, $Rad_{\bar{v}_k(e)}(\bar x)$, $Rad_{\bar{v}_k(e)}(\bar y)$ and $Rad_{\bar{v}_k(e)}(\bar z)$ are on a straight line. Since $Rad_{\bar v}$ is a linear map, $\bar x$, $\bar y$ and $\bar z$ are on a straight line, hence the translation of $\ax{AxLine^{CK}}$ follows.
\item Let us now prove that $\sy{SpecRel^e_{Full}}\vdash Tr'_+(\ax{AxSelf^{CK}})$. The translation of $\ax{AxSelf^{CK}}$ is equivalent to
\[
\bforall {k} { IOb}  \bforall {e} {E}  
\Bforall {\bar{y}}{\Q^4} \big[ W\big((k,k,Rad_{\bar{v}_k(e)}(\bar{y})\big) \leftrightarrow y_1=y_2=y_3=0 \big].
\]
To prove the formula above, let $k$ be an inertial observer, let $e$ be a primitive ether observer, and let $\bar y \in \Q^4$. We have to prove that $W\big((k,k,Rad_{\bar{v}_k(e)}(\bar{y})\big)$ if and only if $y_1=y_2=y_3=0$. Let $\bar x \in \Q^4$ be such that $Rad^{-1}_{\bar{v}_k(e)}(\bar{x})=\bar y$. By \ax{AxSelf^{SR}}, $W\big((k,k,\bar x)\big)$ if and only if $x_1=x_2=x_3=0$. This holds if and only if $y_1=y_2=y_3=0$ since by item \ref{tx} of Lemma \ref{thm-lemma} $Rad_{\bar{v}}$ transformation maps the time axis to the time axis.
\item Let us now prove that $\sy{SpecRel^e_{Full}}\vdash Tr'_+(\ax{AxSymD^{CK}})$. The translation of $\ax{AxSymD^{CK}}$ is equivalent to:
\begin{multline*}
\bforall {k,k'}{ IOb}\Bforall{\vx,\vy,\vx',\vy'}{Q^4}\bforall{e}{E} \\
\vast(\fblock{
time(\bar x,\bar y) = time(\bar x', \bar y') = 0 \\ 
ev_k(Rad_{\bar{v}_k(e)}(\vx)) = \ev_{k'}(Rad_{\bar{v}_{k'}(e)}(\vx')) \\
ev_k(Rad_{\bar{v}_k(e)}(\vy)) = \ev_{k'}(Rad_{\bar{v}_{k'}(e)}(\vy')) }
\to  space(\vx,\vy)= space(\vx',\vy')\vast).
\end{multline*}
Let $k$ and $k'$ be inertial observers, let $\bar x$, $\bar y$, $\bar x'$, and $\bar y'$ be coordinate points, and let $e$ be a primitive ether observer such that $time(\bar x,\bar y) = time(\bar x', \bar y') = 0$, $ev_{k}(Rad_{\bar{v}_k(e)}(\vx)) = \ev_{k'}(Rad_{\bar{v}_{k'}(e)}(\vx'))$, and $ev_k(Rad_{\bar{v}_k(e)}(\vy)) = \ev_{k'}(Rad_{\bar{v}_{k'}(e)}(\vy'))$. Let $G=Rad^{-1}_{v_{k'}(e)}\circ w_{kk'}\circ Rad_{v_k(e)}$. By Lemma~\ref{lemma-mosquito}, $G$ is a Galilean transformation. By the assumptions, $G(\bar x)=\bar x'$ and  $G(\bar y)=\bar y'$. Since G is a Galilean transformation and $time(\bar x,\bar y)=0$ we have $space(\vx,\vy)=space(\vx',\vy')$. 

\item Let us now prove that $\sy{SpecRel^e_{Full}}\vdash Tr'_+(\ax{AxTriv^{CK}})$. The translation of $\ax{AxTriv^{CK}}$ is equivalent to:
\begin{multline*}
\bforall{T}{\Triv} 
\bforall {h}  {IOb} 
\bforall {e}{E}
\bexists {k}  {IOb}
\big[Rad^{-1}_{\bar v_{k}(e)}\circ w_{hk} \circ Rad_{\bar v_{h}(e)}=T
\big].
\end{multline*}
To prove $Tr'_+(\ax{AxTriv^{CK}})$, we have to find an inertial observer $k$  for every trivial transformation $T$ and an inertial observer $h$ such that $Rad^{-1}_{\bar v_{k}(e)}\circ w_{hk} \circ Rad_{\bar v_{h}(e)}=T$. 
By $\ax{AxTriv^{SR}}$ and  Lemma~\ref{triv-gal}, there is an inertial observer $k$ such that $Rad^{-1}_{\bar v_{k}(e)}\circ w_{hk} \circ Rad_{\bar v_{h}(e)}=T$. 

\item Let us now prove that $\sy{SpecRel^e_{Full}}\vdash Tr'_+(\ax{AxAbsTime^{CK}})$. The translation of $\ax{AxAbsTime^{CK}}$ is equivalent to:
\begin{multline*}
\bforall {k,k'}{ IOb} \Bforall {\vx,\vy,\vx',\vy'}{\Q^4} \bforall{e}{E} \\
\left(\fblock{\ev_k(Rad_{\bar{v}_k(e)}(\vx))=\ev_{k'}(Rad_{\bar{v}_{k'}(e)}(\vx'))\\
\ev_k(Rad_{\bar{v}_k(e)}(\vy))=\ev_{k'}(Rad_{\bar{v}_{k'}(e)}(\vy'))}\to
 time(\vx,\vy)= time(\vx',\vy') \right).
\end{multline*}
Let $k$ and $k'$ be inertial observers, let $\bar x$, $\bar y$, $\bar x'$, and $\bar y'$ be coordinate points, and let $e$ be a primitive ether observer such that $ev_{k}(Rad_{\bar{v}_k(e)}(\vx)) = \ev_{k'}(Rad_{\bar{v}_{k'}(e)}(\vx'))$, and $ev_k(Rad_{\bar{v}_k(e)}(\vy)) = \ev_{k'}(Rad_{\bar{v}_{k'}(e)}(\vy'))$. Let $G=Rad^{-1}_{v_{k'}(e)}\circ w_{kk'}\circ Rad_{v_k(e)}$. By Lemma~\ref{lemma-mosquito}, $G$ is a Galilean transformation. By the assumptions, $G(\bar x)=\bar x'$ and  $G(\bar y)=\bar y'$. Since $G$ is a Galilean transformation, which keeps simultaneous events simultaneous, and \ax{AxSymD^{SR}} we have $time(\vx,\vy)=time(\vx',\vy')$. 

\item Let us now prove that $\sy{SpecRel^e_{Full}}\vdash Tr'_+(\ax{AxNoFTL})$. The translation of $\ax{AxNoFTL}$ is equivalent to:
$\bforall {k} {IOb} 
\bforall {e} {E} \big[\speed_{e}(k) < \mathfrak{c}\big].$
This follows from Corollary \ref{cor-max}.
\item Let us now prove that $\sy{SpecRel^e_{Full}}\vdash Tr'_+(\ax{AxThExp^{STL}})$. The translation of $\ax{AxThExp^{STL}}$ is equivalent to:
\begin{multline*}
\bexists {h}{B}  [ IOb(h)] \AND 
\bforall {e} {E}
\Bforall {\vx,\vy} {\Q^4}
\big( space(\vx,\vy)<\mathfrak{c}\cdot time(\vx,\vy)
\\ \to \bexists {k}  {IOb} \big[Rad_{\bar{v}_k(e)}(\vx),Rad_{\bar{v}_k(e)}(\vy) \in \wl_e(k)\big]\big).
\end{multline*}
From \ax{AxThExp} we get inertial observers inside of the light cones. Inertial observers  stay inside the light cone by the translation by Items \ref{rad-time} and \ref{rad-cone} of Lemma~\ref{thm-lemma}.

\item Let us now prove that $\sy{SpecRel^e_{Full}}\vdash Tr'_+(\ax{AxEther})$. The translation of $\ax{AxEther}$ is equivalent to: 
 $   \bexists {e}{B} \big[E(e)\big].$
This follows from $\ax{AxPrimitiveEther}$.
\item $\sy{SpecRel^e_{Full}} \vdash Tr'_+(\ax{AxNoAcc})$. The translation of $\ax{AxNoAcc}$ is equivalent to:
\begin{multline*}
\bforall{k}{B} \bexists{\bar x}{Q^4} \bexists{b}{B} \bforall {e} {\Ether} \\
\big[W^{SR}\big(k,b,Rad_{\bar{v}_k(e)}(\bar{x})\big)
\rightarrow IOb^{SR}(k) \big],
\end{multline*}
which follows directly from $\ax{AxNoAcc}$ since $Rad_{\bar{v}_k(e)}$ is the same bijection for all ether observers $e$. \qedhere
\end{itemize}
\end{proof}

\begin{lem}\label{conserve-speed}
Both translations $Tr_+$ and $Tr'_+$ preserve the concept ``ether velocity'':
\begin{align*}
\sy{ClassicalKin^{STL}_{Full}} \vdash&Tr_+\big(E(e)\land \bar{v}_k(e)=\bar v\big) \leftrightarrow Ether(e)\land \bar{v}_k(e)=\bar{v}
\\
\sy{SpecRel^{e}_{Full}}\vdash&Tr'_+\big(Ether(e)\land \bar{v}_k(e)=\bar v\big) \leftrightarrow E(e)\land \bar{v}_k(e)=\bar{v}
\end{align*}
\end{lem}

\begin{proof}
The translation of $\bar{v}_{k}(b) = \bar{v}$ by $Tr_+$ is:
\begin{equation*} 
Tr_+\BIG(\fblock{\Bexists{\vx,\vy}{\wl_k(b)}( \vx \neq \vy )  \\
\Bforall{\vx,\vy}{\wl_k(b)} 
\left[ (y_1-x_1,y_2-x_2,y_3-x_3) = \bar{v} \cdot (y_0 - x_0)\right] }\BIG),
\end{equation*}
which is equivalent to
\begin{equation*}
\bforall {e} {\Ether}
\fblock{
\Bexists{\vx',\vy'}{\wl_k(b)} \big(Rad_{\bar{v}_k(e)}(\vx') \neq Rad_{\bar{v}_k(e)}(\vy')
\big) \\ 
\Bforall{\vx',\vy'}{\wl_k(b)} 
\fblock{(y_1-x_1,y_2-x_2,y_3-x_3) = \bar{v} \cdot (y_0 - x_0) \\
\vx = Rad_{\bar{v}_k(e')}(\vx') \\
\vy = Rad_{\bar{v}_k(e')}(\vy')}
}.
\end{equation*}
The translation of velocity relative to the primitive ether, $Tr_+[E(e) \AND \bar{v}_{k}(e) = \bar{v}]$, is:
\begin{multline*}
\Ether(e) \AND 
\bforall {e'} {\Ether} \\
\fblock{
\Bexists{\vx',\vy'}{\wl_k(e)} \big(Rad_{\bar{v}_k(e')}(\vx') \neq Rad_{\bar{v}_k(e')}(\vy')\big) \\
\Bforall{\vx',\vy'}{\wl_k(e)} 
\fblock{(y_1-x_1,y_2-x_2,y_3-x_3) = \bar{v} \cdot (y_0 - x_0) \\
\vx = Rad_{\bar{v}_k(e')}(\vx') \\
\vy = Rad_{\bar{v}_k(e')}(\vy')}
}.
\end{multline*}

Since $e'$ only occurs in $v_k(e')$ and since all ether observers are at rest relative to each other by Corollary~\ref{thm-lightspeed}, they all have the same speed relative to $k$. Since, by  Lemma \ref{thm-lemma}, $Rad_{\bar{v}}$ is a bijection, $Rad_{\bar{v}}(\vx') \neq Rad_{\bar{v}}(\vy')$ is equivalent to $\vx \neq \vy$. Hence we can simplify the above to:
\begin{multline*}
Tr_+[E(e) \AND \bar{v}_{k}(e) = \bar{v}] \equiv \Ether(e) \land \Bexists{\vx',\vy'}{\wl_k(e)} [\vx' \neq \vy'] \land \\
\Bforall{\vx',\vy'}{\wl_k(e)} 
\fblock{(y_1-x_1,y_2-x_2,y_3-x_3) = \bar{v} \cdot (y_0 - x_0) \\
\vx = Rad_{\bar v_k(e)}(\vx')  \land 
\vy = Rad_{\bar v_k(e)}(\vy')},
\end{multline*}
which says that the $Rad_{\bar{v}_k(e)}$-image of worldline $\wl_k(e)$ moves with speed $\bar{v}$, which is equivalent to that $\wl_k(e)$ moves with speed $\bar{v}$ by Item \ref{rad-vel} of Lemma \ref{thm-lemma}, hence
\[Tr_+\big(E(e)\land \bar{v}_k(e)=\bar v\big)\equiv Ether(e)\land \bar{v}_k(e)=\bar{v}.\]
We will now prove this in the other direction.
$Tr'_+[\bar{v}_{k}(b) = \bar{v}]$ 
is:
\begin{multline*}
\bforall {e} {E}
\fblock{
\Bexists{\vx',\vy'}{\wl_k(b)} \big(Rad^{-1}_{\bar{v}_k(e)}(\vx') \neq Rad^{-1}_{\bar{v}_k(e)}(\vy')
\big) \\ 
\Bforall{\vx',\vy'}{\wl_k(b)} 
\fblock{(y_1-x_1,y_2-x_2,y_3-x_3) = \bar{v} \cdot (y_0 - x_0) \\
\vx = Rad^{-1}_{\bar{v}_k(e')}(\vx') \\
\vy = Rad^{-1}_{\bar{v}_k(e')}(\vy')}
}.
\end{multline*}
The translation of velocity relative to the ether 
$Tr'_+[Ether(e) \AND \bar{v}_{k}(e) = \bar{v}]$ is equivalent to
\begin{multline*}
E(e) \land \Bexists{\vx',\vy'}{\wl_k(e)} [\vx' \neq \vy'] \land \\
\Bforall{\vx',\vy'}{\wl_k(e)} 
\fblock{(y_1-x_1,y_2-x_2,y_3-x_3) = \bar{v} \cdot (y_0 - x_0) \\
\vx = Rad^{-1}_{\bar v_k(e)}(\vx')  \land 
\vy = Rad^{-1}_{\bar v_k(e)}(\vy')},
\end{multline*}
which by $Rad_{\bar{v}_k(e)}$ being a bijection and Item \ref{rad-vel} of Lemma \ref{thm-lemma} leads us to
\begin{equation*}
Tr'_+\big(Ether(e)\land \bar{v}_k(e)=\bar v\big)\equiv E(e)\land \bar{v}_k(e)=\bar{v}.     \tag*{\qedhere} 
\end{equation*}
\end{proof}

\begin{thm}\label{thm-defeq}
$Tr_{+}$ is a definitional equivalence between theories \sy{SpecRel_{Full}^{e}} and \sy{ClassicalKin_{Full}^{STL}}.
\end{thm}
\begin{proof}
We only need to prove that the inverse translations of the translated statements are logical equivalent to the original statements since $Tr_+$ and $Tr_+'$ are interpretations by Theorem \ref{thm-tr+} and Theorem \ref{thm-tr'+}. $\ax{AxNoAcc}$ guarantees that $\bar{v}_k(e)$ is defined for every ether observer $e$ and observer $k$.
\begin{itemize}[leftmargin=*]
\item Mathematical expressions, quantities and light signals are translated into themselves by both $Tr_+$ and $Tr_+'$.
\item $Tr_+'\big(Tr_+[E(e)]\big) \equiv Tr_+'[Ether(e)] \equiv E(e)$ follows from the definition of $Tr_+$ and Lemma \ref{lem-Tr_+'(Ether)}.
\item The back and forth translation of $IOb^{SR}$ is the following:
\begin{multline*} Tr_+'\big(Tr_+[IOb^{SR}(k)]\big) 
\equiv Tr_+'\big( IOb^{CK}(k)\AND \bforall {e} {\Ether} \big[ \speed_{e}(k) < \mathfrak{c_e}\big] \big) \\
\equiv IOb^{SR}(k)\AND \bforall {e} {E} \big[ \speed_{e}(k) < \mathfrak{c}\big] 
\equiv IOb^{SR}(k).
\end{multline*}
The second equivalence is true because of Lemma \ref{lem-Tr_+'(Ether)} and Corollary \ref{cor-Tr_+'(c_e)}. The last equivalence is true because observers are always slower-than-light in \sy{SpecRel^e_{Full}}, as per Corollary \ref{cor-max}.
\item The back and forth translation of $IOb^{CK}$ is the following:
\begin{multline*} Tr_+\big(Tr_+'[IOb^{CK}(b)]\big) \equiv Tr_+[IOb^{SR}(b)] \\ 
\equiv IOb^{CK}(b) \AND \bforall{e}{Ether} \big[ \speed_{e}(b) < \mathfrak{c_e}\big]\equiv IOb^{CK}(b).
\end{multline*}
The last equivalence holds because $\bforall{e}{Ether} \big[ \speed_{e}(b) < \mathfrak{c_e}\big]$ is true by \ax{AxNoFTL}.
\item The back and forth translation of $W^{SR}$ is the following:
\begin{multline*} Tr_+'\big(Tr_+[W^{SR}(k,b,\bar x)]\big) \equiv Tr\Big[\bforall{e}{Ether} \big[W^{CK}\big(k,b,Rad^{-1}_{\bar{v}_k(e)}(\bar{x})\big)\big]\Big] \\
\equiv \bforall{e}{E} \Big(W^{SR}\big[k,b,Rad_{\bar{v}_k(e)}\big(Rad^{-1}_{\bar{v}_k(e)}(\bar{x})\big)\big]\Big) \\
\equiv \bforall{e}{E}[W^{SR}(k,b,\bar x)] \equiv W^{SR}(k,b,\bar x).
\end{multline*}
We use Lemma \ref{conserve-speed} to translate the indexes $\bar{v}_k(e)$ into themselves.
\item The back and forth translation of $W^{CK}$ is the following:
\begin{multline*}
Tr_+\big(Tr_+'[W^{CK}(k,b,\bar x)]\big) \equiv Tr_+\Big[\bforall{e}{E} \big[W^{SR}\big(k,b,Rad_{\bar{v}_k(e)}(\bar{x})\big)\big]\Big] \\
\equiv \bforall{e}{Ether} \Big(W^{SR}\big[k,b,Rad^{-1}_{\bar{v}_k(e)}\big(Rad_{\bar{v}_k(e)}(\bar{x})\big)\big]\Big) \\
\equiv \bforall{e}{Ether}[W^{CK}(k,b,\bar x)] \equiv W^{CK}(k,b,\bar x).
\end{multline*}
We use Lemma \ref{conserve-speed} to translate the indexes $\bar{v}_k(e)$ into themselves.  \qedhere
\end{itemize}
\end{proof}

\section{Faster-Than-Light Observers are Definable from Slower-Than-Light ones in Classical Kinematics}

Now, we show that \sy{ClassicalKin^{STL}_{Full}} and \sy{ClassicalKin_{Full}} are definitionally equivalent theories. In this section, we work only with classical theories. We use the notations $IOb^{STL}$, $W^{STL}$ and $w^{STL}$ for inertial observers, worldview relations and worldview transformations in \sy{ClassicalKin^{STL}_{Full}} to distinguish them from their counterparts in \sy{ClassicalKin_{Full}}.

We can map the interval of speeds $[0,\mathfrak{c_e}]$ to $[0,\infty]$ by replacing slower-than-light speed $v$ by classical speed $V=\frac{v}{\mathfrak{c_e}-v}$, and conversely map the interval of speeds $[0,\infty]$ to $[0,\mathfrak{c_e}]$ by replacing speed $V$ by speed $v=\frac{\mathfrak{c_e}V}{1+V}$.
Similarly, for arbitrary (finite) velocity $\bar{V}$, we have that $\bar v=\frac{\mathfrak{c_e}\bar{V}}{1+\left|\bar{V}\right|}$ is slower than $\mathfrak{c_e}$, and from $\bar v$, we can get $\bar{V}$ back by the equation $\bar V=\frac{\bar v}{\mathfrak{c_e}-\left|\bar v\right|}$.

Let $G_{\bar{V}}$ and $G_{\bar{v}}$, respectively, be the Galilean boosts that map bodies moving with velocity $\bar{V}$ and $\bar{v}$ to stationary ones. Let $X_{\bar{V}}=G^{-1}_{\bar v}\circ G_{\bar{V}}$ and $Y_{\bar{v}}=G^{-1}_{\bar V}\circ G_{\bar{v}}$, see Figure~\ref{figxy}. 

\begin{figure}[!htb]\label{figxy}
  \begin{center}
    \scalebox{.8}{\begin{tikzpicture}[scale=0.4]
\pgfmathsetmacro\size{4}
\pgfmathsetmacro\v{1.3}
\pgfmathsetmacro\w{\v/(1+\v)}
\pgfmathsetmacro\u{-1}
\pgfmathsetmacro\uu{\u/(1-\u)}

\tikzstyle{T0}=[cm={1,0,-\v,1,(0,0)}]
\tikzstyle{T1}=[cm={1,0,-\w,1,(0,0)}]
\tikzstyle{T00}=[cm={1,0,-\u,1,(0,0)}]
\tikzstyle{T11}=[cm={1,0,-\uu,1,(0,0)}]

\draw (0,\size) node[right] {$e$} -- (0,0)  ;
\draw[thick, red] (0,3) ellipse (3 and 0.3) (-3,3)--(0,0)--(3,3);

\draw[->,ultra thick,>=latex] (-6,-1)  to [out=90,in=180] node[above left]{$G_{\bar{V}}$} (-3,1);

\begin{scope}[shift={(-1.2*\size,-1.5*\size)}]
\draw[T0] (0,\size) node[right] {$e$} -- (0,0)  ;
\draw (0,0) to (0*\size,\size);
\draw[thick, red,T0] (0,3) ellipse (3 and 0.3) (-3,3)--(0,0)--(3,3);
\end{scope}

\draw[->,ultra thick,>=latex] (3,1)  to [out=0,in=90] node[above right]{$G^{-1}_{\bar{v}}=G_{-\bar{v}}$} (5,-1);

\begin{scope}[shift={(1.5*\size,-1.5*\size)}]
\draw[T1] (0,\size) node[right] {$e$} -- (0,0)  ;
\draw (0,0) to (0*\size,\size);
\draw[thick, red,T1] (0,3) ellipse (3 and 0.3) (-3,3)--(0,0)--(3,3);
\end{scope}

\draw[->,ultra thick,>=latex] (-2,-3.5)  to [out=0,in=180] node[above]{$X_{\bar{V}}$} (1,-3.5);
\draw[<-,ultra thick,>=latex] (-2,-4.5)  to [out=0,in=180] node[below]{$Y_{\bar{v}}$} (1,-4.5);

\begin{scope}[shift={(0,-3*\size)}]
\draw (0,\size) node[right] {$e$} -- (0,0)  ;
\draw[thick, red] (0,3) ellipse (3 and 0.3) (-3,3)--(0,0)--(3,3);
\end{scope}

\draw[<-,ultra thick,>=latex] (-6,-7)  to [out=-90,in=180] node[below left]{$G^{-1}_{\bar{V}}=G_{-\bar{V}}$} (-3,-10);
\draw[<-,ultra thick,>=latex] (3,-10)  to [out=0,in=-90] node[below right]{$G_{\bar{v}}$} (6,-7);

\node (text) at (0,-14) {$\bar{v}=\frac{c\bar{V}}{1+\left|\bar{V}\right|}$\quad $\bar{V}=\frac{\bar{v}}{c-\left|\bar{v}\right|}$};

\end{tikzpicture}}
    \caption{\label{fig-transformationsXY} $\bar{V}$ is an arbitrary velocity and $\bar v = \frac{\mathfrak{c_e}\bar{V}}{1+\left|\bar{V}\right|}$ is the corresponding STL velocity. $G_{\bar{V}}$ and $G_{\bar{v}}$ are Galilean boosts that, respectively, map bodies moving with velocity $\bar{V}$ and $\bar{v}$ to stationary ones. Transformations $X_{\bar{V}}=G^{-1}_{\bar v}\circ G_{\bar{V}}$ and $Y_{\bar{v}}=G^{-1}_{\bar V}\circ G_{\bar{v}}$ allow us to map between observers seeing the ether frame moving with  up to infinite speeds on the left and  STL speeds on the right. The light cones on the top and on the bottom are the same one.}
  \end{center}
\end{figure}
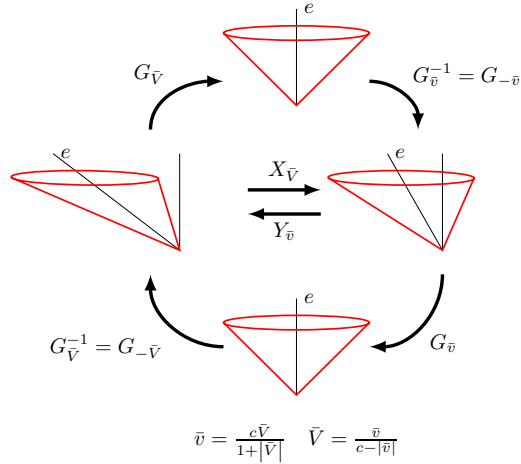

\begin{lem}\label{lem-xy}
Let $\bar v\in Q^3$ for which $|\bar v|\in [0,\mathfrak{c_e}]$  and let $\bar V=\frac{\bar v}{\mathfrak{c_e}-|\bar v|}$. Then $X^{-1}_{\bar V}=Y_{\bar v}$. 
\end{lem}
\begin{proof}
By definition of $X_{\bar V}$ and $Y_{\bar v}$ and by the inverse of composed transformations:
$X^{-1}_{\bar V} = \big(G^{-1}_{\bar v}\circ G_{\bar{V}}\big)^{-1} = G^{-1}_{\bar V}\circ \big(G^{-1}_{\bar{v}}\big)^{-1} = G^{-1}_{\bar V}\circ G_{\bar{v}} = Y_{\bar v}.$
\end{proof}

Assume that the ether frame is moving with a faster-than-light velocity $\bar{V}$ with respect to an inertial observer $k$.  Then by $X_{\bar V}$ we can transform the worldview of $k$ such a way that after the transformation the ether frame is moving slower than light with respect to $k$, see Figure \ref{fig-ftl2stl}. Systematically modifying every observer's worldview using the corresponding transformation $X_{\bar V}$, we can achieve that every observer sees that the ether frame is moving slower than light. These transformations tell us where the observers should see the non-observer bodies. However, this method does usually not work for bodies representing inertial observers because the transformation $X_{\bar {V}}$ leaves the time axis fixed only if $\bar{V}=(0,0,0)$. Therefore, we will translate the worldlines of bodies representing observers in harmony with \ax{AxSelf} to represent the motion of the corresponding observer's coordinate system.
This means that we have to split up the translation of $W$ between the observer and non-observer cases.\footnote{There are other ways to handle this issue, such as introducing a new sort for inertial observers. This however would complicate the previous sections of this paper, and also take our axiom system further away from stock \sy{SpecRel}.}

\noindent
Let us define $X_{\bar v_k(e)}(\bar x)$ and $Y_{\bar v_k(e)}(\bar x)$ and their inverses as:
\begin{equation*}
X_{\bar v_k(e)}(\bar x) = \bar y 
\defiff X^{-1}_{\bar v_k(e)}(\bar y) = \bar x 
\defiff \Bexists{\bar v}{Q^3}[\bar V = \bar v_k(e) \AND X_{\bar V}(\bar x) = \bar y],
\end{equation*}
and
\begin{equation*}
Y_{\bar v_k(e)}(\bar x) = \bar y 
\defiff Y^{-1}_{\bar v_k(e)}(\bar y) = \bar x 
\defiff \Bexists{\bar v}{Q^3}[\bar v = \bar v_k(e) \AND Y_{\bar v}(\bar x) = \bar y].
\end{equation*}

\noindent
Let $Tr_*$ be the following translation:
\begin{multline*}
  Tr_*\big(W^{STL}(k,b,\bar{x})\big)\defeq
  \bforall{e}{Ether} \\
  \fblock{
    b\not\in IOb \to W^{CK}\big(k,b,X^{-1}_{\bar{v}_k(e)}(\bar{x})\big)\\
    b\in IOb \to \bexists{t}{Q} \big[w^{CK}_{kb}\big(X^{-1}_{\bar{v}_k(e)}(\bar{x})\big)= X_{\bar{v}_b(e)}^{-1}(t,0,0,0) \big] }.
\end{multline*}
and $Tr_*$ is the identity on the other concepts.  

\begin{figure}[!htb]
  \begin{center}
    \scalebox{.8}{\begin{tikzpicture}[scale=0.4]
\pgfmathsetmacro\size{4}
\pgfmathsetmacro\v{1.3}
\pgfmathsetmacro\w{\v/(1+\v)}
\pgfmathsetmacro\u{-1}
\pgfmathsetmacro\uu{\u/(1-\u)}

\tikzstyle{T0}=[cm={1,0,-\v,1,(0,0)}]
\tikzstyle{T1}=[cm={1,0,-\w,1,(0,0)}]
\tikzstyle{T00}=[cm={1,0,-\u,1,(0,0)}]
\tikzstyle{T11}=[cm={1,0,-\uu,1,(0,0)}]


\draw (0,\size) node[right] {$e$} -- (0,0)  ;
\draw[ultra thick, blue] (0,0) to (\w*\size,\size) node[right]{$k$};
\draw[thick, red] (0,3) ellipse (3 and 0.3) (-3,3)--(0,0)--(3,3);

\draw[<-,ultra thick,>=latex] (-6,-1)  to [out=90,in=180] node[above left]{$G^{-1}_{{\bar{v}_k(e)}}$} (-3,1);

\begin{scope}[shift={(-2*\size,-1.5*\size)}]
\draw[T0] (0,\size) node[right] {$e$} -- (0,0)  ;
\draw[ultra thick, green!42!black, dashed] (0,0) to (-\size*2.3,\size) node[left]{$b$};
\draw[T0,ultra thick, blue] (0,0) to (\w*\size,\size) node[right]{$k$};
\draw[ultra thick, blue, dashed] (0,0) to (0*\size,\size) node[right]{$k$};
\draw[thick, red,T0] (0,3) ellipse (3 and 0.3) (-3,3)--(0,0)--(3,3);
\end{scope}

\draw[<-,ultra thick,>=latex] (3,1)  to [out=0,in=90] node[above right]{$G_{\frac{\mathfrak{c_e}\bar{v}_k(e)}{1+\left|\bar{v}_k(e)\right |}}$} (5,-1);

\begin{scope}[shift={(1.5*\size,-1.5*\size)}]
\draw[T1] (0,\size) node[right] {$e$} -- (0,0)  ;
\draw[ultra thick, blue] (0,0) to (0*\size,\size) node[right]{$k$};
 \draw[ultra thick, green!42!black] (0,0) to (\u*\size,\size) node[right]{$b$};
\draw[thick, red,T1] (0,3) ellipse (3 and 0.3) (-3,3)--(0,0)--(3,3);
\end{scope}

\draw[<-,ultra thick,>=latex] (-2,-4)  to [out=0,in=180] node[above]{$X^{-1}_{\bar{v}_k(e)}$} (1,-4);

\draw[<-,ultra thick,>=latex] (-10,-10)  to [out=180,in=180] node[left]{$w^{CK}_{kb}$} (-10,-6);


\begin{scope}[shift={(0,-4.5*\size)}]
\draw (0,\size) node[right] {$e$} -- (0,0)  ;
\draw[ultra thick, green!42!black] (0,0) to (\uu*\size,\size) node[left]{$b$};
\draw[thick, red] (0,3) ellipse (3 and 0.3) (-3,3)--(0,0)--(3,3);
\end{scope}

\draw[<-,ultra thick,>=latex] (-6,-14)  to [out=-90,in=-180] node[below left]{$G^{-1}_{\bar{v}_b(e)}$} (-4,-16);

\begin{scope}[shift={(-2*\size,-3*\size)}]
\draw[T00] (0,\size) node[right] {$e$} -- (0,0)  ;
\draw[T00,ultra thick,  green!42!black] (0,0) to (\uu*\size,\size) node[left]{$b$};
\draw[ultra thick, green!42!black,dashed] (0,0) to (0*\size,\size) node[left]{$b$};
\draw[thick, red,T00] (0,3) ellipse (3 and 0.3) (-3,3)--(0,0)--(3,3);
\end{scope}

\draw[<-,ultra thick,>=latex] (3,-16)  to [out=0,in=-90] node[below right]{$G_{\frac{\mathfrak{c_e}\bar{v}_b(e)}{1+\left|\bar{v}_b(e)\right |}}$} (5,-14);

\begin{scope}[shift={(1.5*\size,-3*\size)}]
\draw[T11] (0,\size) node[right] {$e$} -- (0,0)  ;
\draw[ultra thick, green!42!black] (0,0) to (0*\size,\size) node[left]{$b$};
\draw[thick, red,T11] (0,3) ellipse (3 and 0.3) (-3,3)--(0,0)--(3,3);
\end{scope}

\draw[<-,ultra thick,>=latex] (-2,-10)  to [out=0,in=180] node[above]{$X^{-1}_{\bar{v}_b(e)}$} (1,-10);

\draw[<-,ultra thick,>=latex] (3,-10)  to [out=180,in=180] node[above left]{$w^{STL}_{kb}$} (3,-6);

\end{tikzpicture}}
    \caption{\label{fig-ftl2stl-inv}Transformations from, on the right, \sy{ClassicalKin^{STL}_{Full}} to, on the left, \sy{ClassicalKin_{Full}} for, on top, observer $k$ and, at the bottom, observer $b$. The dashed lines are where we put the observers after the transformation in order to respect \ax{AxSelf}. All transformations, including the worldviev transfrormations from the top to the bottom, are Galilean.}
  \end{center}
\end{figure}
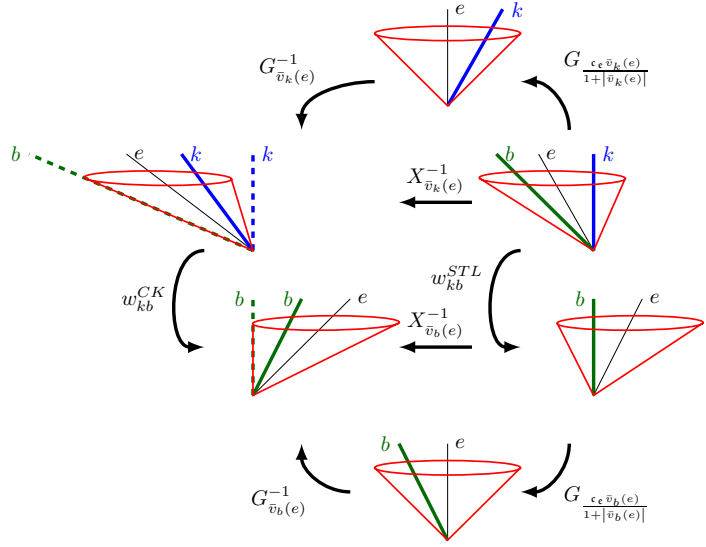
\begin{figure}[!htb]
  \begin{center}
    \scalebox{.8}{\begin{tikzpicture}[scale=0.4]
\pgfmathsetmacro\size{4}
\pgfmathsetmacro\v{1.3}
\pgfmathsetmacro\w{\v/(1+\v)}
\pgfmathsetmacro\u{-1}
\pgfmathsetmacro\uu{\u/(1-\u)}

\tikzstyle{T0}=[cm={1,0,-\v,1,(0,0)}]
\tikzstyle{T1}=[cm={1,0,-\w,1,(0,0)}]
\tikzstyle{T00}=[cm={1,0,-\u,1,(0,0)}]
\tikzstyle{T11}=[cm={1,0,-\uu,1,(0,0)}]


\draw (0,\size) node[right] {$e$} -- (0,0)  ;
\draw[ultra thick, blue] (0,0) to (\v*\size,\size) node[right]{$k$};
\draw[thick, red] (0,3) ellipse (3 and 0.3) (-3,3)--(0,0)--(3,3);

\draw[->,ultra thick,>=latex] (-6,-1)  to [out=90,in=180] node[above left]{$G_{\frac{\bar{v}_k(e)}{\mathfrak{c_e}-\left|\bar v_k(e)\right|}}$} (-3,1);

\begin{scope}[shift={(-2*\size,-1.5*\size)}]
\draw[T0] (0,\size) node[right] {$e$} -- (0,0)  ;
\draw[ultra thick, green!42!black] (0,0) to (-\size*2.3,\size) node[left]{$b$};
\draw[ultra thick, blue] (0,0) to (0*\size,\size) node[right]{$k$};
\draw[thick, red,T0] (0,3) ellipse (3 and 0.3) (-3,3)--(0,0)--(3,3);
\end{scope}

\draw[->,ultra thick,>=latex] (3,1)  to [out=0,in=90] node[above right]{$G^{-1}_{\bar{v}_k(e)}$} (5,-1);

\begin{scope}[shift={(1.5*\size,-1.5*\size)}]
\draw[T1] (0,\size) node[right] {$e$} -- (0,0)  ;
\draw[ultra thick, blue,dashed] (0,0) to (0*\size,\size) node[right]{$k$};
\draw[ultra thick, blue,T1] (0,0) to (\v*\size,\size) node[right]{$k$};
\draw[ultra thick, green!42!black, dashed] (0,0) to (\u*\size,\size) node[right]{$b$};
\draw[thick, red,T1] (0,3) ellipse (3 and 0.3) (-3,3)--(0,0)--(3,3);
\end{scope}

\draw[->,ultra thick,>=latex] (-2,-4)  to [out=0,in=180] node[above]{$Y^{-1}_{\bar{v}_k(e)}$} (1,-4);

\draw[<-,ultra thick,>=latex] (-10,-10)  to [out=180,in=180] node[left]{$w^{CK}_{kb}$} (-10,-6);


\begin{scope}[shift={(0,-4.5*\size)}]
\draw (0,\size) node[right] {$e$} -- (0,0)  ;
\draw[ultra thick, green!42!black] (0,0) to (\u*\size,\size) node[left]{$b$};
\draw[thick, red] (0,3) ellipse (3 and 0.3) (-3,3)--(0,0)--(3,3);
\end{scope}

\draw[->,ultra thick,>=latex] (-6,-14)  to [out=-90,in=-180] node[below left]{$G_{\frac{\bar{v}_b(e)}{\mathfrak{c_e}-\left|\bar v_b(e)\right|}}$} (-4,-16);

\begin{scope}[shift={(-2*\size,-3*\size)}]
\draw[T00] (0,\size) node[right] {$e$} -- (0,0)  ;
\draw[ultra thick,  green!42!black] (0,0) to (0*\size,\size) node[left]{$b$};
\draw[thick, red,T00] (0,3) ellipse (3 and 0.3) (-3,3)--(0,0)--(3,3);
\end{scope}

\draw[->,ultra thick,>=latex] (3,-16)  to [out=0,in=-90] node[below right]{$G^{-1}_{\bar{v}_b(e)}$} (5,-14);

\begin{scope}[shift={(1.5*\size,-3*\size)}]
\draw[T11] (0,\size) node[right] {$e$} -- (0,0)  ;
\draw[ultra thick, green!42!black,dashed] (0,0) to (0*\size,\size) node[left]{$b$};
\draw[ultra thick, green!42!black,T11] (0,0) to (\u*\size,\size) node[left]{$b$};
\draw[thick, red,T11] (0,3) ellipse (3 and 0.3) (-3,3)--(0,0)--(3,3);
\end{scope}

\draw[->,ultra thick,>=latex] (-2,-10)  to [out=0,in=180] node[above]{$Y^{-1}_{\bar{v}_b(e)}$} (1,-10);

\draw[<-,ultra thick,>=latex] (3,-10)  to [out=180,in=180] node[above left]{$w^{STL}_{kb}$} (3,-6);

\end{tikzpicture}}
    \caption{\label{fig-ftl2stl} Transformations from, on the left, \sy{ClassicalKin_{Full}} to, on the right, \sy{ClassicalKin^{STL}_{Full}} for, on top, observer $k$ and, at the bottom, observer $b$. The dashed lines are where we put the observers after the transformation in order to respect \ax{AxSelf}. All transformations, including the worldviev transfrormations from the top to the bottom, are Galilean. The cone at the bottom is only displayed for completion, it is not used in our reasoning.}
  \end{center}
\end{figure}
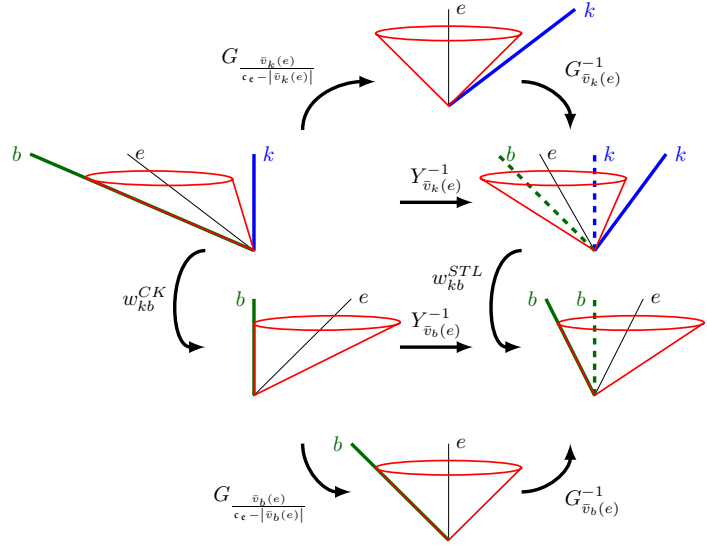

\noindent
Let $Tr_*'$ be the following translation:
\begin{multline*}
  Tr'_*\big(W^{CK}(k,b,\bar{x})\big)\defeq
  \bforall{e}{Ether} \\
  \fblock{
    b\not\in IOb \to W^{STL}\big(k,b,Y^{-1}_{\bar{v}_k(e)}(\bar{x})\big) \\
    b\in IOb \to \bexists{t}{Q} \big[w^{STL}_{kb}\big(Y^{-1}_{\bar{v}_k(e)}(\bar{x})\big)= Y^{-1}_{\bar{v}_b(e)}(t,0,0,0) \big]
  }.
\end{multline*}
and $Tr'_*$ is the identity on the other concepts.

\begin{lem}\label{lem-tr'_FTL} Assume \sy{ClassicalKin^{STL}_{Full}}. Then
\[Tr'_*\left(w^{CK}_{kh}\right)=Y_{\bar{v}_h(e)}\circ w^{STL}_{kh}\circ Y^{-1}_{\bar{v}_k(e)}.\]
\end{lem}

\begin{proof}
We should show that 
\begin{equation}\label{eq-tr_*}
\bforall{b}{B}\big[Tr_*\big(W^{CK}(k,b,\bar{x})\big)\leftrightarrow Tr_*\big(W^{CK}(h,b,\bar{y})\big)\big]
\end{equation}
holds iff 
\begin{equation}\label{eq-tr_*2}
\bforall{b}{B}\bforall{e}{Ether}\big[w^{STL}_{kh}\big(Y^{-1}_{\bar{v}_k(e)}(\bar{x})\big)=Y^{-1}_{\bar{v}_h(e)}(\bar{y})\big]
\end{equation}
holds. 

Let us first assume that \eqref{eq-tr_*} holds. This is equivalent to conjunction of the following two statements: 
\begin{multline}
\label{eq-tr_*3}\bforall{b}{B\setminus \IOb}\bforall{e}{Ether}\\\big[\big(W^{STL}(k,b,Y^{-1}_{\bar{v}_k(e)}(\bar{x}))\big)\leftrightarrow \big(W^{STL}(h,b,Y^{-1}_{\bar{v}_h(e)}(\bar{y}))\big)\big]
\end{multline}
and 
\begin{multline}\label{eq-tr_*4}
 \bforall{b}{IOb}\bforall{e}{Ether} \Big[\bexists{t}{Q} \big[ w^{STL}_{kb}\big(Y^{-1}_{\bar{v}_k(e)}(\bar{x})\big)=Y^{-1}_{\bar{v}_b(e)}(t,0,0,0)\big]\\ \leftrightarrow \bexists{t}{Q} \big[w^{STL}_{hb}\big(Y^{-1}_{\bar{v}_h(e)}(\bar{y})\big)=  Y^{-1}_{\bar{v}_b(e)}(t,0,0,0)\big]\Big].
\end{multline}
By \eqref{eq-tr_*3}, we have that $W^{STL}\big(k,b,Y^{-1}_{\bar{v}_k(e)}(\bar{x})\big)$ iff $W^{STL}\big(h,b,Y^{-1}_{\bar{v}_h(e)}(\bar{y})\big)$ if $b$ is not an inertial observer. To show \eqref{eq-tr_*2}, we have to show this also if $b$ is an inertial observer. Let $b$ be an inertial observer, then $W^{STL}\big(k,b,Y^{-1}_{\bar{v}_k(e)}(\bar{x})\big)$ holds exactly if $w^{STL}_{kb}\big(Y^{-1}_{\bar{v}_k(e)}(\bar{x})\big)$ is on the $Y^{-1}_{\bar{v}_b(e)}$-image of the time axis by the definition of $Tr_*$. By \eqref{eq-tr_*4}, we have that $w^{STL}_{kb}\big(Y^{-1}_{\bar{v}_k(e)}(\bar{x})\big)$ is on the $Y^{-1}_{\bar{v}_b(e)}$-image of the time axis iff $w^{STL}_{hb}\big(Y^{-1}_{\bar{v}_h(e)}(\bar{y})\big)$ is on the $Y^{-1}_{\bar{v}_b(e)}$-image of the time axis, which is equivalent to $W^{STL}\big(h,b,Y^{-1}_{\bar{v}_h(e)}(\bar{x})\big)$ by the definition of $Tr'_*$. Consequently, $W^{STL}\big(k,b,Y^{-1}_{\bar{v}_k(e)}(\bar{x})\big)$ iff $W^{STL}\big(h,b,Y^{-1}_{\bar{v}_h(e)}(\bar{y})\big)$. Therefore, \eqref{eq-tr_*} implies \eqref{eq-tr_*2}.  

Let us now assume that statement \eqref{eq-tr_*2} holds.  Then \eqref{eq-tr_*3} follows immediately from the definition of the worldview  transformation.  By \eqref{eq-tr_*2}, we have $Y^{-1}_{\bar{v}_h(e)}(\bar{y})=w^{STL}_{kh}\big(Y^{-1}_{\bar{v}_k(e)}(\bar{x})\big)$. Hence, by the definition of the worldview transformation, \[w^{STL}_{hb}\big(Y^{-1}_{\bar{v}_h(e)}(\bar{y})\big)=w^{STL}_{hb}\big(w^{STL}_{kh}\big(Y^{-1}_{\bar{v}_k(e)}(\bar{x})\big)\big) =w^{STL}_{kb}\big(Y^{-1}_{\bar{v}_k(e)}(\bar{x})\big),\] which implies \eqref{eq-tr_*4} immediately.  
\end{proof}

\begin{lem}\label{lem-tr_FTL} Assume \sy{ClassicalKin_{Full}}. Then
\[Tr_*\left(w^{STL}_{kb}\right)=X_{\bar{v}_b(e)}\circ w^{CK}_{kb}\circ X^{-1}_{\bar{v}_k(e)}.\]
\end{lem}
\begin{proof}
The proof is completely analogous to that of Lemma~\ref{lem-tr'_FTL}.
\end{proof}

\begin{lem}\label{triv-FTL}
Assume  \sy{{ClassicalKin}_{Full}}.  Let $e$ be an ether observer, and let $h$ and $k$ be inertial observers. Assume that $w^{CK}_{hk}$ is a trivial transformation which is the translation by the vector $\bar{z}$ after the linear trivial transformation $T_0$. Then $X_{\bar{v}_{k}(e)}\circ w^{CK}_{hk}\circ X^{-1}_{\bar{v}_h(e)}$ is the trivial transformation which is the translation by the vector $X_{\bar{v}_{k}(e)}(\bar{z})$ after $T_0$.
\end{lem}
\begin{proof}
Let $M_{\bar{z}}$ denote the translation by vector $\bar{z}$.
By the assumptions, $w^{CK}_{hk}=M_{\bar{z}}\circ T_0$. The linear part $T_0$ of $w^{CK}_{hk}$ transforms the velocity of the ether frame as $(0,\bar{v}_k(e))=T_0(0,\bar{v}_h(e))$ and the translation part $M_{\bar{z}}$ does not change the velocity of the ether frame. Hence $\bar{v}_k(e)$ is $\bar{v}_h(e)$ transformed by the spatial isometry part of $T_0$.

By the definition of $X_{\bar v}$ and Theorem~\ref{thm-gal}, $X_{\bar{v}_k(e)}\circ w^{SR}_{hk}\circ X^{-1}_{\bar{v}_h(e)}$ is a Galilean transformation. Since  $w^{CK}_{hk}$ is a trivial transformation, it maps vertical lines to vertical ones. $X_{\bar{v}_k(e)}\circ w^{CK}_{hk}\circ X^{-1}_{\bar{v}_h(e)}$ also maps vertical lines to vertical ones. Consequently,  $X_{\bar{v}_k(e)}\circ w^{CK}_{hk}\circ X^{-1}_{\bar{v}_h(e)}$  is a Galilean transformation that maps vertical lines to vertical ones. Hence it is a trivial transformation.

We also have that $X_{\bar{v}_k(e)}\circ w^{CK}_{hk}\circ X^{-1}_{\bar{v}_h(e)}$ is $X_{\bar{v}_k(e)}\circ M_{\bar{z}} \circ T_0\circ X^{-1}_{\bar{v}_h(e)}$. Since $X_{\bar{v}_k(e)}$ is linear, we have  $X_{\bar{v}_k(e)}\circ M_{\bar{z}} = M_{X_{\bar{v}_k(e)}(\bar{z})}\circ X_{\bar{v}_k(e)}$. Therefore, it is enough to prove that $X_{\bar{v}_k(e)}\circ w^{CK}_{hk}\circ X^{-1}_{\bar{v}_h(e)}=w^{CK}_{hk}$ if $w^{CK}_{hk}$ is linear. From now on assume that $w^{CK}_{hk}$ is linear.

Since it is a linear trivial transformation, $w^{CK}_{hk}$ maps $(1,0,0,0)$ to itself. $X_{\bar{v}_k(e)}\circ w^{CK}_{hk}\circ X^{-1}_{\bar{v}_h(e)}=w^{CK}_{hk}$ also  maps $(1,0,0,0)$ to itself because $\bar{v}_k(e)$ is $\bar{v}_h(e)$ transformed by the spatial isometry part of $w^{CK}_{hk}$. So $X_{\bar{v}_k(e)}\circ w^{CK}_{hk}\circ X^{-1}_{\bar{v}_h(e)}$ and $w^{CK}_{hk}$ also agree restricted to time.

Now we have to prove that  $X_{\bar{v}_k(e)}\circ w^{CK}_{hk}\circ X^{-1}_{\bar{v}_h(e)}$ and $w^{CK}_{hk}$ also agree restricted to space. Since they are Galilean boosts, $X_{\bar{v}_h(e)}$ and $X_{\bar{v}_k(e)}$ are identical on the vectors orthogonal to the time axis. Consequently, $X_{\bar{v}_k(e)}\circ w^{CK}_{hk}\circ X^{-1}_{\bar{v}_h(e)}$ and $w^{CK}_{hk}$ also agree restricted to space. Hence $X_{\bar{v}_k(e)}\circ w^{CK}_{hk}\circ X^{-1}_{\bar{v}_h(e)}=w^{CK}_{hk}$.
\end{proof}

\begin{lem}\label{tr_Ether}
\begin{equation*}
\sy{ClassicalKin_{Full}}\vdash Tr_*(Ether(e)) \leftrightarrow Ether(e). 
\end{equation*}
\end{lem}

\begin{proof}
Assuming \sy{ClassicalKin_{Full}}, $Tr_*(Ether(e))$ is equivalent to
\begin{multline*}
\Big(IOb(e)\AND \bexists {c} {\Q} \Big[c>0\AND \Bforall {\vx,\vy}{\Q^4} \bforall{e'}{Ether} \\ 
\big(\bexists {p} {\Ph}
\fblock{
W^{CK}\big(e,p,X^{-1}_{v_e(e')}(\vx)\big) \\
W^{CK}\big(e,p,X^{-1}_{v_e(e')}(\vy)\big) }
\leftrightarrow  space(\vx,\vy)=c\cdot time(\vx,\vy)\big)\Big] \Big),
\end{multline*}
which, since by Corollary \ref{thm-lightspeed} all ether observers are stationary relative to each other and therefore $\Bforall{k}{IOb}\Bforall{e_1,e_2}{Ether}[v_k(e_1) = v_k(e_2)]$, can be simplified to
\begin{multline*}
\Big(IOb(e)\AND \bexists {c} {\Q} \Big[c>0\AND \Bforall {\vx,\vy}{\Q^4} \\ 
\big(\bexists {p} {\Ph}
\fblock{
W^{CK}\big(e,p,X^{-1}_{\bar v_e(e)}(\vx)\big) \\
W^{CK}\big(e,p,X^{-1}_{\bar v_e(e)}(\vy)\big) }
\leftrightarrow  space(\vx,\vy)=c\cdot time(\vx,\vy)\big)\Big] \Big).
\end{multline*}
Since $\bar v_e(e) = \bar 0$, $X^{-1}_{\bar 0}$ is by definition the identity, hence from $\ax{AxEther^{CK}}$ it follows that $Tr_*(Ether(e)) \leftrightarrow Ether(e)$.
\end{proof}

\begin{cor}\label{cor_ce1}
\begin{equation*}
\sy{ClassicalKin_{Full}}\vdash Tr_*(\mathfrak{c_e}) = \mathfrak{c_e}. 
\end{equation*}
\end{cor}

\begin{proof}
This is a direct consequence of Lemma~\ref{tr_Ether}.
\end{proof}

\begin{lem}\label{tr_Ether2}
\begin{equation*}
\sy{ClassicalKin^{STL}_{Full}}\vdash Tr'_*(Ether(e)) \leftrightarrow Ether(e). 
\end{equation*}
\end{lem}

\begin{proof}
Assuming \sy{ClassicalKin^{STL}_{Full}}, $Tr'_*(Ether(e))$ is, after simplification by Corollary \ref{thm-lightspeed} equivalent to
\begin{multline*}
\Big(IOb(e)\AND \bexists {c} {\Q} \Big[c>0\AND \Bforall {\vx,\vy}{\Q^4} \\ 
\big(\bexists {p} {\Ph}
\fblock{
W^{CK}\big(e,p,Y^{-1}_{\bar v_e(e)}(\vx)\big) \\
W^{CK}\big(e,p,Y^{-1}_{\bar v_e(e)}(\vy)\big) }
\leftrightarrow  space(\vx,\vy)=c\cdot time(\vx,\vy)\big)\Big] \Big).
\end{multline*}
Since $\bar v_e(e) = \bar 0$, $Y^{-1}_{\bar 0}$ is by definition the identity, hence from $\ax{AxEther^{STL}}$ follows that $Tr'_*(Ether(e)) \leftrightarrow Ether(e)$.
\end{proof}

\begin{cor}\label{cor_ce2}
\begin{equation*}
\sy{ClassicalKin^{STL}_{Full}}\vdash Tr'_*(\mathfrak{c_e}) = \mathfrak{c_e}. 
\end{equation*}
\end{cor}

\begin{proof}
This is a direct consequence of Lemma~\ref{tr_Ether2}.
\end{proof}

\begin{lem}\label{velocity-translation}\leavevmode
\begin{multline*}
\sy{ClassicalKin_{Full}}\vdash Tr_*\big(Ether(e) \AND \bar{v}^{STL}_k(e)=\bar v\big) \\ \leftrightarrow \BIG(Ether(e) \AND \bar{v}^{CK}_k(e)=\frac{\bar v}{\mathfrak{c_e}- |\bar v|}\BIG).
\end{multline*}
\end{lem}
\begin{proof}
By Lemma \ref{tr_Ether}, $Ether(e)$ translates into itself and hence
\begin{multline*}
Tr_*\big(Ether(e) \AND \bar{v}_{k}(e) = \bar{v} \big) \equiv  
Ether(e) \AND
\Bforall{e'}{Ether} \\
\vast(\Bexists{\vx,\vy}{Q^4} \fblock{
\bexists{\tau}{Q}\left(X_{\bar{v}_{e}(e')}\circ w^{CK}_{ke}\circ X^{-1}_{\bar{v}_k(e')}\right)(\vx)=(\tau,0,0,0) \\
\bexists{\tau}{Q}\left(X_{\bar{v}_{e}(e')}\circ w^{CK}_{ke}\circ X^{-1}_{\bar{v}_k(e')}\right)(\vy)=(\tau,0,0,0) \\
\vx \neq \vy}  \AND \\
\Bforall{\vx,\vy}{Q^4}
\fblock{
\bexists{\tau}{Q}\left(X_{\bar{v}_{e}(e')}\circ w^{CK}_{ke}\circ X^{-1}_{\bar{v}_k(e')}\right)(\vx)=(\tau,0,0,0) \\
\bexists{\tau}{Q}\left(X_{\bar{v}_{e}(e')}\circ w^{CK}_{ke}\circ X^{-1}_{\bar{v}_k(e')}\right)(\vy)=(\tau,0,0,0) \\
(y_1-x_1,y_2-x_2,y_3-x_3) = \bar{v} \cdot (y_0 - x_0)}\vast).
\end{multline*}
As $X_{\bar{v}_{e}(e)} = X_{\bar{0}} = Id$ and $\bar{v}_k(e') = \bar{v}_k(e)$ by Corollary \ref{thm-lightspeed}, we can simplify this to
\begin{multline*}
Tr_*\big(Ether(e) \AND \bar{v}_{k}(e) = \bar{v} \big) \equiv  
Ether(e) \AND \\
\Bexists{\vx,\vy}{Q^4}
\fblock{
\bexists{\tau}{Q}\left( w^{CK}_{ke}\circ X^{-1}_{\bar{v}_k(e)}\right)(\vx)=(\tau,0,0,0) \\
\bexists{\tau}{Q}\left( w^{CK}_{ke}\circ X^{-1}_{\bar{v}_k(e)}\right)(\vy)=(\tau,0,0,0) \\
\vx \neq \vy}  \AND \\
\Bforall{\vx,\vy}{Q^4}
\fblock{
\bexists{\tau}{Q}\left( w^{CK}_{ke}\circ X^{-1}_{\bar{v}_k(e)}\right)(\vx)=(\tau,0,0,0) \\
\bexists{\tau}{Q}\left( w^{CK}_{ke}\circ X^{-1}_{\bar{v}_k(e)}\right)(\vy)=(\tau,0,0,0) \\
(y_1-x_1,y_2-x_2,y_3-x_3) = \bar{v} \cdot (y_0 - x_0)}.
\end{multline*}

Since $X^{-1}_{\bar v_k(e)} = G_{-\bar v_k(e)} \circ G_{\frac{\mathfrak{c_e}\bar v_k(e)}{1+\left|\bar v_k(e)\right|}}$, and the images of $\vx$ and $\vy$ are on the time axis and hence by \ax{AxSelf} are on the worldline of $e$ according to $e$ (meaning that in Figure \ref{fig-ftl2stl-inv} $e$ coincides with $b$), $\bar {v}= \frac{\mathfrak{c_e}\bar v_k(e)}{1+\left|\bar v_k(e)\right|}$. From this it follows that the above is equivalent to \[Ether(e) \AND \bar{v}_k(e)=\frac{\bar v}{\mathfrak{c_e}-|\bar v|}. \tag*{\qedhere} \]
\end{proof}

\begin{lem}\label{velocity-translation2}\leavevmode
\begin{multline*}
\sy{ClassicalKin^{STL}_{Full}}\vdash Tr'_*\Big(Ether(e) \AND \bar{v}^{CK}_k(e)=\bar V\Big) \\ \leftrightarrow \BIG(Ether(e) \AND \bar{v}^{STL}_k(e)=\frac{\mathfrak{c_e}\cdot \bar V}{1+|\bar V|}\BIG).
\end{multline*}
\end{lem}
\begin{proof}
The proof is analogous to the proof of Lemma \ref{velocity-translation}.
\end{proof} 

\begin{thm}\label{thm-tr_FTL}
$Tr_*$ is an interpretation of \sy{ClassicalKin^{STL}_{Full}} in \sy{ClassicalKin_{Full}}, i.e.,
\begin{equation*}
\sy{ClassicalKin_{Full}}\vdash Tr_*(\varphi) \enskip\text{ if }\enskip \sy{ClassicalKin^{STL}_{Full}}\vdash \varphi.
\end{equation*}
\end{thm}

\begin{proof}

\begin{itemize}[leftmargin=*]
\item  $Tr_*(\ax{EField^{STL}})$ follows from  \sy{ClassicalKin_{Full}} because $Tr_*(\ax{EField^{STL}})\equiv\ax{EField^{CK}}\in\sy{ClassicalKin_{Full}}$.

\item $Tr_*(\ax{AxSelf^{STL}})$ is equivalent to the following formula:
\begin{multline*}
\bforall {k}  {IOb}
\bforall {t,x,y,z}{\Q}
\Big(\bforall{e}{Ether}\\
\bexists{\tau}{Q} \big[X_{\bar{v}_k(e)} \big(w^{CK}_{kk}\big(X^{-1}_{\bar{v}_k(e)}(\bar{x})\big)\big)= (\tau,0,0,0)\big]\leftrightarrow x=y=z=0\Big).
\end{multline*}
Since $w^{CK}_{kk}=Id$, $Tr_*(\ax{AxSelf^{STL}})$ is equivalent to:
\begin{multline*}
\bforall {k}  {IOb}
\bforall {t,x,y,z}{\Q}
\Big(\bforall{e}{Ether}\\
\bexists{\tau}{Q} \big[\bar{x}= (\tau,0,0,0)\big]\leftrightarrow x=y=z=0\Big),
\end{multline*}
which is a tautology since $(t,x,y,z)=(\tau,0,0,0)$ for some $\tau\in Q$ exactly if $x=y=z=0$.

\item By Lemma \ref{lem-tr_FTL}, $Tr_*(\ax{AxEv^{STL}})$ is equivalent to the following formula:
  \begin{multline*}
    \bforall {k,h}{ IOb} 
    \Bforall {\bar x}{\Q^4}
    \Bexists {\bar y}{\Q^4}
    \bforall{e}{Ether}\\
    \left[\left(X_{\bar{v}_h(e)}\circ w^{CK}_{kh}\circ X^{-1}_{\bar{v}_k(e)}\right)(\bar x)=\bar y\right].
  \end{multline*}

Since $X_{\bar{v}_{k'}(e)}\circ w^{CK}_{kk'}\circ X^{-1}_{\bar{v}_k(e)}$ only consists of Galilean transformations, it is also a Galilean transformation. Since $X_{\bar{v}_{k'}(e)}\circ w^{CK}_{kk'}\circ X^{-1}_{\bar{v}_k(e)}$ is a Galilean transformation, a $\bar y$ for which $\left(X_{\bar{v}_h(e)}\circ w^{CK}_{kh}\circ X^{-1}_{\bar{v}_k(e)}\right)(\bar x)=\bar y$ must exist because Galilean transformations are surjective.

\item By Lemma \ref{lem-tr_FTL}, $Tr_*(\ax{AxSymD^{STL}})$, is equivalent to the following formula:
 \begin{multline*}
    \bforall {k,k'}{ IOb}
    \Bforall {\vx,\vy,\vx',\vy'}{\Q^4}
    \bforall{e}{Ether}\\
    \left(\fblock{  time(\vx,\vy)= time(\vx',\vy')=0\\
      \left(X_{\bar{v}_{k'}(e)}\circ w^{CK}_{kk'}\circ X^{-1}_{\bar{v}_k(e)}\right)(\vx)=\vx' \\
      \left(X_{\bar{v}_{k'}(e)}\circ w^{CK}_{kk'}\circ X^{-1}_{\bar{v}_k(e)}\right)(\vy)=\vy'}\to  space(\vx,\vy)= space(\vx',\vy')\right).
  \end{multline*}

Since $\vx$ and $\vy$ are simultaneous before and after the Galilean transformation $X_{\bar{v}_{k'}(e)}\circ w^{CK}_{kk'}\circ X^{-1}_{\bar{v}_k(e)}$, their distance before and after the transformation must, by $\ax{AxSymD^{CK}}$ and the definition of Galilean transformations, be the same.

\item By the definition of $Tr_*$, $Tr_*(\ax{AxLine^{STL}})$ is equivalent to
\begin{multline*}
\bforall {k,h}{ IOb}
\Bforall {\vx,\vy,\vz}{Q^4} 
\bforall{e}{Ether}\\
\fblock{
\bexists{\tau}{Q}\left(X_{\bar{v}_{h}(e)}\circ w^{CK}_{kh}\circ X^{-1}_{\bar{v}_k(e)}\right)(\vx)=(\tau,0,0,0) \\
\bexists{\tau}{Q}\left(X_{\bar{v}_{h}(e)}\circ w^{CK}_{kh}\circ X^{-1}_{\bar{v}_k(e)}\right)(\vy)=(\tau,0,0,0) \\
\bexists{\tau}{Q}\left(X_{\bar{v}_{h}(e)}\circ w^{CK}_{kh}\circ X^{-1}_{\bar{v}_k(e)}\right)(\vz)=(\tau,0,0,0) \\
\bexists {a} {\Q}  \big[\vz-\vx=a(\vy-\vx) \lor
\vy-\vz=a(\vz-\vx)\big]}.
\end{multline*}
$X_{\bar{v}_{h}(e)}\circ w^{CK}_{kh}\circ X^{-1}_{\bar{v}_k(e)}$ is a linear bijection, so the images of $\vx$, $\vy$ and $\vz$ are on a line if and only if $\vx$, $\vy$ and $\vz$, are on a line. Since the images of $\vx$, $\vy$ and $\vz$ are all on the time axis, this translation follows.

\item By Lemma \ref{lem-tr_FTL}, $Tr_*(\ax{AxTriv^{STL}})$ is equivalent to
\[
\bforall{T}{\Triv} 
\bforall {h}  {IOb} 
\bexists {k}  {IOb}
\bforall{e}{Ether}
[X_{\bar v_{k}(e)}\circ w^{CK}_{hk} \circ X^{-1}_{\bar v_{h}(e)}=T].
\]

To prove $Tr_*(\ax{AxTriv^{STL}})$, we have to find an inertial observer $k$ for every trivial transformation $T$ and the inertial observer $h$ such that $X_{\bar v_{k}(e)}\circ w^{CK}_{hk} \circ X^{-1}_{\bar v_{h}(e)}=T$. Since $T$ is an affine transformation it is a translation by a vector, say $\bar z'$, after a linear transformation, say $T_0$.  By Lemma~\ref{triv-FTL}, if $w^{CK}_{hk}$ is the translation by the vector $\bar{z}=X^{-1}_{\bar{v}_{k}(e)}(\bar{z}')$ after $T_0$, then $X_{\bar v_{k}(e)}\circ w^{CK}_{hk} \circ X^{-1}_{\bar v_{h}(e)}=T$. By axiom $\ax{AxTriv^{CK}}$, there is an inertial observer $k$ such that $w^{CK}_{hk}$ is the translation by the vector $\bar{z}=X^{-1}_{\bar{v}_{k}(e)}(\bar{z}')$ after the trivial linear transformation $T_0$. Hence there is an inertial observer $k$ for which $X_{\bar v_{k}(e)}\circ w^{CK}_{hk} \circ X^{-1}_{\bar v_{h}(e)}=T$, and that is what we wanted to show. 

\item  By Lemma \ref{lem-tr_FTL}, $Tr_*(\ax{AxAbsTime^{STL}})$ is equivalent to
\begin{multline*}
\bforall {k,k'}{ IOb}
\Bforall {\vx,\vy,\vx',\vy'}{\Q^4} 
\Bforall {e}{Ether}\\
\left(\fblock{
 \left(X_{\bar{v}_{k'}(e)}\circ w^{CK}_{kk'}\circ X^{-1}_{\bar{v}_k(e)}\right)(\vx)=\vx' \\
 \left(X_{\bar{v}_{k'}(e)}\circ w^{CK}_{kk'}\circ X^{-1}_{\bar{v}_k(e)}\right)(\vy)=\vy'}\to
 time(\vx,\vy)= time(\vx',\vy') \right).
\end{multline*}
Since $X_{\bar{v}_{k'}(e)}\circ w^{CK}_{kk'}\circ X^{-1}_{\bar{v}_k(e)}$ is a Galilean transformation, it preserves absolute time. Therefore, the above follows from the definition of Galilean transformations and \ax{AxAbsTime^{CK}}

\item By Lemma~\ref{tr_Ether}, $Tr_*(\ax{AxEther^{STL}})$ is equivalent to $\ax{AxEther^{CK}}$.

\item  By the definition of $X_{\bar{v}_{k}(e)}$ and Lemma~\ref{tr_Ether}, $Tr_*(\ax{AxThExp^{STL}})$ is equivalent to
\begin{multline*}
\bexists {h}{B}  [ IOb(h)] \AND 
\bforall {e} {Ether}
\Bforall {\vx,\vy} {\Q^4}
\BIGG(space(\vx,\vy)<\mathfrak{c_e}\cdot time(\vx,\vy)
\\ \to \bexists {k}  {IOb} 
\fblock{\bexists{\tau}{Q}\left(X_{\bar{v}_{k}(e)}\circ w^{CK}_{ek}\circ X^{-1}_{\bar{v}_e(e)}\right)(\vx)=(\tau,0,0,0) \\
\bexists{\tau}{Q}\left(X_{\bar{v}_{k}(e)}\circ w^{CK}_{ek}\circ X^{-1}_{\bar{v}_e(e)}\right)(\vy)=(\tau,0,0,0) \\}
\BIGG),
\end{multline*}
which, since $X^{-1}_{\bar{v}_e(e)} = X^{-1}_{\bar{0}} = Id$, is equivalent to
\begin{multline*}
\bexists {h}{B}  [ IOb(h)] \AND 
\bforall {e} {Ether}
\Bforall {\vx,\vy} {\Q^4}
\BIGG(space(\vx,\vy)<\mathfrak{c_e}\cdot time(\vx,\vy)
\\ \to \bexists {k}  {IOb} 
\fblock{\bexists{\tau}{Q}\left(X_{\bar{v}_{k}(e)}\circ w^{CK}_{ek} \right)(\vx)=(\tau,0,0,0) \\
\bexists{\tau}{Q}\left(X_{\bar{v}_{k}(e)}\circ w^{CK}_{ek} \right)(\vy)=(\tau,0,0,0) \\}
\BIGG).
\end{multline*}
The first conjunct follows immediately as it is the same as the first conjunct of \ax{AxThExp^{CK}_+}.
If $x_0=y_0$ then $time(\vx,\vy)=0$, and $space(\vx,\vy)$ cannot be smaller than zero, then the antecedent $space(\vx,\vy)<\mathfrak{c_e}\cdot time(\vx,\vy)$ is false and the implication is true. 
Let us now assume that $x_0 \neq y_0$. Let us denote the velocity corresponding to the line containing both $\bar x$ and $\bar y$ by $\bar v$.  $|\bar v|<\mathfrak{c_e}$ as $space(\vx,\vy)<\mathfrak{c_e}\cdot time(\vx,\vy)$.  The worldview transformation $w^{CK}_{ek}$ is $G^{-1}_{\bar v_k(e)}\circ M$ for a translation $M$ by some spacetime vector. By the definition of $X_{\bar{v}_{k}(e)}$, if $\bar{v}_k(e)=\frac{\bar v}{\mathfrak{c_e}-\bar v}$, then $X_{\bar{v}_{k}(e)}\circ w^{CK}_{ek}$ is $G^{-1}_{\bar v}\circ M$ which always maps $\bar x$ and $\bar y$ to a line parallel to the time axis; and with an appropriately chosen translation $M$, $G^{-1}_{\bar v}\circ M$ maps $\bar x$ and $\bar y$ to the time axis. Therefore, by \ax{AxTriv^{CK}} and \ax{AxThExp^{CK}_+} there is an observer $k$ such that $X_{\bar{v}_{k}(e)}\circ w^{CK}_{ek}$ maps $\bar{x}$ and $\bar {y}$ to the time axis. So $Tr_*(\ax{AxThExp^{STL}})$ follows from \sy{ClassicalKin_{Full}}.

\item By Lemma~\ref{tr_Ether} and Corollary~\ref{cor_ce1}, $Tr_*(\ax{AxNoFTL})$ is equivalent to

\[\bforall {k} {IOb} \bexists {e}{Ether} \big[Tr_*\big(\speed_{e}(k) < \mathfrak{c_e}\big)\big].\]

Since $speed_e(k)=speed_k(e)$, this is equivalent to 
\[\bforall {k} {IOb} \bexists {e}{Ether}\big[ Tr_*\big(\speed_{k}(e) < \mathfrak{c_e}\big)\big].\]

By Lemma~\ref{velocity-translation}, this is equivalent to
\[\bforall {k} {IOb} \bforall {e}{Ether}\bexists{v}{Q^3}\fblock{speed_k(e)=\frac{v}{\mathfrak{c_e}-v}\\  v<\mathfrak{c_e}},\]

which is equivalent to 
  \[\bforall {k} {IOb} \bforall {e}{Ether}\left[\frac{\mathfrak{c_e}\cdot speed_k(e)}{1+speed_k(e)}<\mathfrak{c_e}\right],\]
which is true since $speed_k(e)< 1+speed_k(e)$ always holds because $\mathfrak{c_e}$ and $speed_k(e)$ are positive. 

\item  $Tr_*(\ax{AxNoAcc^{STL}})$ is equivalent to
\begin{multline*}
\bforall{k}{B}\bexists{\bar x}{Q^4}\bexists{b}{B} \bforall{e}{Ether} \\
\left(
\fblock{
    b\not\in IOb \to W^{CK}\big(k,b,X^{-1}_{\bar{v}_k(e)}(\bar{x})\big)\\
    b\in IOb \to \bexists{t}{Q} \big[w^{CK}_{kb}\big(X^{-1}_{\bar{v}_k(e)}(\bar{x})\big)= X_{\bar{v}_b(e)}^{-1}(t,0,0,0) \big] }
\rightarrow IOb(k) \right).
\end{multline*} 
If $k$ is an inertial observer, then the consequent of the implation is true, and hence the formula is true whatever the truth value of the antecedent. If $k$ is not an inertial observer, then by \ax{AxNoAcc^{CK}} it has no worldview and hence $v_k(e)$ is not defined, which makes the antecedent false.  \qedhere
\end{itemize}
\end{proof}

\begin{lem}\label{triv-STL}
Assume  \sy{{ClassicalKin}^{STL}_{Full}}.  Let $e$ be an ether observer, and let $k$ and $h$ be inertial. Assume that $w^{CK}_{hk}$ is a trivial transformation which is the translation by the vector $\bar{z}$ after the trivial linear transformation $T$. Then $Y_{\bar{v}_k(e)}\circ w^{SDL}_{hk}\circ Y^{-1}_{\bar{v}_h(e)}$ is the trivial transformation which is the translation by the vector $Y_{\bar{v}_k(e)}(\bar{z})$ after $T$.
\end{lem}
\begin{proof}
  The proof is analogous to the proof for Lemma \ref{triv-FTL}.
\end{proof}

\begin{thm}\label{thm-tr'_FTL}
$Tr'_*$ is an interpretation of \sy{ClassicalKin_{Full}} in \sy{ClassicalKin^{STL}_{Full}}, i.e.,
\begin{equation*}
\sy{ClassicalKin^{STL}_{Full}}\vdash Tr'_*(\varphi) \enskip\text{ if }\enskip \sy{ClassicalKin_{Full}}\vdash \varphi.
\end{equation*}
\end{thm}

\begin{proof}
The proofs for translations $Tr'_*(\ax{EField^{CK}})$, $Tr'_*(\ax{AxSelf^{CK}})$, $Tr'_*(\ax{AxSymD^{CK}})$, $Tr'_*(\ax{AxLine^{CK}})$, $Tr'_*(\ax{AxTriv^{CK}})$, $Tr'_*(\ax{AxAbsTime}^{CK})$ and $Tr'_*(\ax{AxNoAcc^{CK}})$ are analogous to the corresponding proofs in Theorem \ref{thm-tr_FTL}. For the full proof see \citep[p.69-71]{diss}.

\noindent
\begin{itemize}[leftmargin=*]

\item $Tr'_*(\ax{AxEv^{CK}})$ is equivalent to the following formula:
  \begin{multline*}
    \bforall {k,h}{ IOb} 
    \Bforall {\bar x}{\Q^4}
    \Bexists {\bar y}{\Q^4}
    \bforall{e}{Ether}\\
    \left[\left(Y_{\bar{v}_h(e)}\circ w^{STL}_{kh}\circ Y^{-1}_{\bar{v}_k(e)}\right)(\bar x)=\bar y\right],
  \end{multline*}
which follows from \sy{ClassicalKin^{STL}_{Full}} since $Y_{\bar{v}_{k'}(e)}\circ w^{STL}_{kk'}\circ Y^{-1}_{\bar{v}_k(e)}$ is a Galilean transformation.

\item By Lemma \ref{tr_Ether2}, $Tr'_*(\ax{AxEther}^{CK})$ is $\ax{AxEther}^{STL}$.

\item $Tr'_*(\ax{AxThExp_+})$  is equivalent to
\begin{multline*}
\bexists {h}{B} \big[ IOb(h)\big] \AND
\bforall {k}{ IOb} 
\Bforall {\vx,\vy}{\Q^4} 
\bforall {e} {Ether} 
\BIG(x_0\neq y_0 \to\\ \Bexists {h}{ IOb} 
\fblock{\bexists{\tau}{Q}\left(Y_{\bar{v}_{h}(e)}\circ w^{STL}_{kh}\circ Y^{-1}_{\bar{v}_k(e)}\right)(\vx)=(\tau,0,0,0) \\
\bexists{\tau}{Q}\left(Y_{\bar{v}_{h}(e)}\circ w^{STL}_{kh}\circ Y^{-1}_{\bar{v}_k(e)}\right)(\vy)=(\tau,0,0,0) \\}
\BIG).
\end{multline*}
Since the worldviews of any two inertial observers differ only by a Galilean transformation it is enough to prove $Tr'_*(\ax{AxThExp_+})$ when $k$ is an ether observer. So it is enough to prove the following: 
\begin{multline*}
\bexists {h}{B} \big[ IOb(h)\big] \AND
\bforall {k}{ IOb} 
\Bforall {\vx,\vy}{\Q^4}\\ 
\left(x_0\neq y_0 \to \Bexists {h}{ Ether} 
\fblock{\bexists{\tau}{Q}\left(Y_{\bar{v}_{h}(e)}\circ w^{STL}_{eh}\right)(\vx)=(\tau,0,0,0) \\
\bexists{\tau}{Q}\left(Y_{\bar{v}_{h}(e)}\circ w^{STL}_{eh}\right)(\vy)=(\tau,0,0,0) \\}
\right).
\end{multline*}
The first conjunct follows immediately as it is the same as the first conjunct of \ax{AxThExp^{STL}}.
If $x_0 = y_0$ the above statement is true as the antecedent of the implication is false. If $x_0\neq y_0$, then let us denote the velocity corresponding to the line containing both $\bar x$ and $\bar y$ by $\bar V$.  The worldview transformation $w^{STL}_{ek}$ is $G^{-1}_{\bar v_k(e)}\circ M$ for a translation $M$ by some spacetime vector. By the definition of $Y_{\bar{v}_{k}(e)}$, if $\bar{v}_k(e)=\frac{\mathfrak{c_e}\cdot \bar V}{1+\bar V}$, then $Y_{\bar{v}_{k}(e)}\circ w^{STL}_{ek}$ is $G^{-1}_{\bar V}\circ M$, which always maps $\bar x$ and $\bar y$ to a line parallel to the time axis; and with an appropriately chosen translation $M$, $G^{-1}_{\bar V}\circ M$ maps $\bar x$ and $\bar y$ to the time axis. Since $\bar{v}_k(e)=\frac{\mathfrak{c_e}\cdot \bar V}{1+\bar V}<\mathfrak{c_e}$, by \ax{AxThExp^{STL}} and \ax{AxTriv^{STL}} there is an observer $k$ such that $Y_{\bar{v}_{k}(e)}\circ w^{STL}_{ek}$ maps $\bar{x}$ and $\bar {y}$ to the time axis. So $Tr_*(\ax{AxThExp^{STL}})$ follows from \sy{ClassicalKin^{STL}_{Full}}.\qedhere
\end{itemize}
\end{proof}

\begin{thm}\label{thm-defeqFTL}
$Tr_{*}$ is a definitional equivalence between theories \sy{ClassicalKin_{Full}} and \sy{ClassicalKin_{Full}^{STL}}.
\end{thm}

\begin{proof}
We only need to prove that the inverse translations of the translated statements are logical equivalent to the original statements since $Tr_*$ and $Tr'_*$ are interpretations by Theorem~\ref{thm-tr_FTL} and Theorem~\ref{thm-tr'_FTL}. Since $Tr_*$ and $Tr'_*$ are identical on every concept but the worldview relation, we only have to prove the following two statements: 
\begin{align*}
\sy{ClassicalKin_{Full}^{STL}}\vdash &Tr'_*\big(Tr_*[W^{STL}(k,b,\bar x)]\big)\equiv W^{STL}(k,b,\bar x)\text{ and }\\
\sy{ClassicalKin_{Full}} \vdash &Tr_*\big(Tr'_*[W^{CK}(k,b,\bar x)]\big)\equiv W^{CK}(k,b,\bar x).
\end{align*}

The back and forth translation of $W^{STL}$ is the following:
\begin{multline*}
Tr'_*\big(Tr_*[W^{STL}(k,b,\bar x)]\big) \equiv  Tr'_*\BIGG( \bforall{e}{Ether} \bexists{\bar{V}}{Q^3}\\
  \fblock{\bar V=\bar{v}_b(e)\\
    b\not\in IOb \to W^{CK}\big(k,b,X^{-1}_{\bar{V}}(\bar{x})\big)\\
    b\in IOb \to \bexists{t}{Q} \big[\big(X_{\bar{V}}\circ w^{CK}_{kb}\circ X^{-1}_{\bar{V}}\big)(\bar{x})= (t,0,0,0)\big]
  }\BIGG),
\end{multline*}
which is equivalent to
\begin{multline*}
  \bforall{e}{Ether} \bexists{\bar{v},\bar{V}}{Q^3}\\
  \fblock{\bar v=\bar{v}_b(e)\land Tr'_*\big(\bar v_k(e)=\bar V\big)\\
    b\not\in IOb \to W^{STL}\Big(k,b,Y^{-1}_{\bar{v}}(X^{-1}_{\bar{V}}\big(\bar{x})\big)\Big)\\
    b\in IOb \to \bexists{t}{Q} \big[\big(X_{\bar{V}}\circ Y_{\bar{v}}\circ w^{STL}_{kb}\circ Y^{-1}_{\bar{v}}\circ X^{-1}_{\bar{V}}\big)(\bar{x})= (t,0,0,0)\big]
  }.
\end{multline*}
By Lemma \ref{lem-xy}, this is equivalent to
\[
\bforall{e}{Ether}
  \fblock{
    b\not\in IOb \to W^{STL}\big(k,b,\bar{x}\big)\\
    b\in IOb \to \bexists{t}{Q} [ w^{STL}_{kb}(\bar{x})= (t,0,0,0)]}.
\]

This is clearly equivalent to $W^{STL}(k,b,\bar x)$ if $b$ is not an inertial observer. If $b$ is an inertial observer, then we need that $\bexists{t}{Q} [ w^{STL}_{kb}(\bar{x})= (t,0,0,0)]$ is equivalent to $W^{STL}(k,b,\bar x)$, which holds because of $\ax{AxSelf}$ in \sy{ClassicalKin_{Full}^{STL}}  and the definition of the worldview transformation.

Proving   $Tr_*\big(Tr'_*[W^{CK}(k,b,\bar x)]\big)\equiv W^{CK}(k,b,\bar x)$ 
from \sy{ClassicalKin_{Full}}  is completely analogous. 
\end{proof}

\begin{cor}\label{cor-strong}
$Tr_{*} \circ Tr_+$ is a definitional equivalence between theories \sy{SpecRel_{Full}^{e}} and \sy{ClassicalKin_{Full}}.
\end{cor}
\begin{proof}
By transitivity of definitional equivalence (Theorem \ref{thm-eqrel}), Theorem \ref{thm-defeq} and Theorem \ref{thm-defeqFTL}.
\end{proof}

\section{Concluding remarks}

With the clarity of mathematical logic, we have achieved the following new results: 
constructing an interpretation of special relativity into classical kinematics using Poincar{\'e}--Einstein synchronisation;
turning this interpretation into a definitional equivalence by extending special relativity with an ether concept and restricting classical kinematics to slower-than-light (STL) observers;
proving a definitional equivalence between classical kinematics and classical kinematics restricted to STL observers;
concluding by transitivity of definitional equivalence the main result that classical kinematics is definitionally equivalent to special relativity extended with an ether concept.

To get special relativity theory, ether is the only concept that has to be removed from classical kinematics. However, removing ether from classical kinematics also leads to a change of the notions of space and time, which in the framework of this paper is handled by the translation functions.

It is philosophically interesting to note that our results are not identical to what we might expect from physical intuition: having or not having an upper speed limit appears to be an important physical distinction, but by Theorem~\ref{thm-defeqFTL} we can establish a definitional equivalence between theories with and without this speed limit. While we based the translation $Tr$ between \ax{SpecRel_{Full}} and \ax{ClassicalKin_{Full}} on the physical intuition of using a Poincar{\'e}--Einstein synchronisation, there was no such physical ground for the translation between \ax{ClassicalKin^{STL}_{Full}} and \ax{ClassicalKin_{Full}}. 


\section{Appendix: Simplification of translated formulas} \label{appendix}

The simplification tools in this appendix are (up to Lemma \ref{lemma-merge}) consequences of all ether observers being at rest relative to each other in \sy{{ClassicalKin}_{Full}}.

Let $b$, $k_1$, \ldots, $k_n$ be variables of sort $\B$. We say that formula $\varphi$ is \emph{ether-observer-independent} in variable $b$ provided that $k_1$, \ldots, $k_n$ are inertial observers if the truth or falsehood of $\varphi$ does not depend on to which ether observer we evaluated $b$, assuming $k_1$, \ldots, $k_n$ are evaluated to inertial observers, that is: 
\begin{multline*}
EOI^{k_1,\ldots,k_n}_{b}[\varphi] \defiff \\ \sy{{ClassicalKin}_{Full}}\vdash\bforall{k_1, \ldots, k_n}{ IOb}\bforall{e, e'}{\Ether} 
[\varphi(e / b) \leftrightarrow \varphi(e' / b)].
\end{multline*}
where $\varphi(k/b)$ means that $b$ gets replaced by $k$ in all free occurrences of $b$ in 
$\varphi$.

Let us note that the fewer variables there are in the upper index of $EOI^{k_1,\ldots,k_n}_b(\varphi)$, the stronger a statement we have about the ether-observer independence of 
$\varphi$. The strongest statement about the ether-observer independence of $\varphi$ is $EOI_b(\varphi)$.

\begin{lem}\label{lem-EOI-nob} 
If $b$ is not a free variable of $\varphi$, then $EOI_{b}[\varphi]$ holds.
\end{lem}
\begin{proof}
In the definition of EOI, there is nothing to replace in formula $\varphi$, so both parts $\varphi(e / b)$ and $\varphi(e' / b)$ of the equivalence remain the same.
\end{proof}

From Lemma \ref{lem-EOI-nob}, the following immediatly follows:

\begin{cor}\label{cor-EOI-equiv} Assume \sy{{ClassicalKin}_{Full}}. Let $\alpha$ and $\beta$ be quantity terms and let $k$ and $h$ be body variables. Then atomic formulas $\alpha = \beta$ and $\alpha<\beta$ are ether-observer-independent, i.e., for any body variable $b$ we have:
$EOI_{b}[\alpha = \beta]$,
$EOI_{b}[\alpha < \beta]$,
$EOI_{b}[k = k]$, and
$EOI_{b}[k = h]$.
\end{cor}
\noindent
Let us note here that $EOI_{b}[b = b]$ holds and that $EOI_{b}[b = h]$ does not hold.

\begin{lem}\label{lemma-EOI-IOb-Ph} Assume \sy{{ClassicalKin}_{Full}}. Let $k$ be a body variable. Then $k$ being an inertial-observer or a light signal are ether-observer-independent, i.e., for any body variable $b$ we have:
$EOI_{b}[ IOb(b)]$,
$EOI_{b}[ IOb(k)]$,
$EOI_{b}[\Ph(b)]$, and
$EOI_{b}[\Ph(k)]$.
\end{lem}
\begin{proof}
For proving $EOI_{b}[ IOb(b)]$, note that all ether observers are inertial observers, it by definition means that \[\sy{{ClassicalKin}_{Full}}\vdash\bforall{k_1, \ldots, k_n}{ IOb}\bforall{e, e'}{\Ether} 
[ IOb(e / b) \leftrightarrow  IOb(e' / b)],\] which is true. 
$EOI_{b}[Ph(b)]$ is true because ether observers go slower-than-light, therefore both sides of the equivalence are false, which makes the equivalence true.
$EOI_{b}[\Ph(b)]$ and $EOI_{b}[\Ph(k)]$ follow from Lemma \ref{lem-EOI-nob}.
\end{proof}

\begin{cor}\label{cor-ei-s} Assuming \sy{\mathsf{ClassicalKin}_{Full}}, the speed of an inertial observer is the same according to every ether observers, i.e.,
$EOI^{k}_{b}[\speed_{b}(k)=v].$
\end{cor}
\noindent
Let us note that $EOI^{k}_{b}[\bar v_{b} (k) =\bar v]$ does not hold, because ether observers can be rotated relative to each other; hence the direction of $\bar v_b(k)$ depends on which ether observer variable $b$ is evaluated to.

\begin{cor}\label{cor-ei-v} Assuming \sy{\mathsf{ClassicalKin}_{Full}}, the velocity of all ether observers is the same according to every inertial observers, i.e.,
$EOI^{k}_{b}[\bar v_k (b) =\bar {v}].$
\end{cor}

\begin{cor}\label{cor-ei-c} Assuming \sy{\mathsf{ClassicalKin}_{Full}}, an inertial observer being slower than 
light is an ether-observer-independent statement, i.e.,
$EOI^{k}_{b}[\speed_{b}(k) < \mathfrak{c_e}].$
\end{cor}

The next three lemmas are being used to simplify the formulas  that the translation $Tr$ provides us.

The following rules can be used to show the ether independence of complex formulas:

\begin{lem}\label{lemma-EOI-formulas} 

\noindent 
\begin{enumerate}
\item \label{EOI-neg} $EOI^{k_1, \ldots, k_n}_{b}[\varphi]$ implies $EOI^{k_1, \ldots, k_n}_{b}[\neg\varphi]$.
\item \label{EOI-conn} If $*$ is a logical connective, then from $EOI^{k_1, \ldots, k_n}_{b}[\varphi]$ and $EOI^{h_1,\ldots,h_m}_{b}[\psi]$ follows $EOI^{k_1, \ldots, k_n,h_1,\ldots,h_m}_{b}[\varphi * \psi]$.
\item \label{EOI-exists} $EOI^{k_1, \ldots, k_n}_{b}[\varphi]$ implies $EOI^{k_1, \ldots, k_n}_{b}[\bexists {x}{Q}(\varphi)]$ and  $EOI^{k_1, \ldots, k_n}_{b}[\bexists {h}{B}(\varphi)]$.
\item \label{EOI-forall} $EOI^{k_1, \ldots, k_n}_{b}[\varphi]$ implies $EOI^{k_1, \ldots, k_n}_{b}[\bforall {x}{Q}(\varphi)]$ and  $EOI^{k_1, \ldots, k_n}_{b}[\bforall {h}{B}(\varphi)]$.
\end{enumerate}
\end{lem}

\begin{proof}

 If $EOI^{k_1,\ldots,k_n}_{b}[\varphi]$ holds, then by definition \[\sy{{ClassicalKin}_{Full}}\vdash\bforall{k_1, \ldots, k_n}{ IOb}\bforall{e, e'}{\Ether}[\varphi(e / b) \leftrightarrow \varphi(e' / b)],\] which is equivalent to \[\sy{{ClassicalKin}_{Full}}\vdash\bforall{k_1, \ldots, k_n}{ IOb}\bforall{e, e'}{\Ether}[\neg \varphi(e / b) \leftrightarrow \neg \varphi(e' / b)]\] because $A \leftrightarrow B$ is equivalent to $\neg A \leftrightarrow \neg B$ and therefore $EOI^{k_1,\ldots,k_n}_{b}[\neg \varphi]$ also holds. 

Let us now prove item \ref{EOI-conn}. Since all logical connectives can be constructed from negation and conjunction, we only need to prove the property of item \ref{EOI-exists} for conjunction, as we have already proven it for negation in item \ref{EOI-neg} above. 

To do so, we have to prove that in any model of \sy{{ClassicalKin}_{Full}}, $(\varphi\AND\psi)(e/b)$ holds if and only if $(\varphi\AND\psi)(e'/b)$ holds provided that $e$ and $e'$ are ether observers and $k_1, \ldots, k_n, h_1, \ldots, h_m$ are inertial observers. Formula $(\varphi\AND\psi)(e/b)$ holds exactly if both $\varphi(e/b)$ and $\psi (e/b)$ hold. Similarly, $(\varphi\AND\psi)(e'/b)$ holds exactly if both $\varphi(e'/b)$ and $\psi (e'/b)$ hold.

By $EOI^{k_1 \ldots k_n}_{b}[\varphi]$, formula $\varphi(e/b)$ holds exactly if $\varphi(e'/b)$ holds provided that $e$ and $e'$ are ether observers and $k_1, \ldots, k_n,$ are inertial observers. By  $EOI^{h_1,\ldots,h_m}_{b}[\psi]$, formula $\psi (e/b)$ holds exactly if $\psi (e'/b)$ holds provided that $e$ and $e'$ are ether observers and $h_1, \ldots, h_m$ are inertial observers. Therefore, $(\varphi\AND\psi)(e/b)$ holds exactly if $(\varphi\AND\psi)(e'/b)$ holds provided that $e$ and $e'$ are ether observers and $k_1, \ldots, k_n$, $h_1, \ldots, h_m$ are inertial observers; and this is what we wanted to prove.

Item \ref{EOI-exists} is true because from $A \leftrightarrow B$ it follows that $(\exists u A) \leftrightarrow (\exists u B)$.

Item \ref{EOI-forall} follows from items \ref{EOI-neg} and \ref{EOI-exists}, since the universal quantifier can be composed of the negation and the existential quantifier: $\forall u(\varphi)$ is equivalent to $\neg\exists u(\neg\varphi)$.
\end{proof}

\begin{lem}\label{lemma-EOI-Rad} Assume \sy{{ClassicalKin}_{Full}}. Let $k$ be a body variable let and $\bar \alpha$ and $\bar \beta$ quantity terms. Then for any body variable $b$ we have:
$EOI^k_{b}[Rad^{-1}_{\bar{v}_k(b)}(\bar \alpha)=\bar \beta]$ and
$EOI^k_{b}[W(k,h,Rad^{-1}_{\bar{v}_k(b)}(\bar \alpha))]$.
\end{lem}
\begin{proof}
To prove $EOI^k_{b}[Rad^{-1}_{\bar{v}_k(b)}(\bar \alpha)=\bar \beta]$, we have to prove that \[\Bforall{\bar v}{Q^3}[\bar{v}_{k}(b) = \bar v \to Rad^{-1}_{\bar v}(\bar \alpha) = \bar \beta]\] is ether-observer-independent in variable $b$ provided that $k$ is an inertial observer. $EOI^k_b[\bar{v}_{k}(b) = \bar v]$ holds because of Corollary \ref{cor-ei-v}. $EOI_b[Rad^{-1}_{\bar v}(\bar \alpha) = \bar \beta]$ holds because of Lemma \ref{lem-EOI-nob}. From items \ref{EOI-conn} and \ref{EOI-forall} of Lemma \ref{lemma-EOI-formulas} follows what we want to prove.

To prove $EOI^k_{b}[W(k,h,Rad^{-1}_{\bar{v}_k(b)}(\bar \alpha))]$, we have to prove that $W(k,h,\bar \beta) \AND Rad^{-1}_{\bar{v}_k(b)}(\bar \alpha)=\bar \beta$ is ether-observer-independent in variable $b$ provided that $k$ is an inertial observer. $EOI^k_b[Rad^{-1}_{\bar{v}_k(b)}(\bar \alpha)=\bar \beta]$ holds because $EOI^k_{b}[Rad^{-1}_{\bar{v}_k(b)}(\bar \alpha)=\bar \beta]$ by the first part of this lemma. $EOI_b[W(k,h,\bar \beta)]$ holds because of Lemma \ref{lem-EOI-nob}. From item \ref{EOI-conn} of Lemma \ref{lemma-EOI-formulas} follows what we want to prove.
\end{proof}

\begin{lem}\label{lemma-merge}  Assume \sy{{ClassicalKin}_{Full}} and that $EOI^{k_1, \ldots, k_n}_{e}[\varphi]$ and $EOI^{h_{1}, \ldots, h_m}_{e}[\psi]$ hold. For every logical connective $*$,
\begin{equation*}
\bforall{e}{\Ether}(\varphi) * \bforall{e}{\Ether}(\psi)
\end{equation*}
is equivalent to
\begin{equation*}
\bforall{e}{\Ether}(\varphi * \psi)
\end{equation*}
provided that $k_1, \ldots, k_n, h_{1}, \ldots, h_m$ are inertial observers.
\end{lem}

\begin{proof}
From Lemma \ref{lemma-EOI-formulas} Item \ref{EOI-conn}, we know that $(\varphi * \psi)$ is ether-observer-independent in variable $e$ provided that $k_1, \ldots, k_n, h_{1}, \ldots, h_m$ are inertial observers. Because $\varphi$, $\psi$ and $\varphi * \psi$ are ether-observer-independent in variable $e$, it does not matter which ether observer we fill in there. Therefore, the two formulas are equivalent.
\end{proof}

As an example on how we use this, the mechanographical translation of $\ax{AxSelf^{SR}}$ is
\begin{multline*}
\forall k  
\Big( IOb(k) \AND \big(\bforall {e} {\Ether} \big( \speed_{e}(k) < \mathfrak{c_e} \big) \\\to
\Bforall {\bar{y}}{\Q^4} \bforall {e} {\Ether}  \big[ W\big((k,k,Rad^{-1}_{\bar{v}_k(e)}(\bar{y})\big) \leftrightarrow y_1=y_2=y_3=0 \big] \Big),
\end{multline*}
which by Lemma \ref{lemma-merge} and the bounded quantifiers notation can be simplified to
\begin{multline*}
\bforall {k} { IOb}  \bforall {e} {\Ether}  
\Big( \speed_{e}(k) < \mathfrak{c_e} \\\to
\Bforall {\bar{y}}{\Q^4} \big[ W\big((k,k,Rad^{-1}_{\bar{v}_k(e)}(\bar{y})\big) \leftrightarrow y_1=y_2=y_3=0 \big] \Big).
\end{multline*}

Since the observers in the set $E$ of ``primitive ether'' observers as introduced in the section on definitional equivalence are, by definition, only differing from each other by a trivial transformation, they are at rest relative to each other. So we can introduce a concept of ``primitive ether independent observers'' (PEIO) with the same properties as ``ether independent observers'' we have proven above. The proofs of these properties are analogous to the proofs for EIO in this appendix, which we for brevity will not repeat.

\subsection*{Acknowledgements}
We are grateful to Hajnal Andr{\'e}ka, Marcoen Cabbolet, Steffen Ducheyne, P{\'e}ter Fekete, Michele Friend, Istv{\'a}n N{\'e}meti, Sonja Smets, Laszl{\'o} E. Szab{\'o}, Jean Paul Van Bendegem and Bart Van Kerkhove for enjoyable discussions and helpful suggestions while writing this paper, as well as to the two anonymous reviewers for valuable feedback.



\bibliographystyle{agsm}
\bibliography{LogRel12017} 




\begin{flushright}
KOEN LEFEVER\\
Vrije Universiteit Brussel\\
\href{mailto:koen.lefever@vub.ac.be}{koen.lefever@vub.ac.be}\\
\url{http://www.vub.ac.be/CLWF/members/koen/index.shtml}

\vspace{.7cm}

GERGELY SZ{\' E}KELY\\
MTA Alfr{\' e}d R{\' e}nyi Institute for Mathematics\\
\href{mailto:szekely.gergely@renyi.mta.hu}{szekely.gergely@renyi.mta.hu}\\
\url{http://www.renyi.hu/~turms/}
\end{flushright}

\end{document}